\documentclass[10pt]{article}
\usepackage[utf8]{inputenc}
\usepackage[T1]{fontenc}
\usepackage[a4paper, margin=2cm]{geometry}
\usepackage{times}
\usepackage{microtype}
\usepackage{amsmath,amssymb,amsfonts,amsthm}
\usepackage{mathtools}
\usepackage{xspace}
\usepackage{tikz}
\usepackage{fleqn}
\usepackage{multirow}
\usepackage[export]{adjustbox}
\usepackage{thmtools, thm-restate}

\newif\iflong
\newif\ifshort

\longtrue
 
\iflong
\else
\shorttrue
\fi

\usetikzlibrary{patterns}
\usetikzlibrary{arrows}
\usetikzlibrary{arrows.meta}
\usetikzlibrary{decorations.pathmorphing}
\usetikzlibrary{decorations.markings}
\usetikzlibrary{calc}
\usetikzlibrary{positioning}

\tikzset{
  circ/.style = {circle,draw,fill,inner sep=1.3pt},
  mcirc/.style = {circle,draw,fill,inner sep=1pt},
  circR/.style = {circle,draw=red,fill=red,text=red,inner sep=1.3pt},
  circG/.style = {circle,draw=green,fill=green,text=green,inner sep=1.3pt},
  circB/.style = {circle,draw=blue,fill=blue,text=blue,inner sep=1.3pt},
  circb/.style = {circle,draw=blue,fill=blue,text=blue,inner sep=1.1pt},
  circr/.style = {circle,draw=red,fill=red,inner sep=1pt},
  scirc/.style = {circle,draw,fill,inner sep=.8pt},
  invisible/.style = {draw=none,inner sep=0pt,font=\tiny},
  nonedge/.style={decorate,decoration={snake,amplitude=.3mm,segment length=1mm},draw}
}

\usepackage{graphicx}
\usepackage{subcaption}
\usepackage{enumerate}
\usepackage{color,soul} %
\usepackage{todonotes}

\theoremstyle{plain}
  \newtheorem{theorem}{Theorem}
  
  \newtheorem{observation}[theorem]{Observation}
  \newtheorem{lemma}[theorem]{Lemma}
  \newtheorem{proposition}[theorem]{Proposition}
\newcounter{specialtheorem}
  \newtheorem{claim}{Claim}[specialtheorem]
  \newtheorem*{theorem*}{Theorem}
  \newtheorem*{corollary*}{Corollary}
  \newtheorem*{lemma*}{Lemma}
  \newtheorem*{proposition*}{Proposition}
  \newtheorem*{claim*}{Claim}

\theoremstyle{definition}
  \newtheorem{definition}[theorem]{Definition}

  \newtheorem*{definition*}{Definition}
  \newtheorem*{example*}{Example}
  \newtheorem*{question*}{Question}
  \newtheorem*{conjecture*}{Conjecture}

\DeclarePairedDelimiter\abs{\lvert}{\rvert}
\DeclarePairedDelimiter\norm{\lVert}{\rVert}

\DeclarePairedDelimiter\close{\langle}{\rangle}

\mathchardef\mhyphen="2D

\newcommand*{\rationals}{\mathbb{Q}}
\newcommand*{\naturals}{\mathbb{N}}

\newcommand{\cc}[1]{{\mbox{\textnormal{\textsf{#1}}}}\xspace}  %

\renewcommand{\P}{\cc{P}}
\newcommand{\NP}{\cc{NP}}
\newcommand{\FPT}{\cc{FPT}}

\newcommand{\Weft}{{\cc{W}}}
\newcommand{\W}[1]{{\Weft}{{[#1]}}}

\newcommand*{\mincsp}[1]{\textsc{MinCSP}\ensuremath{(#1)}}
\newcommand*{\csp}[1]{\textsc{CSP}\ensuremath{(#1)}}

\def \simcut    {\textsc{Sim-Cut}\xspace}
\def \simsep    {\textsc{Sim-Separator}\xspace}

\def \mbicl     {\textsc{MC-BiClique}\xspace}
\def \simdfas   {\textsc{Sim-DFAS}\xspace}

\def \subdfas   {\textsc{Sub-DFAS}\xspace}
\def \lsmfas    {\textsc{LS-MFAS}\xspace}

\def \probbc    {\textsc{Bundled Cut}\xspace}

\DeclareMathOperator{\cost}{cost}
\def \bundles {\mathcal{B}}
\def \ordering {\ensuremath{\sigma}}

\def \A   {\ensuremath{\cal A}\xspace}
\def \p   {{\sf p}\xspace}
\def \o   {{\sf o}\xspace}
\def \e   {\ensuremath{\equiv}\xspace}
\def \m   {{\sf m}\xspace}
\def \d   {{\sf d}\xspace}
\def \s   {{\sf s}\xspace}
\def \f   {{\sf f}\xspace}

\def \oi  {{\sf oi}\xspace}
\def \pi  {{\sf pi}\xspace}
\def \di  {{\sf di}\xspace}
\def \si  {{\sf si}\xspace}
\def \fii {{\sf fi}\xspace}

\def \rr  {{\sf r}\xspace}

\newcommand{\spii}{\ \ }

\newcommand{\spiv}{\ \ \ \ }
\newcommand{\spv}{\ \ \ \ \ }

\newcommand{\stexttt}[1]{{\tiny\tt #1}}

\usepackage{boxedminipage}
\newcommand{\pbDef}[3]{%
  \noindent
  \begin{center}
  \begin{boxedminipage}{0.98 \columnwidth}
  {\sc #1}\\[5pt]
  \begin{tabular}{l p{0.82 \columnwidth}}
  {\sc Instance}: & #2\\
  {\sc Question}: & #3
  \end{tabular}
  \end{boxedminipage}
  \end{center}
}

\newcommand{\pbDefP}[4]{%
  \noindent
  \begin{center}
  \begin{boxedminipage}{0.98 \columnwidth}
  {\sc #1}\\[5pt]
  \begin{tabular}{l p{0.82 \columnwidth}}
  {\sc Instance}: & #2\\
  {\sc Parameter}: & #3\\
  {\sc Question}: & #4
  \end{tabular}
  \end{boxedminipage}
  \end{center}
}

\newcommand{\inn}{\texttt{in}}
\newcommand{\out}{\texttt{out}}

\newcommand{\sss}{\texttt{s}}
\newcommand{\ttt}{\texttt{t}}
\newcommand{\Ac}{\textsf{\bf A}}
\newcommand{\B}{\textsf{\bf B}}
\newcommand{\D}{\textsf{\bf D}}
\newcommand{\PP}{\textsf{\bf P}}
\newcommand{\pth}{\rho}

\newcommand{\Whard}{\textsf{W}[1]-\textsf{hard}}
\newcommand{\Wh}{\textsf{W}[1]-\textsf{h}}

\newcommand{\yes}{\textsc{Yes}}

\newcommand{\III}{\mathcal{I}}
\newcommand{\CCC}{\mathcal{C}}

\title{Parameterized Complexity Classification for Interval Constraints\thanks{
  The second and the fourth authors were supported by
  the Wallenberg AI, Autonomous Systems and Software Program (WASP) funded
  by the Knut and Alice Wallenberg Foundation.
  In addition, the second author was partially supported 
  by the Swedish Research Council (VR)
  under grant 2021-04371. The third author acknowledges support from the
  Engineering and Physical Sciences Research Council (EPSRC, project EP/V00252X/1).
  The research of the fifth author is a part of a project that has received funding from the European Research Council (ERC)
under the European Union's Horizon 2020 research and innovation programme
Grant Agreement 714704.
  In addition, the fifth author is also part of BARC, supported by the VILLUM Foundation grant 16582.
  }
}

\author{
    Konrad K. Dabrowski\thanks{School of Computing, Newcastle University, UK, \texttt{konrad.dabrowski@newcastle.ac.uk}} \and
    Peter Jonsson\thanks{Department of Computer and Information Science, Link{\"o}ping University, Sweden, \texttt{peter.jonsson@liu.se}}  \and
    Sebastian Ordyniak\thanks{School of Computing, University of Leeds, UK, \texttt{sordyniak@gmail.com}} \and
    George Osipov\thanks{Department of Computer and Information Science, Link{\"o}ping University, Sweden, \texttt{george.osipov@pm.me}} \and
    Marcin Pilipczuk\thanks{Institute of Informatics, University of Warsaw, Poland and ITU Copenhagen, Denmark, \texttt{malcin@mimuw.edu.pl}} \and
    Roohani Sharma\thanks{Max Planck Institute for Informatics, Saarland Informatics Campus, Saarbr\"ucken, Germany, \texttt{rsharma@mpi-inf.mpg.de}}
}
\iflong
\date{\today}
\fi

\begin{document}

\begin{titlepage}
\def\thepage{}
\thispagestyle{empty}
\maketitle

\begin{abstract}
Constraint satisfaction problems form
a nicely behaved class of problems that lends itself
to complexity classification results.
From the point of view of parameterized complexity, a natural task is to classify the
parameterized complexity of \textsc{MinCSP} problems parameterized
by the number of unsatisfied constraints. In other words, 
we ask whether we can delete at most $k$ constraints, where $k$ is the parameter, 
to get a satisfiable instance.
In this work, we take a step towards classifying the parameterized complexity for an important
infinite-domain CSP: {\em Allen's interval algebra} (IA).
This CSP has closed intervals with rational endpoints as domain values and employs a set $\A$ of 13 basic comparison relations such as ``precedes'' or ``during''
for relating intervals.
IA is a highly influential and well-studied
formalism within AI and qualitative reasoning that has numerous applications in, for instance,
planning, natural language processing and molecular biology. 
We provide an \FPT vs. \W{1}-hard dichotomy for \textsc{MinCSP}$(\Gamma)$ for all 
$\Gamma \subseteq \A$.
IA is sometimes extended with unions of the relations in $\A$
or first-order definable relations over $\A$, but extending our results to these cases would require first solving the parameterized complexity of \textsc{Directed Symmetric Multicut}, which is a notorious open problem.
Already in this limited setting, we uncover connections to new variants 
of graph cut and separation problems. 
This includes hardness proofs for simultaneous cuts or feedback arc set
problems in directed graphs, as well as new tractable cases with algorithms based
on the recently introduced \emph{flow augmentation} technique.
Given the intractability of \textsc{MinCSP}$(\A)$ in general, we then consider
(parameterized) approximation algorithms. We first show that \textsc{MinCSP}$(\A)$ cannot be
polynomial-time approximated within any constant factor
and continue by presenting a factor-$2$
fpt-approximation algorithm. Once again, this algorithm has its roots in flow augmentation.
\end{abstract}

\end{titlepage}

\newpage

\tableofcontents

\newpage

\section{Introduction}
{\bf Background.}
The {\em constraint satisfaction problem} over a constraint language $\Gamma$ (CSP$(\Gamma)$) is
the problem of deciding whether there is a variable assignment which satisfies a set of constraints,
where each constraint is constructed from a relation in $\Gamma$.
CSPs
over different constraint languages
is a nicely behaved class of problems that lends itself
to complexity classification results.
Such results are an important testbed for studying the power of algorithmic techniques and proving their limitations---
\iflong
the inability to obtain a full classification indicates 
that new algorithmic tools or new lower bound techniques are required.
\fi
\ifshort
the inability to obtain a full classification indicates 
that new tools are required.
\fi
A prime example is the dichotomy theorem for finite-domain CSPs that
was conjectured by Feder and Vardi~\cite{Feder:Vardi:sicomp98} and independently proved by Bulatov~\cite{Bulatov:focs2017}
and Zhuk~\cite{Zhuk:jacm2020}. Here, all hardness results were known since the work of Bulatov, Jeavons and Krokhin~\cite{Bulatov:etal:sicomp2005}
and the algorithms of \cite{Bulatov:focs2017} and \cite{Zhuk:jacm2020} completed the proof of the conjecture. In between, {\em lots} of work went into studying the problem
from various algorithmic and algebraic angles, and
many ideas emerging from this project have been re-used 
in different contexts 
\iflong
(such as
proving complexity results for
infinite-domain CSPs~\cite{Bodirsky:InfDom} or for analysing {\em promise CSPs}~\cite{Krokhin:Oprsal:siglognews2022} and other generalisations).
\fi
\ifshort
(such as
infinite-domain CSPs~\cite{Bodirsky:InfDom} or {\em promise CSPs}~\cite{Krokhin:Oprsal:siglognews2022}).
\fi
\iflong
Optimization versions of CSP such as 
the 
Valued CSP (\textsc{VCSP}) have also been intensively studied.
An instance of the \textsc{VCSP} consists of a set of variables, a variable domain and a sum of functions, each function depending on a subset of the variables. Each function can take finite values specifying costs of assignments to its variables or the infinite value, indicating an infeasible assignment. The goal is to find an assignment to the
variables that minimizes the sum. 
It is obvious that many other CSP optimisation problems can be viewed
as \textsc{VCSP}s; examples include the \textsc{MaxCSP} and \textsc{MinCSP} problems where the instance
is a CSP instance and the goal is to find an assignment that maximises the number
of satisfied constraints (\textsc{MaxCSP}) or minimises the number of unsatisfied constraints (\textsc{MinCSP}).
\fi
\ifshort
Optimization versions of the CSP such as \textsc{MaxCSP} and \textsc{MinCSP} (where the goal is to find an assignment that maximises the number
of satisfied constraints (\textsc{MaxCSP}) or minimises the number of unsatisfied constraints (\textsc{MinCSP})) and the generalisation
Valued CSP (\textsc{VCSP}) have also been intensively studied.
\fi
\iflong
Such problems
have been the subject of various complexity classification projects and this has generated a wide range of interesting results.
\fi
Some notable results include the proof of
that every finite-domain \textsc{VCSP} is either
polynomial-time solvable or \NP-complete~\cite{Kolmogorov:etal:sicomp2017}, and
\iflong
the fine-grained result by Raghavendra~\cite{Raghavendra:stoc2008} stating that, under the Unique Games Conjecture, 
the semi-definite programming relaxation of finite-domain \textsc{MaxCSP}
achieves the best possible approximation factor in polynomial time.
\fi
\ifshort
the optimal approximability result for finite-domain \textsc{MaxCSP} under
the Unique Games Conjecture~\cite{Raghavendra:stoc2008}.
\fi
One should note that even if the \P/\NP borderline for finite-domain \textsc{VCSP}s
is fully known, there are big gaps in our understanding of the corresponding \FPT/\W{1} borderline (with parameter solution weight).
The situation is even worse if we consider infinite-domain optimization versions of the \textsc{CSP}s, since we cannot
expect to get a full picture of the \P/\NP borderline even for the basic CSP problem~\cite{Bodirsky:Grohe:icalp2008}. 
\iflong
Thus, one must concentrate on subclasses
that have sufficiently advantageous properties.
\fi

In the parameterized complexity world, 
\textsc{MinCSP} is a natural problem to study.
\iflong
The parameterized complexity of \textsc{MinCSP} is far less 
well studied than 
its classical computational complexity.
\fi
Subproblems that have gained attraction include Boolean constraint
languages~\cite{Bonnet:etal:esa2016,Kim:etal:FA3,razgon2009almost},
Dechter et al.'s~\cite{dechter1991temporal} simple temporal problem
(STP)~\cite{Dabrowski:etal:aaai2022},
linear inequalities~\cite{berczi2022resolving} and
linear equations~\cite{crowston2013parameterized,Dabrowski:etal:soda2023}, to list a few. Highly interesting results have emerged from studying 
the parameterized complexity of problems like these. For instance,
the recent dichotomy
for \textsc{MinCSP}s over the Boolean domain by Kim et al.~\cite{Kim:etal:FA3}
was obtained using a novel technique called \emph{directed flow augmentation}.
\iflong
The showcase application of this tool is a parameterized algorithm for \textsc{Chain SAT}~\cite{Kim:etal:stoc2022}.
This notorious open problem appeared as the only missing step 
in the study by Chitnis, Egri and Marx~\cite{Chitnis:etal:algorithmica2017} towards classifying
the complexity of list $H$-colouring a graph by removing few vertices.
\fi
Recent work has indicated {\em temporal} CSPs 
\iflong
as the next possible area of interest in this context~\cite{EibenRW22,kim2023flowaugmentation}. 
\fi
\ifshort
as a possible next step~\cite{EibenRW22,kim2023flowaugmentation}. 
\fi
Temporal CSPs are CSPs where the relations underlying the constraints are first-order
definable in $({\mathbb Q};<)$.
The computational complexity of temporal CSPs where we fix the set of allowed constraints
exhibits a dichotomy: every such problem is either polynomial-time solvable
or NP-complete~\cite{Bodirsky:Kara:jacm2010}. 
The \textsc{MinCSP} problem for temporal CSPs is closely related to a number
of graph separation problems. For example,
if we take the rationals as the domain and allow constraints $\leq$ and
$<$, the \textsc{MinCSP} problem is equivalent to \textsc{Directed Subset Feedback Arc Set}~\cite{kim2022weighted},
a problem known to be fixed-parameter tractable for two different, but both quite involved
reasons~\cite{chitnis2015directed,kim2022weighted}. 
If we allow the relations $\leq$ and $\neq$, we obtain a problem
equivalent to \textsc{Directed Symmetric Multicut}, whose parameterized complexity
is identified as the main open problem in the area of directed graph separation
problems~\cite{EibenRW22,kim2022weighted}.
Another related way forward is to analyse the \textsc{MinCSP}
for {\em Allen's interval algebra}. 
Allen's interval algebra is a highly influential formalism within AI and qualitative reasoning that has numerous applications, e.g. in
planning~\cite{Allen:Koomen:ijcai83,Mudrova:Hawes:icra2015,Pelavin:Allen:aaai87}, natural language processing~\cite{Denis:Muller:ijcai2011,Song:Cohen:aaai88} and molecular biology~\cite{Golumbic:Shamir:jacm93}. 
This CSP uses closed intervals with rational endpoints as domain values
and employs a set $\A$ of 13 basic comparison relations such as ``precedes'' (one interval finishes before the other
starts) or ``during'' (one interval is a strict subset of the other); see Table~\ref{tb:allen}.
Formally speaking, the
CSP for the interval algebra is not
a temporal CSP, since the underlying domain is based on intervals instead of points.
This difference is important: complexity classifications for the interval algebra
have been harder to obtain than for temporal constraints. There are full classifications
for binary relations~\cite{Krokhin:etal:jacm2003} and for first-order definable constraint languages
containing all basic relations~\cite{Bodirsky:etal:arxiv2022}; a classification
for all first-order definable constraint languages appears remote.

\smallskip

\noindent
{\bf Our contributions.}
The aim in this paper is to initiate a study of \textsc{MinCSP} in the context
of Allen's interval algebra.
Obtaining a full parameterized complexity classification for Allen's interval algebra
would
entail resolving the status of \textsc{Directed Symmetric Multicut}
and we do not aim at this very ambitious task.
Instead, we restrict ourselves to languages that are subsets of $\A$ and
do not consider more involved expressions (say, first-order logic)
built on top of $\A$.
Even in this limited quest, we are able to uncover new relations to graph separation problems,
and new areas of both tractability and intractability.
One of the main ingredients for both our tractability and intractability
results is a characterization of unsatifiable instances of \textsc{CSP}$(\A)$
in terms of minimal obstructions given by certain arc-labelled mixed cycles of the instance. That is, for certain key subsets $\Gamma \subseteq \A$, we provide a complete description of \emph{bad} cycles such that an instance of \textsc{CSP}$(\Gamma)$ is satisfiable if and only if it does not contain a bad cycle. This allows us show that \textsc{MinCSP}$(\Gamma)$ is equivalent to the problem of finding a minimum set of arcs that hit every bad cycle in an arc-labelled mixed graph.
We prove that there are seven inclusion-wise maximal subsets $\Gamma$
of $\A$ such that \textsc{MinCSP}$(\Gamma)$ is in \FPT, and that \textsc{MinCSP}$(\Gamma)$
is W[1]-hard in all other cases. 
We show that $\mincsp{\A}$
is not approximable in polynomial time within {\em any} constant under the UGC. 
In fact, we prove this to hold for \textsc{MinCSP}$(\rr)$
whenever $\rr \in \A \setminus \{\equiv\}$.
As a response to this,
we suggest the use of {\em fixed-parameter approximation algorithms}. 
We show that \textsc{MinCSP}$(\A)$ admits such an algorithm with approximation ratio~$2$ and a
substantially faster algorithm with approximation ratio~$4$.
We describe the results in greater detail below.

\smallskip

\noindent
{\bf Intractability results.}
Our intractability results are based on novel \W{1}-hardness
results for a variety of natural \emph{paired} and \emph{simultaneous}
cut and separation problems, which we believe to be of independent interest. 
Here, the input consists of two (directed or
undirected) graphs and the task is to find a ``generalized'' cut that
extends to both graphs. 
The two input graphs
share some arcs/edges that 
can be deleted simultaneously at unit cost, and
the goal is to compute a set of $k$ arcs/edges
that is a cut in both graphs.
Both paired and simultaneous problems have recently 
received attention from the parameterized complexity 
community~\cite{agrawal2018simultaneous,agrawal2021simultaneous,Kim:etal:FA3}.
In the FPT/\W{1}-hardness dichotomy for the Boolean domain~\cite{Kim:etal:FA3}, 
the fundamental difficult problem is \textsc{Paired Cut}
(proven to be \W{1}-hard by Marx and Razgon~\cite{MarxR09}): given an integer $k$
and a directed graph
$G$ with two terminals $s,t \in V(G)$ and some arcs grouped into \emph{pairs}, 
delete at most $k$ pairs to cut all paths from $s$ to $t$. 
An intuitive reason why \textsc{Paired Cut} is difficult can be seen as follows.
Assume $G$ contains two long $st$-paths $P$ and $Q$ and the arcs of $P$ are arbitrarily
paired with the arcs of $Q$. Then, one cuts both paths with a cost of only one pair,
but the arbitrary pairing of the arcs allow us to encode an arbitrary permutation
--- which is very powerful for encoding edge-choice gadgets when
reducing from \textsc{Multicolored Clique}. Our strategy for proving \W{1}-hardness
of paired problems elaborates upon this idea.
Our reductions of simultaneous problems use a similar leverage. 
While the simultaneous setting may not be as versatile at the first glance as
the arbitrary pairing of \textsc{Paired Cut}, the possibility of choosing common arcs/edges
while deleting them at a unit costs still leaves enough freedom to encode
arbitrary permutations.

Altogether, we obtain novel \W{1}-hardness results for \textsc{Paired Cut Feedback Arc Set}, \textsc{Simultaneous st-Separator}, \textsc{Simultaneous Directed st-Cut}, and \textsc{Simultaneous Directed Feedback Arc Set}. 
This allows us to identify six intractable fragments subsets of basic interval relations:
$\{\m,\rr_1\}$ for $\rr_1 \in \{\e\nobreak,\allowbreak\s,\f\}$, $\{\d, \rr_2\}$ for $\rr_2 \in \{\o,\p\}$ and $\{\p,\o\}$.
The hardness reduction for $\{\m,\rr_1\}$, is based on the hardness of \textsc{Paired Cut} and \textsc{Paired Cut Feedback Arc Set}. The other hardness results follow from reductions from \textsc{Simultaneous Directed Feedback Arc Set}, whose \W{1}-hardness is shown by a reduction from \textsc{Simultaneous st-Separator}.

\smallskip

\noindent
{\bf Tractability results.}
We identify seven maximal tractable sets of basic interval relations:
$\{\m, \p\}$ and $\{r_1, r_2, \e\}$ for
$r_1 \in \{\s, \f\}$ and $r_2 \in \{\p, \d, \o\}$.
All problems are handled by reductions to variants of
\textsc{Directed Feedback Arc Set} (DFAS).
\iflong
This classical problem asks, given a directed graph $G$ and a budget $k$,
whether one can make $G$ acyclic by deleting $k$ arcs.
\fi
DFAS and variations are extensively studied in parameterized complexity~\cite{DBLP:journals/algorithmica/BergougnouxEGOR21,DBLP:conf/wg/BonamyKNPSW18,chen2008fixed,chitnis2015directed,DBLP:conf/icalp/GokeMM20,DBLP:journals/disopt/GokeMM22,DBLP:journals/talg/LokshtanovRS18,DBLP:conf/soda/LokshtanovR018}.
\iflong
One variant particularly relevant to our work is \textsc{Subset DFAS} (\subdfas).
Here the input comes with a subset of red arcs,
and the goal is to delete $k$ arcs so that the resulting
graph has no directed cycles containing red arcs.
Note that directed cycles of black (i.e. non-red) arcs are allowed to stay.
\subdfas is equivalent to $\mincsp{<, \leq}$:
regard the constraints $u < v$ as red $(u, v)$-arcs
and constraints $u \leq v$ as black $(u, v)$-arcs,
and observe that directed cycles of $\leq$-constraints are satisfiable,
while any directed cycle with a $<$-constraint is not.
\subdfas is fixed-parameter tractable~\cite{chitnis2015directed},
and the weighted version of the problem (where the input arcs have weights,
and a separate weight budget is given in the input)
was recently~\cite{kim2022weighted} solved using {\em flow augmentation}~\cite{Kim:etal:stoc2022}.
To show that \mincsp{\m, \p} is in FPT, we use a fairly straightforward
reduction to \mincsp{<, =} which, in turn, reduces to \mincsp{<, \leq}.
\fi
\ifshort
DFAS is equivalent to $\mincsp{<}$,
and a variant particularly important
in our work is \subdfas equivalent to $\mincsp{<,\leq}$, where
the goal is to destroy only 
directed cycles that have a $<$-arc.
To show that \mincsp{\m, \p} is in FPT, we use a straightforward
reduction to \mincsp{<, =} which, in turn, reduces to \mincsp{<, \leq}.
\fi

For the remaining six tractable cases, we reduce them all
to a new variant of DFAS called 
\textsc{Mixed Feedback Arc Set with Short and Long Arcs} (\lsmfas).
\iflong
The input to this problem is a mixed graph, i.e. a graph with
both arcs and undirected edges; the arcs are either short or long.
\fi 
\ifshort
The input is a mixed graph with edges and
long and short arcs.
\fi
Forbidden cycles in this graph are of two types:
(1) cycles with at least one short arc, no long arcs, and all short arcs in the same direction, and
(2) cycles with at least one long arc, all long arcs in the same direction, but 
the short arcs can be traversed in arbitrary direction.
One intuitive way %
to think about the problem is to observe that a graph $G$ has no forbidden cycle
if there exists a placement of the vertices on the number line
with certain distance constraints represented by the edges and arcs.
Vertices connected by edges (which correspond to $\e$-constraints) should be placed at the same point.
If there is an arc $(u, v)$, we need to place $u$ before $v$.
Moreover, if the arc is long, then the distance from $u$ to $v$ should be big 
(say, greater than twice the number of vertices), while is the arc is short,
then the distance from $u$ to $v$ should be small (say, at most $1$).
\iflong
In other words, the graph $G$ after contracting all undirected edges is acyclic
and we request a topological ordering of this contracted graph where
all (weak) connected components of short arcs are grouped together and long arcs
go only between distinct connected components of short arcs. 
\lsmfas generalizes \mincsp{<, =}, which is a special case without long arcs.
\fi
\ifshort
The reduction from $\mincsp{\rr_1, \rr_2, \e}$ to \lsmfas
creates a mixed graph with edges for $\e$-constraints,
short arcs for $\rr_1$-constraints and
long arcs for $\rr_2$-constraint.
For correctness, consider for example the case with
$\rr_1 = \s$ and $\rr_2 = \p$. 
Forbidden cycles of the first and the second kind
imply an inconsistent order on the right and left endpoints, respectively. 
On the other hand, if bad cycles are absent, we can assign intervals as follows:
if two variables are $\e$-connected, they are assigned the same interval,
if they are $\s$-connected, their intervals have the same left endpoints,
left endpoints are ordered according to $\p$-constraints,
and the right endpoints of $\s$-connected intervals
are ordered according to the $\s$-constraints.
\fi

Our algorithm for \lsmfas builds upon  
the algorithm of~\cite{kim2022weighted} for \subdfas.
\iflong
The opening steps are standard for parameterized deletion problems. 
We start with iterative compression,
which reduces the problem to a special case
when the input comes together with
a set of $k + 1$ vertices intersecting all forbidden cycles.
We call these vertices \emph{terminals} and
guess how they are positioned relative to each other in the graph
obtained by deleting the arcs of a hypothetical optimal solution.
Altogether, this requires $O^*(2^{O(k \log k)})$ time.
\fi
\ifshort
By iterative compression and branching,
we may assume access to $k + 1$ vertices
that intersect all forbidden cycles,
and we know their relative positions
in the graph obtained after deleting a hypothetical optimal solution.
\fi
The aforementioned ``placing on a line'' way of phrasing the lack of forbidden cycles
is the main reason why this leads to a complete algorithm. 
In the following step, we try to place 
all remaining vertices relative to the terminals
while breaking at most $k$ distance constraints.
Note that the number of ways in which a vertex can relate to a terminal is constant:
it may be placed in a short/long distance 
before/after the terminal, or in the same position.
Thus, we can define $O(k)$ types for each vertex,
and the types determine whether distance constraints are satisfied or not.
The optimal type assignment is then obtained by a reduction to
\textsc{Bundled Cut} with pairwise linked deletable arcs,
a workhorse problem shown to be fixed-parameter tractable in~\cite{Kim:etal:FA3}.
\ifshort
We remark that our algorithms also handle the weighted versions of the problems.
\fi

\iflong 
To see how \textsc{MinCSP}$(\rr_1,\rr_2,\equiv)$ for $\rr_1 \in \{\s, \f\}$ and $\rr_2 \in \{\p, \d, \o\}$
relates to \lsmfas{}, consider the exemplary case of \textsc{MinCSP}$(\s,\p,\equiv)$.
Represent variables as vertices, $\s$ relations as short arcs and $\p$ relations as long arcs. 
Observe that in a satisfiable instance, the resulting graph has no forbidden cycles
for \lsmfas{}: a cycle consisting of undirected edges and short arcs oriented in the same direction
corresponds to a loop on $\s$ relations, while a cycle consisting of undirected edges,
short arcs in arbitrary directions, and long arcs in the same direction
corresponds to a cycle in the \textsc{MinCSP} instance with $\equiv$, $\s$ and $\p$ relations, where
all $\p$ relations are in the same direction; such a cycle imposes an unsafisfiable cyclic order
on the starting points of the corresponding intervals. 
Interestingly, one can prove an if-and-only-if statement: the satisfiability of a \textsc{MinCSP}$(\s,\p,\equiv)$
instance is equivalent to the absence of forbidden cycles in the corresponding graph. 
Hence, it suffices to solve \lsmfas there.
\fi

\smallskip

\noindent
{\bf Approximation results.}
\iflong
As a response to the negative complexity results for $\mincsp{\A}$, it is natural to consider approximation algorithms.
However, we show that $\mincsp{\A}$
is not approximable in polynomial time within any constant under the UGC. 
In fact, we prove this to hold for all $\mincsp{\rr}$
with $\rr \in \A \setminus \{\equiv\}$.

An obvious next step is to
use {\em fixed-parameter approximation algorithms}. This is an approach combining parameterized complexity with
approximability and
it has received rapidly increasing interest
(see, for instance, the surveys by Feldmann et al.~\cite{Feldmann:etal:algorithms2020} and Marx~\cite{Marx:tcj2008}). 
Let $c > 1$ be a constant. A factor-$c$ {\em fpt-approximation
algorithm} for \textsc{MinCSP}$(\A)$
takes an instance
$(I,k)$ of \textsc{MinCSP}$(\A)$ and
either returns that there is no solution of size at most $k$ or returns that there is 
a solution of size at most $c \cdot k$. The running time of the algorithm is bounded by 
$f(k) \cdot \norm{I}^{O(1)}$ where $f: \naturals \rightarrow \naturals$
is some computable function. Thus, the algorithm is given more
time to compute the solution (compared to a polynomial-time
approximation algorithm) and it may output a
slightly oversized solution (unlike an exact fpt algorithm for $\mincsp{\A}$).
This combination turns out to be successful.
We show that $\mincsp{\A}$ admits such an algorithm with $c=2$
and a substantially faster algorithm with $c=4$. These results
are (once again) based on the
\subdfas problem.
The idea is to look at two relaxations and consider the left and the right endpoints of the intervals separately.
The problem for left/right endpoints is an instance of \subdfas, which we can solve in fpt time.
If either of the relaxations is not solvable, then we deduce that the original problem has no solution either.
Otherwise, we have two deletion sets, each of size at most $k$, that ensure consistency for the left and the right endpoints, respectively.
Combining them, we obtain a deletion set for the initial problem of size at most $2k$.
If we use a factor-$2$ fpt approximation algorithm for \subdfas instead of the slower exact fpt algorithm
when solving the relaxations,
we obtain two solutions of size at most $2k$, and combined they yield a solution of size at most $4k$.
We also show that there exists a factor-$2$ fpt approximation algorithm for the weighted version of $\mincsp{\A}$.
This harder problem can be solved with
directed flow augmentation~\cite{Kim:etal:stoc2022},
albeit at the price of increased time complexity. \fi
\ifshort 
In response to the negative complexity results for $\mincsp{\A}$, 
we consider approximation algorithms.
We show that $\mincsp{\A}$
is not approximable in polynomial time within any constant under the UGC. 
We relax the restrictions even more by allowing
our approximation algorithms to run in fixed-parameter tractable time.
We show that $\mincsp{\A}$ admits such an algorithm with $c=2$
and a substantially faster algorithm with $c=4$. 
These results are based on the observation
that every relation in $\A$ can be defined
as a conjunction of $\{<,=\}$-constraints
on the endpoints.
In the relaxation we disregard conjunctions
and view all $\{<,=\}$-constraints as an instance
of $\mincsp{<,=}$, which is then reduced to \subdfas.
By invoking \subdfas algorithm of~\cite{kim2022weighted},
one obtains a $2$-approximation algorithm
for the weighted variant of the problem,
albeit at the price of increased time complexity.
\fi 

\smallskip

\iflong
\noindent
{\bf Roadmap.}
\label{sec:roadmap}
We begin by introducing some necessary preliminaries in Section~\ref{sec:prelims}.
Section~\ref{sec:classification} contains a bird's eye view of our 
dichotomy
theorem for the parameterized complexity of $\mincsp{\Gamma}$ where
$\Gamma \subseteq \A$, and the technical details are collected in the
following sections. 
We describe the minimal obstructions to satisfiability for certain subsets of~$\A$ in Section~\ref{sec:cycles}. These results are essential for connecting the \textsc{MinCSP}$(\A)$
problem with the graph-oriented problems that we use in the paper.
We complete our dichotomy result by presenting a number of
fixed-parameter algorithms 
in Section~\ref{sec:fpt-algs} and a collection of \W{1}-hardness results
in Section~\ref{sec:w-hard}. Our approximability results can be found
in Section~\ref{sec:approx}.
We conclude the paper with a brief discussion of our results and future
research directions in Section~\ref{sec:discussion}.
\fi

\ifshort
\noindent
{\bf Roadmap.}
\label{sec:roadmap}
We present the necessary preliminaries in Section~\ref{sec:prelims}.
Section~\ref{sec:classification} is a bird's eye view of our 
results for the parameterized complexity and approximability of $\mincsp{\Gamma}$ and the technical details are collected in the
following sections. 
We describe the minimal obstructions to satisfiability for certain subsets of~$\A$ in Section~\ref{sec:cycles}. These results are essential for connecting the \textsc{MinCSP}$(\A)$
problem with the graph-oriented view that we use.
We complete our dichotomy result by a number of
fixed-parameter algorithms 
in Section~\ref{sec:fpt-algs} and a collection of \W{1}-hardness results
in Section~\ref{sec:w-hard}. 
We conclude the paper with a discussion of our results and future
research directions in Section~\ref{sec:discussion}.
This is a shortened version of the full paper where, for instance, some proofs are sketched. The full version can be found in the supplementary material.
\fi

\iflong\begin{table}[bt]
    \begin{center}
      {\normalsize
    \begin{tabular}{|ll|c|l|}\hline
              Basic relation        &               & Example                   & Endpoint Relations\\
        \hline\hline
              $I$ precedes      $J$ & ${\sf p}$     & \stexttt{iii\spv}         & $I^{+}<J^{-}$ \\ 
              $J$ preceded by   $I$ & ${\sf pi}$    & \stexttt{\spv jjj}        & \\
        \hline
              $I$ meets         $J$ & ${\sf m}$     & \stexttt{iiii\spiv}       & $I^{+}=J^{-}$\\ 
              $J$ met-by        $I$ & ${\sf mi}$    & \stexttt{\spiv jjjj}      & \\
        \hline
              $I$ overlaps      $J$ & ${\sf o}$     & \stexttt{iiii\spii}       & $I^{-}<J^{-}<I^{+}<J^{+}$\\
              $J$ overlapped-by $I$     & ${\sf oi}$    & \stexttt{\spii jjjj}      & \\
        \hline
              $I$ during        $J$ & ${\sf d}$     & \stexttt{\spii iii\spii}  & $I^{-}>J^{-}$, $I^{+}<J^{+}$ \\
              $J$ includes      $I$ & ${\sf di}$    & \stexttt{jjjjjjj}         &  \\
        \hline
              $I$ starts        $J$ & ${\sf s}$     & \stexttt{iii\spiv}        & $I^{-}=J^{-}$, $I^{+}<J^{+}$
        \\ 
              $J$ started by    $I$ & ${\sf si}$    & \stexttt{jjjjjjj}         & \\
        \hline
              $I$ finishes      $J$ & ${\sf f}$     & \stexttt{\spiv iii}       & $I^{+}=J^{+}$, $I^{-}>J^{-}$
        \\ 
              $J$ finished by   $I$ & ${\sf fi}$    & \stexttt{jjjjjjj}         & \\
        \hline
              $I$ equals        $J$ & $\equiv$     & \stexttt{iiii}            & $I^{-}=J^{-}$, $I^{+}=J^{+}$\\
                                    &               & \stexttt{jjjj}            & \\
\hline
    \end{tabular}}
    \caption{The thirteen basic relations in Allen's Interval Algebra. The
    endpoint relations $I^- < I^+$ and $J^- < J^+$ that are valid for
    all relations have been omitted.}
    \label{tb:allen}
    \end{center}
\end{table}\fi
\ifshort
\begin{table}[bt]
    \noindent\adjustbox{left=0.5\textwidth,valign=t}{%
    \begin{tabular}{|ll|c|l|}\hline
              Basic relation        &               & Ex.                   & Endpoint Relations\\
        \hline\hline
              $I$ precedes      $J$ & ${\sf p}$     & \stexttt{iii\spv}         & $I^{+}<J^{-}$\\
              $J$ preceded by   $I$ & ${\sf pi}$    & \stexttt{\spv jjj}        & \\
        \hline
              $I$ meets         $J$ & ${\sf m}$     & \stexttt{iiii\spiv}       & $I^{+}=J^{-}$\\
              $J$ met-by        $I$ & ${\sf mi}$    & \stexttt{\spiv jjjj}      & \\
        \hline
              $I$ overlaps      $J$ & ${\sf o}$     & \stexttt{iiii\spii}       & $I^{-}<J^{-}<I^{+}<J^{+}$\\
              $J$ overlapped-by $I$     & ${\sf oi}$    & \stexttt{\spii jjjj}      & \\
        \hline
                            $I$ during        $J$ & ${\sf d}$     & \stexttt{\spii iii\spii}  & $I^{-}>J^{-}$, $I^{+}<J^{+}$ \\
              $J$ includes      $I$ & ${\sf di}$    & \stexttt{jjjjjjj}         &  \\
        \hline
    \end{tabular}
    }%
    \adjustbox{right=0.5\textwidth,valign=t}{%
    \begin{tabular}{|ll|c|l|}\hline
              Basic relation        &               & Ex.                   & Endpoint Relations\\
        \hline\hline
              $I$ starts        $J$ & ${\sf s}$     & \stexttt{iii\spiv}        & $I^{-}=J^{-}$, $I^{+}<J^{+}$
        \\ 
              $J$ started by    $I$ & ${\sf si}$    & \stexttt{jjjjjjj}         & \\
        \hline
              $I$ finishes      $J$ & ${\sf f}$     & \stexttt{\spiv iii}       & $I^{+}=J^{+}$, $I^{-}>J^{-}$
        \\ 
              $J$ finished by   $I$ & ${\sf fi}$    & \stexttt{jjjjjjj}         & \\
        \hline
              $I$ equals        $J$ & $\equiv$     & \stexttt{iiii}            & $I^{-}=J^{-}$, $I^{+}=J^{+}$\\
                                    &               & \stexttt{jjjj}            & \\
\hline
    \end{tabular}
    }
    \caption{The thirteen basic relations in Allen's Interval Algebra. The
    endpoint relations $I^- < I^+$ and $J^- < J^+$ that are valid for
    all relations have been omitted.}
    \label{tb:allen}
    
\end{table}
\fi

\section{Preliminaries}
\label{sec:prelims}

\iflong
In this section we introduce the necessary prerequisites: we briefly present the rudiments of parameterized complexity in Section~\ref{sec:paramcomplexity}, we define the CSP and \textsc{MinCSP} problems in Section~\ref{sec:csp}, and we 
provide some basics concerning interval relations in Section~\ref{sec:intrel}.
Before we begin, we need some terminology and notation for graphs.
Let $G$ be a (directed or undirected) graph; we allow graphs to contain loops.
We denote the set of vertices in $G$ by $V(G)$.
If $G$ is undirected, let $E(G)$ denote the set of edges in $G$.
If $G$ is directed, let $A(G)$ denote the set of arcs in $G$,
and $E(G)$ denote the set of edges in the underlying undirected graph of $G$.
We use~$uv$ to denote an undirected edge with end-vertices $u$ and~$v$.
We use~$(u,v)$ to denote a directed arc from~$u$ to~$v$; the end-vertex~$u$ is the \emph{tail} of the arc and the end-vertex~$v$ is the \emph{head}.
For every subset $X \subseteq E(G)$,
we write $G - X$ to denote the directed graph obtained by
removing all edges/arcs corresponding to $X$ from $G$
(formally, $V(G - X) = V(G)$ with $E(G - X) = E(G) \setminus X)$ if $G$ is undirected
and $A(G - X) = A(G) \setminus \{ (u, v), (v, u) \mid \{u,v\} \in X \})$ if $G$ is directed).
If $X \subseteq V(G)$ is a set of vertices,
then we let $G - X = G[V(G) \setminus X]$ be the subgraph
induced in $G$ by $V(G) \setminus X$.
An $st$-cut in $G$ is a set of edges/arcs $X$ such that the vertices $s$ and $t$
are separated in $G - X$.
If $n \in \mathbb{N}$, then $[n]$ denotes the set $\{1,\ldots,n\}$.

\subsection{Parameterized Complexity}
\label{sec:paramcomplexity}

A {\em parameterized} problem is a subset of $\Sigma^* \times \naturals$,
where $\Sigma$ is the input alphabet.
The parameterized complexity class \FPT contains the problems decidable in $f(k)\cdot n^{O(1)}$ time, where $f$ is a computable function and $n$ is the instance size; this forms the set of the computationally most tractable problems.
Reductions between parameterized problems need to take
the parameter into account. 
To this end, we use \emph{parameterized reductions} (or fpt-reductions).
Consider two parameterized problems $L_1, L_2 \subseteq \Sigma^* \times \naturals$.
A mapping $P: \Sigma^* \times \naturals \rightarrow \Sigma^* \times \naturals$
is a parameterized reduction from $L_1$ to $L_2$ if
\begin{enumerate}[(1)]
  \item $(x, k) \in  L_1$ if and only if $P((x, k)) \in L_2$, 
  \item the mapping can be computed in $f(k) \cdot n^{O(1)}$ time for some computable function $f$, and 
  \item there is a computable function $g : \naturals \rightarrow \naturals$  
  such that for all $(x,k) \in \Sigma^* \times \naturals$, if $(x', k') = P((x, k))$, then $k' \leq g(k)$.
\end{enumerate}
We will sometimes prove
that certain problems are not in $\FPT$.
The class $\Weft[1]$ contains all problems that are fpt-reducible to \textsc{Independent Set} 
parameterized by the solution size, i.e. the number of vertices in the independent set.
Showing $\Weft[1]$-hardness (by an fpt-reduction) for a problem rules out the existence of an fpt
algorithm under the standard assumption that $\FPT \neq \Weft[1]$.

\subsection{Constraint Satisfaction}
\label{sec:csp}

A {\em constraint language} $\Gamma$ is a set of relations over a domain $D$.
Each relation $R \in \Gamma$ has an associated \emph{arity} $r \in \naturals$ and $R \subseteq D^r$.
All relations considered in this paper are binary
and all constraint languages are finite.
An instance $\III$ of $\csp{\Gamma}$ consists of a set of variables $V(\III)$ 
and a set of constraints $C(\III)$ of the form $R(x,y)$,
where $R \in \Gamma$ and $x,y \in V(\III)$.
To simplify notation, we may write $R(x,y)$ as $xRy$.
An assignment $\varphi : V(\III) \rightarrow D$ \emph{satisfies} a constraint $R(x,y)$
if $(\varphi(x), \varphi(y)) \in R$
and \emph{violates} $R(x,y)$ if $(\varphi(x), \varphi(y)) \notin R$. 
The assignment $\varphi$ is a {\em satisfying assignment} (or a {\em solution}) if it satisfies
every constraint in $C(\III)$.

\pbDef{\csp{\Gamma}}
{An instance $\III$ of $\csp{\Gamma}$.}
{Does $\III$ admit a satisfying assignment?}

The {\em value} of an assignment $\varphi$ for $\III$ is the number of constraints in $C(\III)$ satisfied by $\varphi$.
For any subset of constraints $X \subseteq C(\III)$, 
let $\III-X$ denote the instance with $V(\III-X)=V(\III)$ and $C(\III-X) = C(\III) \setminus X$.
The (parameterized) \emph{almost constraint satisfaction problem} ($\mincsp{\Gamma}$)
is defined as follows:

\pbDefP{\mincsp{\Gamma}}
{An instance $\III$ of $\csp{\Gamma}$ and an integer $k$.}
{$k$.}
{Is there a set $X \subseteq C(\III)$ such that $\abs{X} \leq k$ and $\III-X$ is satisfiable?}

Given an instance $\close{\III,k}$ of $\mincsp{\Gamma}$, the set 
$X$ can be computed 
with $|C(\III)|$ calls to an algorithm for \mincsp{\A}. Hence, we can view \mincsp{\Gamma} as a decision problem without loss of generality. \
Additionally note that if \mincsp{\Gamma} is in \FPT, then \csp{\Gamma} must be polynomial-time solvable: an instance $\III$ of
\csp{\Gamma} is satisfiable if and only if the instance $\close{\III,0}$ of \mincsp{\Gamma} is a yes-instance, and instances of the form
$\close{\III,0}$ are solvable in polynomial time if \mincsp{\Gamma} is in \FPT.

\subsection{Interval Relations}
\label{sec:intrel}

We begin by reviewing the basics of
{\em Allen's interval algebra}~\cite{Allen:cacm83} (IA). 
Its domain is the set
${\mathbb I}$ of all pairs $(x,y) \in {\mathbb Q}^2$ such that $x<y$, i.e. ${\mathbb I}$ can be viewed as the set of all closed intervals $[a,b]$ of 
rational numbers. If $I=[a,b] \in \mathbb I$, then we write $I^-$ for $a$ and $I^+$ for $b$. 
Let $\A$ denote the set of 13 {\em basic} relations
that are presented in Table~\ref{tb:allen}, and let
$2^{\A}$ denote the 8192
binary relations that can be formed by
taking unions of relations in $\A$.
The complexity of CSP$(\Gamma)$ is known for every
$\Gamma \subseteq 2^{\A}$~\cite{Krokhin:etal:jacm2003} and in each case
CSP$(X)$ is either polynomial-time solvable or NP-complete. In particular, CSP$(\A)$ is in P.
When considering subsets $\Gamma \subseteq \A$,
note that any constraint $x {\sf ri} y$ is equivalent to $y \rr x$
for $\rr \in \{\p, \m, \o, \d, \s, \f\}$,
so we may assume that $\rr \in \Gamma$ if and only if ${\sf ri} \in \Gamma$.
Furthermore, for the remainder of the paper, we may assume $\A = \{\p,\m,\o,\d,\s,\f,\e\}$.

When studying the optimization problem \textsc{MinCSP},
it is convenient to allow \emph{crisp} constraints,
i.e. constraints that cannot be deleted.
Formally, for a language $\Gamma \subseteq \A$ 
and a relation $\rr \in \Gamma$, we say that 
\emph{$\Gamma$ supports crisp $\rr$-constraints}
if, for every value of the parameter $k \in \naturals$,
we can construct an instance $\III_\rr$ of \mincsp{\Gamma}
with primary variables $x,y \in V(\III_\rr)$
such that the constraint $x \rr y$ is equivalent to $\III_\rr - X$ 
for all $X \subseteq C(\III)$ such that $|X| \leq k$.
Then, if we want to enforce a constraint $x \rr y$ in
an instance of \mincsp{\Gamma}, we can use $\III_\rr$ with fresh variables $V(\III) \setminus \{x,y\}$ in its place.
We collect the tricks to make constraints crisp in the following two lemmas.
First, we show the result for \emph{idempotent} binary relations over an infinite domain $D$,
i.e. relations $\rr$ that are transitive and dense:
$a \rr b$ and $b \rr c$ implies $a \rr c$ for all $a,b,c \in D$, and, 
for every $a,c \in D$ such that $a \rr c$ holds,
there exists $b \in D$ such that $a \rr b$ and $b \rr c$ hold.

\begin{lemma}
  \label{lem:idempotent}
  Let $\rr$ be an idempotent binary relation over an infinite domain $D$.
  Then $\{\rr\}$ supports crisp $\rr$-constraints.
\end{lemma}
\begin{proof}
The instance $\III_\rr$ of $\mincsp{\rr}$ can be constructed as follows:
introduce $k + 1$ auxiliary variables $u_1, \dots, u_{k+1}$ and
let $\III_\rr$ contain constraints $x \rr u_i$, $u_i \rr y$ for all $i \in [k+1]$.
After deleting any $k$ constraints from $I$,
there remains $i \in [k+1]$ such that 
$x \rr u_i \rr y$ is still in the instance.
Correctness follows by idempotence of $\rr$:
if an assignment satisfies $x \rr u_i \rr y$, then it satisfies $x \rr y$ by transitivity of $\rr$;
conversely, if an assignment satisfies $x \rr y$,
then it can be extended to the auxiliary variables since $D$ is infinite and $\rr$ is dense.
\end{proof}

We now show that all interval languages $\Gamma \subseteq \A$ support crisp $\rr$-constraints 
for all relations $\rr \in \Gamma$.

\begin{lemma}
  \label{lem:crisp}
  Let $\rr \in \A$ be an interval relation.
  Then $\{\rr\}$ supports crisp $\rr$-constraints.
\end{lemma}
\begin{proof}
First, observe that $\e, \p, \d, \s, \f \in \A$ are idempotent,
so Lemma~\ref{lem:idempotent} applies.
There are two relations left, namely $\m$ and $\o$.
To construct $\III_\m$, introduce $2(k + 1)$ auxiliary variables $u_i$, $v_i$, 
and create constraints 
$x \m u_i$, $v_i \m u_i$, $v_i \m y$ for all $i \in [k+1]$.
For one direction,
observe that after deleting any $k$ constraints,
three constraints with $x, y$ and $u_i, v_i$ for some $i$ remain.
By definition of the relation $\m$,
these four constraints imply that $x^+ = u_i^- = v_i^+ = y^-$,
hence they imply $x \m y$.
For the converse,
if an assignment satisfies $x \m y$,
then setting $u_i = y$ and $v_i = x$ satisfies $\III_\m$.

To construct $\III_\o$, introduce $2(k + 1)$ auxiliary variables
$u_i, v_i$, and create constraints
$x \o u_i$, $u_i \o y$,
$v_i \o x$, $v_i \o y$
for all $i \in [k+1]$.
For one direction, observe that
constraints $x \o u_i$, $u_i \o y$ imply 
$x^- < u_i^- < x^+ < u_i^+$
and $u_i^- < y^- < u_i^+ < y^+$, hence
$x^- < y^-$ and $x^+ < y^+$.
The constraints $v_i \o x$, $v_i \o y$ imply
$v_i^- < x^- < v_i^+ < x^+$ and
$v_i^- < y^- < v_i^+ < y^+$,
hence $y^- < v_i^+ < x^+$.
We obtain that $x^- < y^- < x^+ < y^+$, i.e. $x \o y$ by definition.
For the converse,
if an assignment satisfies $x \o y$,
we can extend it to satisfy $\III_\o$
by choosing any two intervals $u$ and $v$ such that
$v^- < x^- < u^- < y^- < v^+ < x^+ < u^+ < y^+$ and setting $u_i=u$ and $v_i=v$ for all $i \in [k+1]$.
\end{proof}
\fi

\ifshort
In this section, we briefly present the rudiments of parameterized complexity, define the CSP and \textsc{MinCSP} problems, and  
provide some basics concerning interval relations.
Before we begin, we need some terminology and notation for graphs.
Let $G$ be a (directed or undirected) graph; we allow graphs to contain loops.
We denote the set of vertices in $G$ by $V(G)$.
If $G$ is undirected, then $E(G)$ denotes the set of edges in $G$.
If $G$ is directed, then $A(G)$ denotes the set of arcs in $G$,
and $E(G)$ denotes the set of edges in the underlying undirected graph of $G$.
We use~$uv$ to denote an undirected edge with end-vertices $u$ and~$v$.
We use~$(u,v)$ to denote a directed arc from~$u$ to~$v$; $u$ is the \emph{tail} and $v$ is the \emph{head}.
For $X \subseteq E(G)$,
we write $G - X$ to denote the directed graph obtained by
removing all edges/arcs corresponding to $X$ from $G$
if $G$ is undirected
and $A(G - X) = A(G) \setminus \{ (u, v), (v, u) \mid \{u,v\} \in X \})$ if $G$ is directed.
If $X \subseteq V(G)$,
then we let $G - X = G[V(G) \setminus X]$ be the subgraph
induced in $G$ by $V(G) \setminus X$.
An $st$-cut in $G$ is a set of edges/arcs $X$ such that the vertices $s$ and $t$
are separated in $G - X$.

A {\em parameterized} problem is a subset of $\Sigma^* \times \naturals$,
where $\Sigma$ is the input alphabet.
The parameterized complexity class \FPT contains the problems decidable in $f(k)\cdot n^{O(1)}$ time, where $f$ is a computable function and $n$ is the instance size.
Reductions between parameterized problems need to take
the parameter into account. 
To this end, we use \emph{parameterized reductions} (or fpt-reductions).
Consider two parameterized problems $L_1, L_2 \subseteq \Sigma^* \times \naturals$.
A mapping $P: \Sigma^* \times \naturals \rightarrow \Sigma^* \times \naturals$
is a parameterized reduction from $L_1$ to $L_2$ if
(1) $(x, k) \in  L_1$ if and only if $P((x, k)) \in L_2$, 
(2) the mapping can be computed in $f(k) \cdot n^{O(1)}$ time for some computable function $f$, and
(3) there is a computable function $g : \naturals \rightarrow \naturals$  
  such that for all $(x,k) \in \Sigma^* \times \naturals$, if $(x', k') = P((x, k))$, then $k' \leq g(k)$.
We will sometimes prove
that certain problems are not in $\FPT$.
The class $\Weft[1]$ contains all problems that are fpt-reducible to \textsc{Independent Set} 
parameterized by the number of vertices in the independent set.
Showing $\Weft[1]$-hardness (by an fpt-reduction) for a problem rules out the existence of an fpt
algorithm under the standard assumption that $\FPT \neq \Weft[1]$.

We continue by defining CSPs.
A {\em constraint language} $\Gamma$ is a set of relations over a domain $D$.
Each relation $R \in \Gamma$ has an associated \emph{arity} $r \in \naturals$ and $R \subseteq D^r$.
All relations considered in this paper are binary
and all constraint languages are finite.
An instance $\III$ of $\csp{\Gamma}$ consists of a set of variables $V(\III)$ 
and a set of constraints $C(\III)$ of the form $R(x,y)$,
where $R \in \Gamma$ and $x,y \in V(\III)$.
To simplify notation, we may write $R(x,y)$ as $xRy$.
An assignment $\varphi : V(\III) \rightarrow D$ \emph{satisfies} a constraint $R(x,y)$
if $(\varphi(x), \varphi(y)) \in R$
and \emph{violates} $R(x,y)$ if $(\varphi(x), \varphi(y)) \notin R$. 
The assignment $\varphi$ is a {\em satisfying assignment} (or a {\em solution}) if it satisfies
every constraint in $C(\III)$.

\pbDef{\csp{\Gamma}}
{An instance $\III$ of $\csp{\Gamma}$.}
{Does $\III$ admit a satisfying assignment?}

The {\em value} of an assignment $\varphi$ for $\III$ is the number of constraints in $C(\III)$ satisfied by $\varphi$.
For any subset of constraints $X \subseteq C(\III)$, 
let $\III-X$ denote the instance with $V(\III-X)=V(\III)$ and $C(\III-X) = C(\III) \setminus X$.
The (parameterized) \emph{almost constraint satisfaction problem} ($\mincsp{\Gamma}$)
is defined as follows:

\pbDefP{\mincsp{\Gamma}}
{An instance $\III$ of $\csp{\Gamma}$ and an integer $k$.}
{$k$.}
{Is there a set $X \subseteq C(\III)$ such that $\abs{X} \leq k$ and $\III-X$ is satisfiable?}

\iflong
Given an instance $\close{\III,k}$ of $\mincsp{\Gamma}$, the set 
$X$ can be computed 
with $|C(\III)|$ calls to an algorithm for \mincsp{\A}. Hence, we can view \mincsp{\Gamma} as a decision problem without loss of generality. \
Additionally note that if \mincsp{\Gamma} is in \FPT, then \csp{\Gamma} must be polynomial-time solvable: an instance $\III$ of
\csp{\Gamma} is satisfiable if and only if the instance $\close{\III,0}$ of \mincsp{\Gamma} is a yes-instance, and instances of the form
$\close{\III,0}$ are solvable in polynomial time if \mincsp{\Gamma} is in \FPT. 
\fi

Next, we review the basics of
{\em Allen's interval algebra}~\cite{Allen:cacm83} (IA). 
Its domain is the set
${\mathbb I}$ of all pairs $(x,y) \in {\mathbb Q}^2$ such that $x<y$, i.e. ${\mathbb I}$ can be viewed as the set of all closed intervals $[a,b]$ of 
rational numbers. If $I=[a,b] \in \mathbb I$, then we write $I^-$ for $a$ and $I^+$ for $b$. 
Let $\A$ denote the set of 13 {\em basic} relations
that are presented in Table~\ref{tb:allen}, and let
$2^{\A}$ denote the 8192
binary relations that can be formed by
taking unions of relations in $\A$.
The complexity of CSP$(\Gamma)$ is known for every
$\Gamma \subseteq 2^{\A}$~\cite{Krokhin:etal:jacm2003} and in each case
CSP$(X)$ is either polynomial-time solvable or NP-complete. In particular, CSP$(\A)$ is in P.
When considering subsets $\Gamma \subseteq \A$,
note that any constraint $x {\sf ri} y$ is equivalent to $y \rr x$
for $\rr \in \{\p, \m, \o, \d, \s, \f\}$,
so we may assume that $\rr \in \Gamma$ if and only if ${\sf ri} \in \Gamma$.
Furthermore, for the remainder of the paper, we may assume $\A = \{\p,\m,\o,\d,\s,\f,\e\}$.

When studying \textsc{MinCSP} and its parameterized complexity,
it is convenient to allow \emph{crisp} constraints,
i.e. constraints that cannot be deleted.
Formally, for a language $\Gamma \subseteq \A$ 
and a relation $\rr \in \Gamma$, we say that 
\emph{$\Gamma$ supports crisp $\rr$-constraints}
if, for every value of the parameter $k \in \naturals$,
we can construct an instance $\III_\rr$ of \mincsp{\Gamma}
with primary variables $x,y \in V(\III_\rr)$
such that the constraint $x \rr y$ is equivalent to $\III_\rr - X$ 
for all $X \subseteq C(\III)$ such that $|X| \leq k$.
Then, if we want to enforce a constraint $x \rr y$ in
an instance of \mincsp{\Gamma}, we can use $\III_\rr$ with fresh variables $V(\III) \setminus \{x,y\}$ in its place. Straightforward reasoning about interval constraints readily shows that every $\rr \in \A$ supports crisp constraints.

\fi

\iflong
\section{Classification Result} 
\fi 

\ifshort
\section{Overview}
\fi

\label{sec:classification}

In this section we prove the dichotomy theorem for the parameterized complexity of $\mincsp{\Gamma}$
for every subset $\Gamma \subseteq \A$ of interval relations.
We defer the proofs of the main technical details to the following
parts of the paper, while keeping this section relatively high-level
and focusing on the proof structure.
\ifshort
We also discuss constant-factor approximation algorithms for $\mincsp{\A}$.
\fi
Some observations reduce
the number of subsets of relations that 
we need to consider in the classification.
For the first one, we need a simplified definition of \emph{implementations}.
More general definitions are used in e.g.~\cite{Khanna:etal:sicomp2000}~and~\cite{kim2021solving}.

\iflong
\begin{definition}
  Let $\Gamma$ be a constraint language and $\rr$ be a binary relation over the same domain.
  A \emph{(simple) implementation} of a relation $\rr$ in $\Gamma$ 
  is an instance $\CCC_\rr$ of $\csp{\Gamma}$ with
  primary variables $x_1,x_2$ and, possibly, some auxiliary variables 
  $y_1,\dots,y_\ell$ such that:
  \begin{itemize}
    \item if an assignment $\varphi$ satisfies $\CCC_\rr$, 
    then it satisfies the constraint $x_1 \rr x_2$;
    \item if an assignment $\varphi$ does not satisfy $x_1 \rr x_2$,
    then it can be extended to the auxiliary variables $y_1,\dots,y_\ell$
    so that all but one constraint in $\CCC_\rr$ are satisfied. 
  \end{itemize}
  In this case we say that $\Gamma$ \emph{implements} $\rr$.
\end{definition}
\fi

\ifshort
\begin{definition}
  Let $\Gamma$ be a constraint language and $\rr$ be a binary relation over the same domain.
  A \emph{(simple) implementation} of a relation $\rr$ in $\Gamma$ 
  is an instance $\CCC_\rr$ of $\csp{\Gamma}$ with
  primary variables $x_1,x_2$ and, possibly, some auxiliary variables 
  $y_1,\dots,y_\ell$ such that (1)
  if an assignment $\varphi$ satisfies $\CCC_\rr$, 
    then it satisfies the constraint $x_1 \rr x_2$, and
    (2) if an assignment $\varphi$ does not satisfy $x_1 \rr x_2$,
    then it can be extended to the variables $y_1,\dots,y_\ell$
    so that all but one constraint in $\CCC_\rr$ are satisfied. 
  In this case we say that $\Gamma$ \emph{implements} $\rr$.
\end{definition}
\fi

\iflong
Intuitively, we can replace every occurrence of a constraint
$x \rr y$ with its implementation in $\Gamma$ while preserving
the cost of making the instance consistent.
This intuition is made precise in the following lemma.
\fi
\ifshort
Intuitively, we can replace every occurrence of a constraint
$x \rr y$ with its implementation in $\Gamma$ while preserving
the cost of making the instance consistent.
This intuition is made precise in the following lemma, and identifying the two
implementations in Lemma~\ref{lem:implementations} is left to the reader.
\fi

\begin{lemma}[Proposition~5.2~in~\cite{kim2021solving}]
\label{lem:impl-reduction}
  Let $\Gamma$ be a constraint language that implements a relation $\rr$.
  If $\mincsp{\Gamma}$ is in \FPT, then so is $\mincsp{\Gamma \cup \{\rr\}}$.
  If $\mincsp{\Gamma \cup \{\rr\}}$ is \W{1}-hard, then so is $\mincsp{\Gamma}$.
\end{lemma}

\iflong
\begin{figure}
  \centering
  \begin{tikzpicture}
  \def \offset {0.1};
  \coordinate (Lx1) at (0,2);
  \coordinate (Rx1) at (1.5,2);
  \coordinate (Ly)  at (1.5,1);
  \coordinate (Ry)  at (3,1);
  \coordinate (Lx2) at (3,0);
  \coordinate (Rx2) at (4.5,0);

  \coordinate (Mx1) at (0.75,2);
  \coordinate (Mx2) at (2.25,1);
  \coordinate (My)  at (3.75,0);

  \draw (Lx1) -- (Rx1);
  \node[anchor=south] at (Mx1) {$x_1$};
  \draw ($ (Lx1) + (0, \offset) $) -- ($ (Lx1) - (0, \offset) $);
  \draw ($ (Rx1) + (0, \offset) $) -- ($ (Rx1) - (0, \offset) $);
  
  \draw (Lx2) -- (Rx2);
  \node[anchor=south] at (Mx2) {$y$};
  \draw ($ (Lx2) + (0, \offset) $) -- ($ (Lx2) - (0, \offset) $);
  \draw ($ (Rx2) + (0, \offset) $) -- ($ (Rx2) - (0, \offset) $);

  \draw[red] (Ly) -- (Ry);
  \node[anchor=south] at (My) {$x_2$};
  \draw[red] ($ (Ly) + (0, \offset) $) -- ($ (Ly) - (0, \offset) $);
  \draw[red] ($ (Ry) + (0, \offset) $) -- ($ (Ry) - (0, \offset) $);

  \draw[dashed] (Rx1) -- (Ly);
  \draw[dashed] (Lx2) -- (Ry);
\end{tikzpicture}
  \hfill
  \begin{tikzpicture}
  \def \offset {0.1};
  \coordinate (Lx1) at (1,2);
  \coordinate (Rx1) at (3,2);
  \coordinate (Ly)  at (0,1);
  \coordinate (Ry)  at (3,1);
  \coordinate (Lx2) at (0,0);
  \coordinate (Rx2) at (4,0);

  \coordinate (Mx1) at (2,2);
  \coordinate (My)  at (2,1);
  \coordinate (Mx2) at (2,0);

  \draw (Lx1) -- (Rx1);
  \node[anchor=south] at (Mx1) {$x_1$};
  \draw ($ (Lx1) + (0, \offset) $) -- ($ (Lx1) - (0, \offset) $);
  \draw ($ (Rx1) + (0, \offset) $) -- ($ (Rx1) - (0, \offset) $);
  
  \draw (Lx2) -- (Rx2);
  \node[anchor=south] at (Mx2) {$x_2$};
  \draw ($ (Lx2) + (0, \offset) $) -- ($ (Lx2) - (0, \offset) $);
  \draw ($ (Rx2) + (0, \offset) $) -- ($ (Rx2) - (0, \offset) $);

  \draw[red] (Ly) -- (Ry);
  \node[anchor=south] at (My) {$y$};
  \draw[red] ($ (Ly) + (0, \offset) $) -- ($ (Ly) - (0, \offset) $);
  \draw[red] ($ (Ry) + (0, \offset) $) -- ($ (Ry) - (0, \offset) $);

  \draw[dashed] (Lx2) -- (Ly);
  \draw[dashed] (Rx1) -- (Ry);
\end{tikzpicture}  
  \hfill
  \begin{tikzpicture}
  \def \offset {0.1};
  \coordinate (Lx1) at (0,2);
  \coordinate (Rx1) at (3,2);
  \coordinate (Ly)  at (1,1);
  \coordinate (Ry)  at (3,1);
  \coordinate (Lx2) at (1,0);
  \coordinate (Rx2) at (4,0);

  \coordinate (Mx1) at (2,2);
  \coordinate (My)  at (2,1);
  \coordinate (Mx2) at (2,0);

  \draw (Lx1) -- (Rx1);
  \node[anchor=south] at (Mx1) {$x_1$};
  \draw ($ (Lx1) + (0, \offset) $) -- ($ (Lx1) - (0, \offset) $);
  \draw ($ (Rx1) + (0, \offset) $) -- ($ (Rx1) - (0, \offset) $);
  
  \draw (Lx2) -- (Rx2);
  \node[anchor=south] at (Mx2) {$x_2$};
  \draw ($ (Lx2) + (0, \offset) $) -- ($ (Lx2) - (0, \offset) $);
  \draw ($ (Rx2) + (0, \offset) $) -- ($ (Rx2) - (0, \offset) $);

  \draw[red] (Ly) -- (Ry);
  \node[anchor=south] at (My) {$y$};
  \draw[red] ($ (Ly) + (0, \offset) $) -- ($ (Ly) - (0, \offset) $);
  \draw[red] ($ (Ry) + (0, \offset) $) -- ($ (Ry) - (0, \offset) $);

  \draw[dashed] (Lx2) -- (Ly);
  \draw[dashed] (Rx1) -- (Ry);
\end{tikzpicture}  

  \caption{Implementations of the relations $\p$, $\d$ and $\o$ from Lemma~\ref{lem:implementations}.}
  \label{fig:implementations}
\end{figure}
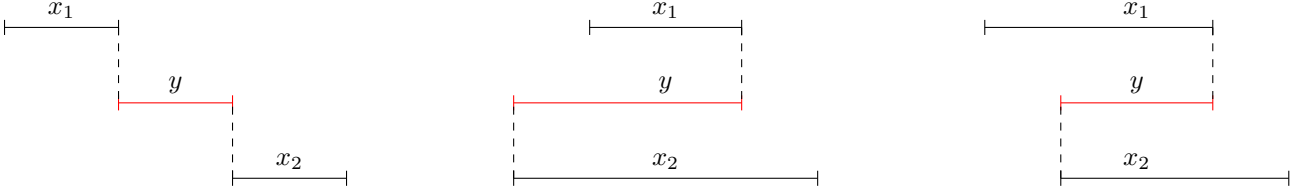
\fi

\ifshort
\begin{lemma}[Implementations] \label{lem:implementations}
  Let $\Gamma \subseteq \A$ be a subset of interval relations.
  If $\Gamma$ contains $\m$, then $\Gamma$ implements $\p$, and
    if $\Gamma$ contains $\f$ and $\s$, then $\Gamma$ implements $\d$ and $\o$. 
\end{lemma}
\fi

\iflong
\begin{lemma}[Implementations] \label{lem:implementations}
  Let $\Gamma \subseteq \A$ be a subset of interval relations.
  \begin{enumerate}
    \item If $\Gamma$ contains $\m$, then $\Gamma$ implements $\p$.
    \item If $\Gamma$ contains $\f$ and $\s$, then $\Gamma$ implements $\d$ and $\o$. 
  \end{enumerate}
\end{lemma}
\begin{proof}
Let
$\CCC_\p = \{ x_1 \{ \m \} y, y \{ \m \} x_2 \}$,
$\CCC_\d = \{ x_1 \{ \f \} y, y \{ \s \} x_2 \}$ and
$\CCC_\o = \{ y \{ \f \} x_1, y \{ \s \} x_2 \}$.
See Figure~\ref{fig:implementations} for illustrations.
One direction of the correctness proof is the same for all three cases:
any assignment to $x_1$ and $x_2$ can be extended to $y$ 
so that it satisfies the constraint between $x_1$ and $y$
regardless of the interval assigned to $x_2$.
For the other direction, recall the definitions of the interval relations.
An assignment satisfying $\CCC_\p$ satisfies
$x_1^+ = y^- < y^+ = x_2^-$, i.e. $x_1^+ < x_2^-$ and so $x_1 \{\p\} x_2$ holds.
An assignment satisfying $\CCC_\d$ satisfies
$x_1^- > y^- = x_2^-$ and $x_1^+ = y^+ < x_2^+$, 
i.e. $x_1^- > x_2^- \land x_1^+ < x_2^+$ and so $x_1 \{\d\} x_2$ holds.
Finally, an assignment satisfying $\CCC_\o$ satisfies
$x_1^- < y^- = x_2^-$, $x_2^- = y^- < y^+ = x_1^+$, and $x_1^+ = y^+ < x_2^+$,
i.e. $x_1^- < x_2^- < x_1^+ < x_2^+$ and so $x_1 \{\o\} x_2$ holds.
\end{proof}
\fi

Another observation utilizes the symmetry of interval relations.
By switching the left and the right endpoints of all intervals in 
an instance $\III$ of $\mincsp{\A}$ and then negating their values,
we obtain a \emph{reversed instance} $\III^R$.
\iflong
This way, in every constraint of $\III$, all relations except $\d$, $\s$ and $\f$ 
are mapped to their inverses, $\d$ stays unchanged,
while $\s$ and $\f$ are switched.
\fi
Formally, instance $\III^R$ of $\csp{\A}$ has the same set of variables as $\III$, and
contains a constraint $u f(\rr) v$ for every $u \rr v$ in $C(\III)$,
where $f : \A \to \A$ is defined as
$f(\rr) = {\sf ri}$ for $\rr \in \{\m, \p, \o\}$,
$f(\e) = \e$, $f(\d) = \d$, $f(\s) = \f$ and $f(\f) = \s$.

\begin{lemma}[Lemma~4.2~of~\cite{Krokhin:etal:jacm2003}]
  \label{lem:reversal}
  An instance $\III$ of $\csp{\A}$ is satisfiable if and only if
  the reversed instance $\III^R$ is satisfiable.
\end{lemma}

\iflong
For example, this implies that obstructions for satisfiability
of $\csp{\s, \p}$ are the same as those for $\csp{\f, \p}$
up to relabelling $\s$-arcs into $\f$-arcs, as shown in Lemma~\ref{lem:bad-cycles}.
\fi
\ifshort
\fi

To obtain our results, we use combinatorial tools and 
represent an instance $\III$ of \csp{\A} as
an \emph{arc-labelled mixed graph} $G_\III$,
i.e. a graph that contains edges for symmetric constraints
and labelled arcs for asymmetric ones.
More precisely, the graph $G_\III$ is obtained by introducing
all variables of $\III$ as vertices,
adding directed arcs $(u, v)$ labelled with $\rr \in \A \setminus \{\equiv\}$
for every constraint $u \rr v$ in $C(\III)$,
and undirected edges $uv$ for every constraint $u \equiv v$ in $\III$.
Note that $G_\III$ may have parallel arcs with different labels and may contain loops.
The undirected graph underlying $G_\III$ is called
the \emph{primal graph} of $\III$; we allow the primal graph to contain loops and parallel edges (in both cases, this will mean the primal graph contains a cycle).
The advantage of the graph representation is 
supported by the following lemma:

\begin{restatable}[Cycles]{lemma}{cycleslemma}
\label{lem:cycles}
  Let $\III$ be an inclusion-wise minimal unsatisfiable instance of \csp{\A}
  (i.e. removing any constraint of~$\III$ results in a satisfiable instance).
  Then the primal graph of~$\III$ is a cycle.
\end{restatable}

\iflong
The proof of the lemma is deferred to Section~\ref{sec:cycles}.
\fi
\ifshort
A sketch of the proof of the lemma is deferred to Section~\ref{sec:cycles}.
\fi
From the combinatorial point of view,
minimal unsatisfiable instances are 
\emph{bad cycles} in the labelled graph.
For example, in $\mincsp{\p}$, the bad cycles are all directed cycles.
For $\mincsp{\p, \e}$, the bad cycles contain at least one $\p$-arc
and all $\p$-arcs in the same direction.
Thus, $\mincsp{\Gamma}$ can now be cast as 
a certain feedback edge set problem --
our goal is to find a set of $k$ edges in the primal graph
that intersects all bad cycles.
\iflong
Understanding the complexity of $\mincsp{\Gamma}$ 
for a particular subset of interval relations $\Gamma \subseteq \A$
requires characterizing the bad cycles explicitly.
\fi
\iflong
We present such a characterization for several cases below,
and refer to Section~\ref{sec:cycles} for the proof.
As it turns out, these are sufficient for proving the classification theorem.

\begin{restatable}[Bad Cycles]{lemma}{badcycleslemma}
\label{lem:bad-cycles}
  Let $\III$ be an instance of $\csp{\rr_1, \rr_2}$
  for some $\rr_1, \rr_2 \in \A$,
  and consider the arc-labelled mixed graph $G_\III$.
  Then $\III$ is satisfiable if and only if $G_\III$ does not contain
  any bad cycles described below.
  \begin{enumerate}
    \item If $\rr_1 = \d$ and $\rr_2 = \p$, then the bad cycles are
    cycles with $\p$-arcs in the same direction and
    no $\d$-arcs meeting head-to-head. 
    \item\label{case:bad-cycles-do}If $\rr_1 = \d$ and $\rr_2 = \o$, then the bad cycles are
    cycles with all $\d$-arcs in the same direction and 
    all $\o$-arcs in the same direction 
    (the direction of the $\d$-arcs may differ from that of the $\o$-arcs).
    \item If $\rr_1 = \o$ and $\rr_2 = \p$, then the bad cycles are
    \begin{itemize}
    \item directed cycles of $\o$-arcs and 
    \item cycles with all $\p$-arcs in the forward direction, with
    every consecutive pair of $\o$-arcs in the reverse direction
    separated by a $\p$-arc (this case includes directed cycles of $\p$-arcs).
    \end{itemize}
    \item If $\rr_1 \in \{\f,\s\}$ and $\rr_2 \in \{\d, \o, \p\}$, then the bad cycles are
    \begin{itemize}
    \item directed cycles of $\rr_1$-arcs and
    \item cycles with at least one $\rr_2$-arc and all $\rr_2$-arcs in the same direction
    (and $\rr_1$-arcs directed arbitrarily).  
    \end{itemize}
  \end{enumerate}
\end{restatable}
\setcounter{specialtheorem}{\thetheorem} %
\fi
\ifshort
We present such a characterization for several cases 
in Section~\ref{sec:cycles}.
\fi

\ifshort

Our algorithmic results can be summarised as follows.

\begin{lemma} \label{lem:mp-fpt}
  $\mincsp{\m, \p}$ and $\mincsp{\rr_1, \rr_2, \e}$ are in \FPT for 
  $\rr_1 \in \{\s, \f\}$ and $\rr_2 \in \{ \p, \o, \d \}$.
\end{lemma}

The algorithm for $\mincsp{\m, \p}$ is obtained using a simple reduction to
\textsc{Subset Directed Feedback Arc Set}.

\pbDefP{Subset Directed Feedback Arc Set (\subdfas)}
{A directed graph $G$, a subset of red arcs $R \subseteq A(G)$, 
and an integer $k$.}
{$k$.}
{Is there a subset $Z \subseteq A(G)$ of size at most $k$ 
such that $G - Z$ contains no cycles with at least one red arc?}

Chitnis et al.~\cite{chitnis2015directed} have proved that
  \subdfas is solvable in $O^*(2^{O(k^3)})$ time.
The algorithm for the remaining cases is
more complicated and relies on the bad cycle characterization in
Section~\ref{sec:cycles} 
and a sophisticated modification
of the algorithm for \subdfas from~\cite{kim2022weighted}.

\fi

\iflong
Our first algorithmic result is based on a reduction to the
\textsc{Subset Directed Feedback Arc Set} (\subdfas) problem.
The formal definition of \subdfas and the proof 
can be found in Section~\ref{sec:fpt-algs}.

\begin{lemma} \label{lem:mp-fpt}
  $\mincsp{\m, \p}$ is in \FPT.
\end{lemma}

While Lemma~\ref{lem:implementations} implies that we can replace
a $\p$-constraint with an implementation
using $\m$ constraints in the above case, such an implementation is not possible for our second algorithmic case.
The latter relies on 
the much more complicated characterization in Lemma~\ref{lem:bad-cycles}, 
the symmetry of $\s$ and $\f$ implied by Lemma~\ref{lem:reversal}, 
and a sophisticated modification
of the algorithm for \subdfas from~\cite{kim2022weighted}.
The algorithm is based on the novel flow augmentation technique~\cite{Kim:etal:stoc2022}.
See Section~\ref{sec:fpt-algs} for the details.

\begin{lemma} \label{lem:sf-pod-fpt}
  $\mincsp{\rr_1, \rr_2, \e}$ is in \FPT for 
  $\rr_1 \in \{\s, \f\}$ and $\rr_2 \in \{ \p, \o, \d \}$.
\end{lemma}

\fi

\iflong
For the negative results, we start by proving \W{1}-hardness
for paired and simultaneous graph cut problems in 
Sections~\ref{ssec:me-ms-hard}~and~\ref{ssec:do-dp-op-hard}, respectively.
In paired problems, the input consists of two graphs
together with a pairing of their edges.
The goal is to compute cuts in both graphs 
using at most $k$ pairs.
While the problems on individual graphs are solvable in polynomial time,
the pairing requirement leads to \W{1}-hardness.
For certain $\Gamma \subseteq \A$, 
paired problems reduce to $\mincsp{\Gamma}$.
For intuition, consider a constraint $x \e y$.
If we consider the left and the right endpoints separately,
then $\e$ implies two equalities: $x^- = y^-$ and $x^+ = y^+$.
Together with another relation (e.g. $\m$), 
this double-equality relation can be used to encode
the pairing of the edges of two graphs
(namely, the left-endpoint graph and the right-endpoint graph).
We note that the double-equality relation is also
the cornerstone of all hardness results in the parameterized
complexity classification of \textsc{Boolean MinCSP}~\cite{Kim:etal:FA3}.
This intuition is crystallized in Section~\ref{ssec:me-ms-hard}
in the proof of the following lemma.
\fi

\ifshort
For the negative results, we start by proving \W{1}-hardness
for certain paired and simultaneous graph cut problems, and we
identify $\Gamma \subseteq \A$ such that 
paired or simultaneous problems reduce to $\mincsp{\Gamma}$.
For intuition, consider a constraint $x \e y$.
If we consider the left and the right endpoints separately,
then $\e$ implies two equalities: $x^- = y^-$ and $x^+ = y^+$.
Together with another relation (e.g. $\m$), 
this double-equality relation can be used to encode
the pairing of the edges of two graphs
(namely, the left-endpoint graph and the right-endpoint graph).
We note that the double-equality relation is also
the cornerstone of all hardness results in the parameterized
complexity classification of \textsc{Boolean MinCSP}~\cite{Kim:etal:FA3}.
Lemma~\ref{lem:me-ms-mf-hard} is based on paired problems
and Lemma~\ref{lem:do-po-dp-hard} is based on simultaneous problems.
\fi

\begin{lemma} \label{lem:me-ms-mf-hard}
  \mincsp{\m, \e}, \mincsp{\m, \s} and \mincsp{\m, \f} are \textnormal{\W{1}}-hard.
\end{lemma}

\iflong
In simultaneous cut problems, the input also consists of two graphs.
Here, the graphs share some arcs/edges that 
can be deleted simultaneously at unit cost.
The goal is to compute a set of $k$ arcs/edges
that is a cut in both graphs.
As for the paired problems, computing cuts for individual graphs
is feasible in polynomial time.
However, the possibility of choosing
common arcs/edges while deleting them at a unit cost
correlates the choices in one graph with the choices in the other,
and leads to \W{1}-hardness.
Starting from this problem, we prove the following
lemma in Section~\ref{ssec:sim-graph}.
\fi

\begin{lemma} \label{lem:do-po-dp-hard}
  \mincsp{\d, \o}, \mincsp{\p, \o} and
  \mincsp{\d, \p} are \textnormal{\W{1}}-hard.
\end{lemma}
\iflong
\begin{table}[tb]
  \centering
  \begin{tabular}{| c | c c c c c c |}
    \hline
    & \m & \p & \s & \f & \d & \o \\
    \hline
    \e & \Wh  & \FPT & \FPT & \FPT & \FPT & \FPT \\
    \m &      & \FPT & \Wh  & \Wh  & \Wh  & \Wh  \\
    \p &      &      & \FPT & \FPT & \Wh  & \Wh  \\
    \s &      &      &      &  \Wh & \FPT & \FPT \\
    \f &      &      &      &      & \FPT & \FPT \\
    \d &      &      &      &      &      & \Wh  \\
    \hline
  \end{tabular}
\caption{Complexity of \textsc{MinCSP} with two interval relations.}
\label{tb:two-relations}
\end{table}    

\begin{figure}[tb]
  \centering
  \begin{tikzpicture}[%
scale=.8,
transform shape,
myvertex/.style={circ,scale=2}
]

\node[myvertex,label=above:{\m}] (m) at (0-0.866,1) {};
\node[myvertex,label=above:{\p}] (p) at (2-0.866,1) {};
\node[myvertex,label=above:{\f/\s}] (sf) at (3,2) {};
\node[myvertex,label=below:{\e}] (e) at (3,0) {};
\node[myvertex,label=above:{\o}] (o) at (5,2) {};
\node[myvertex,label=below:{\d}] (d) at (5,0) {};

\draw[thick] (e) -- (p);
\draw[thick] (e) -- (sf);
\draw[thick] (e) -- (o);
\draw[thick] (e) -- (d);
\draw[thick] (m) -- (p);
\draw[thick] (p) -- (sf);
\draw[thick] (sf) -- (d);
\draw[thick] (sf) -- (o);

\end{tikzpicture}
  \caption{Compatibility graph for interval relations. Two vertices $u$, $v$ are adjacent in this graph if $\mincsp{u,v}$ is in \FPT. The $\f/\s$ vertex represents either of these two relations.}
  \label{fig:compat-graph}
\end{figure}
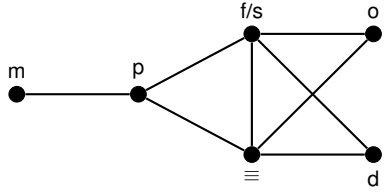
\fi
\ifshort
\begin{figure}[tb]
  \begin{minipage}{0.6\textwidth}
  \centering
  \small
  \begin{tabular}{| c | c c c c c c |}
    \hline
    & \m & \p & \s & \f & \d & \o \\
    \hline
    \e & \Wh  & \FPT & \FPT & \FPT & \FPT & \FPT \\
    \m &      & \FPT & \Wh  & \Wh  & \Wh  & \Wh  \\
    \p &      &      & \FPT & \FPT & \Wh  & \Wh  \\
    \s &      &      &      &  \Wh & \FPT & \FPT \\
    \f &      &      &      &      & \FPT & \FPT \\
    \d &      &      &      &      &      & \Wh  \\
    \hline
  \end{tabular}
  \end{minipage}%
  \begin{minipage}{0.4\textwidth}
    \centering
    \begin{tikzpicture}[%
scale=.8,
transform shape,
myvertex/.style={circ,scale=2}
]

\node[myvertex,label=above:{\m}] (m) at (0-0.866,1) {};
\node[myvertex,label=above:{\p}] (p) at (2-0.866,1) {};
\node[myvertex,label=above:{\f/\s}] (sf) at (3,2) {};
\node[myvertex,label=below:{\e}] (e) at (3,0) {};
\node[myvertex,label=above:{\o}] (o) at (5,2) {};
\node[myvertex,label=below:{\d}] (d) at (5,0) {};

\draw[thick] (e) -- (p);
\draw[thick] (e) -- (sf);
\draw[thick] (e) -- (o);
\draw[thick] (e) -- (d);
\draw[thick] (m) -- (p);
\draw[thick] (p) -- (sf);
\draw[thick] (sf) -- (d);
\draw[thick] (sf) -- (o);

\end{tikzpicture}
    \end{minipage}
\caption{\textbf{Left:} Complexity of \textsc{MinCSP} with two interval relations. \textbf{Right:} Compatibility graph for interval relations. Two vertices $u$, $v$ are adjacent in this graph if $\mincsp{u,v}$ is in \FPT. The $\f/\s$ vertex represents either of these two relations.}
\label{tb:two-relations-comp}
\end{figure}
\fi

Combining all results above, we are ready to present
the full classification.

\begin{theorem}[Full classification] \label{thm:classification}
  Let $\Gamma \subseteq \A$ be a subset of interval relations.
  Then $\mincsp{\Gamma}$ is in \FPT if
  $\Gamma \subseteq \{\m, \p\}$ or $\Gamma \subseteq \{\rr_1, \rr_2, \e\}$ 
  for any $\rr_1 \in \{\s, \f\}$ and $\rr_2 \in \{\p, \o, \d\}$, and 
  \textnormal{\W{1}}-hard otherwise.
\end{theorem}
\begin{proof}
\iflong
The tractable cases follow from 
Lemmas~\ref{lem:mp-fpt}~and~\ref{lem:sf-pod-fpt}.
\fi
\ifshort
The tractable cases follow from 
Lemma~\ref{lem:mp-fpt}.
\fi
For the full classification, we consider subsets of relations
$\Gamma \subseteq \A$ in order of increasing cardinality.
If $|\Gamma| = 1$, then $\Gamma$ is contained
in one of the maximal tractable subsets.
Now consider $|\Gamma| = 2$.
The classification for this case is summarized in 
\iflong Table~\ref{tb:two-relations} \fi\ifshort Figure~\ref{tb:two-relations-comp} \fi.
We consider the table row by row.
We only spell out \W{1}-hardness arguments
since tractability follows in all cases 
by inclusion in one of the
maximal tractable subsets.

\ifshort
\smallskip

\noindent
(1) $\mincsp{\e, \m}$ is \W{1}-hard by Lemma~\ref{lem:me-ms-mf-hard}.
  $\mincsp{\m, \s}$ and $\mincsp{\m, \f}$ are \W{1}-hard by 
  Lemma~\ref{lem:me-ms-mf-hard}; 
  \W{1}-hardness of $\mincsp{\m, \o}$ and $\mincsp{\m, \d}$ follows 
  from the hardness of
  $\mincsp{\p, \o}$ and $\mincsp{\p, \d}$ (Lemma~\ref{lem:do-po-dp-hard}) by using Lemma~\ref{lem:impl-reduction}
  combined with the fact that $\m$ implements $\p$ (Lemma~\ref{lem:implementations}).

  \smallskip

  \noindent
(2)   $\mincsp{\p, \o}$ and $\mincsp{\p, \d}$ are \W{1}-hard by 
  Lemma~\ref{lem:do-po-dp-hard}.

\smallskip

\noindent
(3)  \W{1}-hardness of $\mincsp{\s, \f}$ follows
  from the hardness of $\mincsp{\d, \o}$ (Lemma~\ref{lem:do-po-dp-hard}) by using Lemma~\ref{lem:impl-reduction}
  combined with the fact that $\{\s, \f\}$ implements $\d$ and $\o$ (Lemma~\ref{lem:implementations}).

  \smallskip

  \noindent
(4)  $\mincsp{\d, \o}$ is \W{1}-hard by Lemma~\ref{lem:do-po-dp-hard}.

\smallskip
\fi
\iflong
\begin{itemize}
\item $\mincsp{\e, \m}$ is \W{1}-hard by Lemma~\ref{lem:me-ms-mf-hard}.
  $\mincsp{\m, \s}$ and $\mincsp{\m, \f}$ are \W{1}-hard by 
  Lemma~\ref{lem:me-ms-mf-hard}; 
  \W{1}-hardness of $\mincsp{\m, \o}$ and $\mincsp{\m, \d}$ follows 
  from the hardness of
  $\mincsp{\p, \o}$ and $\mincsp{\p, \d}$ (Lemma~\ref{lem:do-po-dp-hard}) by using Lemma~\ref{lem:impl-reduction}
  combined with the fact that $\m$ implements $\p$ (Lemma~\ref{lem:implementations}).

\item   $\mincsp{\p, \o}$ and $\mincsp{\p, \d}$ are \W{1}-hard by 
  Lemma~\ref{lem:do-po-dp-hard}.

\item  \W{1}-hardness of $\mincsp{\s, \f}$ follows
  from the hardness of $\mincsp{\d, \o}$ (Lemma~\ref{lem:do-po-dp-hard}) by using Lemma~\ref{lem:impl-reduction}
  combined with the fact that $\{\s, \f\}$ implements $\d$ and $\o$ (Lemma~\ref{lem:implementations}).

\item  $\mincsp{\d, \o}$ is \W{1}-hard by Lemma~\ref{lem:do-po-dp-hard}.

\end{itemize}

\fi

\noindent
We proceed to subsets $\Gamma$ of size $3$.
It is useful to consider the compatibility graph of relations,
i.e. a graph with a vertex for every relation in $\A$
and an edge for every pair of relations $\rr_1, \rr_2$ such that
$\mincsp{\rr_1, \rr_2}$ is in \FPT.
See Figure\iflong~\ref{fig:compat-graph}\fi\ifshort~\ref{tb:two-relations-comp}\fi for an illustration.
Since $\mincsp{\s,\f}$ is \W{1}-hard, $\s$ and $\f$ 
cannot both be included in a tractable subset of relations.
By symmetry (Lemma~\ref{lem:reversal}), we can identify
$\s$ and $\f$, thus simplifying the drawing.
Note that if $\Gamma$ contains a pair $\rr_1, \rr_2$
that are not connected by an edge, then 
$\mincsp{\Gamma}$ is \W{1}-hard.
Thus, candidates for tractable constraint languages 
of size $3$ correspond to triangles in the graph.
By examining the graph, we find that there are exactly
$3$ triangles, which correspond to six constraint
languages (recall that $\s$ and $\f$ are identified).
\iflong
\textsc{MinCSP} for all these languages is
in \FPT by Lemma~\ref{lem:sf-pod-fpt}.
\fi
\ifshort
\textsc{MinCSP} for all these languages is
in \FPT by Lemma~\ref{lem:mp-fpt}.
\fi
Since the compatibility graph does not contain
any cliques of size $4$, every $\Gamma \subseteq \A$
with $|\Gamma| \geq 4$ contains an incompatible 
pair of relations, therefore $\mincsp{\Gamma}$
is \W{1}-hard if $|\Gamma| \geq 4$.
This completes the proof.
\end{proof}

\ifshort
W[1]-hardness of \mincsp{\A} motivates us to look at approximation
algorithms for this problem.
Our first observation is that \mincsp{\rr} for any $\rr \in \A \setminus \{\e\}$
is NP-hard to approximate within any constant under the Unique Games Conjecture~(UGC)
of Khot~\cite{khot2002power}.
This follows by combining two facts:
Lemma~\ref{lem:bad-cycles}, which implies
that an instance $\III$ of \csp{\rr} is consistent if and only if
the arc-labeled graph $G_\III$ is acyclic,
and Corollary~1.2~in~\cite{guruswami2008beating} 
which states that under the UGC, 
\textsc{Directed Feedback Arc Set (DFAS)} is
NP-hard to approximate within any constant~\cite{guruswami2008beating}.
If we allow the approximation algorithm to run in fpt time,
then we obtain the following result.
\begin{theorem} \label{thm:fpt-approx-result}
  \mincsp{\A} is $2$-approximable in $O^*(2^{O(k^3)})$ time
  and $4$-approximable in $O^*(2^{O(k)})$ time.
\end{theorem}

\begin{proof}[Proof sketch]
We obtain the algorithms
by reducing the problem to \subdfas 
and invoking the exact algorithm of~\cite{cygan2015parameterized}
and the faster $O^*(2^{O(k)})$ time $2$-approximation algorithm of~\cite{lokshtanov2021fpt}, respectively.
There are straightforward reductions from $\mincsp{<, =}$ to $\mincsp{\leq, =}$ to \subdfas so we focus on the reduction from $\mincsp{\A}$ to $\mincsp{<, =}$.
Let $(\III, k)$ be an instance of \mincsp{\A}.
Replace every constraint $x\{\o\}y$ by
its implementation in $\{\s,\f\}$
according to Lemma~\ref{lem:implementations}.
By Lemma~\ref{lem:impl-reduction}, this does not change
the cost of the instance.
Using Table~\ref{tb:allen}, we can rewrite all constraints of $\III'$
as conjunctions of two atomic constraints of the form $x < y$ and $x = y$.
Disregarding the pairing, let $S$ be the set of all atomic constraints.
Apply one of the \mincsp{<,=} algorithms to $(S, 2k)$.
On the one hand, deleting $k$ constraints from $\III'$
corresponds to deleting at most $2k$ constraints in $S$.
On the other hand, if there is $X \subseteq S$, $|X| \leq 2k$,
such that $S - X$ is consistent, define
the set of interval constraints $X'$ 
such that at least one of the defining
$\{<,=\}$-constraints is in $X$.
Noting that $\III - X'$ is consistent and $|X'| \leq |X| \leq 2k$
completes the proof.
\end{proof}
\fi

\section{Bad Cycles} \label{sec:cycles}

\iflong
In this section, we prove Lemmas~\ref{lem:cycles} and \ref{lem:bad-cycles}, i.e. we describe the minimal obstructions to satisfiability for certain subsets of~$\A$.
\fi
\ifshort
In this section, we sketch the proof of Lemma~\ref{lem:cycles} and describe the minimal obstructions to satisfiability for certain subsets of~$\A$, along with a brief sketch of why these are the minimal obstructions.
\fi
\iflong
To do this, we will use the {\em point algebra} (PA)~\cite{Vilain:Kautz:aaai86}. This a CSP with the rationals
${\mathbb Q}$ as the variable domain and the constraint language
$\{<,=,\leq,\neq\}$, where the relations are interpreted
over ${\mathbb Q}$ in the obvious way. This CSP is polynomial-time
solvable.
For the subsets of~$\A$ we consider, we can encode our problem instances as instances of point algebra that contain only $<$ and $=$ constraints (see also Table~\ref{tb:allen}).
It is easy to verify that if an instance~$\III$ of the point algebra contains only $<$ and $=$ constraints, then it is satisfiable if and only if its primal graph does not contain a cycle in which all $<$ relations are directed in the same way.
This allows us to show that the minimal unsatisfiable instances of $\csp{\A}$ also have primal graphs that are cycles (Lemma~\ref{lem:cycles}).
We then completely describe these unsatisfiable cycles for our chosen subsets of $\A$ in Lemma~\ref{lem:bad-cycles}.
\fi

\cycleslemma*

\ifshort
The proof of Lemma~\ref{lem:cycles} starts by taking a minimal unsatisfiable instance $\III$.
Using Table~\ref{tb:allen}, we write $\III$ as an instance $\III'$ of the {\em point algebra} (PA)~\cite{Vilain:Kautz:aaai86} CSP, which takes rationals ${\mathbb Q}$ as the variable domain and we use only the basic constraint language $\{<,=\}$, where the relations are interpreted in the obvious way.
This instance $\III'$ must contain a minimal unsatisfiable sub-instance $\III''$ of the point algebra, which has a cycle as its primal graph.
We then map the constraints in $\III''$ back to the constraints in $\III$ that implied them, and find that $\III$ must also have a cycle as its primal graph. 
\fi
\iflong
\begin{proof}
Let~$\III$ be an inclusion-wise minimal unsatisfiable instance of \csp{\A}.
We can write~$\III$ as an instance~$\III'$ of \csp{<,=}, a constraint satisfaction problem over the point algebra.
To do this, we set $V(\III') = \{I^-,I^+ \; | \; I \in V(\III)\}$ and $C(\III') = \{(I^- < I^+) \; | \; I \in V(\III)\}$ and then, for each constraint in $C(\III)$, we add the corresponding~$<$ and/or~$=$ constraints to $C(\III')$, as described in Table~\ref{tb:allen}.
Since~$\III$ is unsatisfiable, $\III'$ must also be unsatisfiable.
This means that for some $\ell \geq 1$, there must be a sequence of
variables $v_1,\ldots,v_\ell \in V(\III')$ and a sequence of
constraints $v_1\rr_1v_2, v_2\rr_2v_3, \ldots, v_{\ell-1}\rr_{\ell-1}v_\ell,
v_\ell\rr_\ell v_1 \in C(\III')$, where $\rr_i \in \{<,=\}$ for all~$i$
and~$\rr_i$ is~$<$ for at least one value of~$i$.
We may assume the sequences are chosen such that (1) $\ell$ is minimum and (2) the number of constraints in the sequence of the form $I^-<I^+$ is maximum for this value of~$\ell$.

Let~$\III''$ be the sub-instance of~$\III'$ with
$V(\III'')=\{v_1,\ldots,v_\ell\}$ and $C(\III'')=\{
v_1\rr_1v_2, v_2\rr_2v_3, \ldots ,v_{\ell-1}\rr_{\ell-1}v_\ell, v_\ell
\rr_\ell v_1\}$ and note that~$\III''$ is unsatisfiable.
By minimality of~$\ell$, the primal graph of~$\III''$ is a cycle and the~$v_i$ variables are pairwise distinct.
Next, suppose for contradiction, that there is a variable $I \in \III$ such that both $I^- \in V(\III'')$ and $I^+ \in V(\III'')$, but $(I^- < I^+) \notin C(\III'')$.
Without loss of generality, assume $v_1=I^-$ and $v_i=I^+$.
Replacing $v_1\rr_1v_2, v_2\rr_2v_3, \ldots ,v_{i-1}\rr_{i-1}v_i$ by $(I^- < I^+)$ in $C(\III'')$ yields another unsatisfiable instance, but with either a smaller value of~$\ell$ or more constraints of the form $(I^- < I^+)$.
Since $(I^- < I^+)$ is a constraint in~$\III'$, this contradicts the choice of~$\III''$.
Therefore if $I^- \in V(\III'')$ and $I^+ \in V(\III'')$, then $(I^- < I^+) \in C(\III'')$.

We now construct an instance~$\III'''$ of \csp{\A}.
We set $V(I''') = \{I \; | \; I^- \; \mathrm{or} \;  I^+ \in \{v_1,\ldots,v_\ell\}\}$.
Next, for each constraint in~$\III''$ that is not of the form $(I^-<I^+)$, we add a constraint from~$\III$ that implies this constraint.
Now if we translate the instance~$\III'''$ into the point algebra (like we did translating~$\III$ into~$\III'$), we obtain a superinstance of~$\III''$, so it follows that~$\III'''$ is unsatisfiable.

Furthermore, $\III'''$ is a sub-instance of~$\III$, so by minimality of~$\III$, it follows that $\III=\III'''$.
If the primal graph of~$\III'''$ is a cycle of length~$2$ or a loop, then we are done.
Otherwise, by construction, each variable in~$\III'''$ is contained in two constraints (which either imply two consecutive constraints in the cycle in~$\III''$, or imply two constraints separated by an $(I^-<I^+)$ constraint in~$\III''$) and the primal graph f of $\III'''$ is connected.
Therefore the primal graph of $\III'''$ is a cycle.
This completes the proof.
\end{proof}
\fi

\iflong
In the proof of the next lemma, we will use the following additional notation.
Suppose~$\varphi$ is an assignment for an instance~$\III$ of $\csp{\A}$ and $c \in \mathbb{Q}$.
For $v \in V(\III)$, we let $\varphi_{+c}(v)^-=\varphi(v)^-+c$ and $\varphi_{+c}(v)^-=\varphi(v)^-+c$ (with $\varphi_{-c}$ defined analogously).
Observe that~$\varphi_{+c}$ is a satisfying assignment for~$\III$ if and only if~$\varphi$ is.
Furthermore, let $\varphi_{max} = \max\{\varphi(v)^+ \; | \; v \in V(\III)\}$ and $\varphi_{min} = \min\{\varphi(v)^- \; | \; v \in V(\III)\}$.
\fi

\ifshort
\begin{restatable}[Bad Cycles]{lemma}{badcycleslemma}
\label{lem:bad-cycles}
  Let $\III$ be an instance of $\csp{\rr_1, \rr_2}$
  for some $\rr_1, \rr_2 \in \A$,
  and consider the arc-labelled mixed graph $G_\III$.
  Then $\III$ is satisfiable if and only if $G_\III$ does not contain
  any bad cycles described below.

\smallskip

\noindent
(1) If $\rr_1 = \d$ and $\rr_2 = \p$, then the bad cycles are
    cycles with $\p$-arcs in the same direction and
    no $\d$-arcs meeting head-to-head. 

\smallskip

\noindent    
(2) \label{case:bad-cycles-do}If $\rr_1 = \d$ and $\rr_2 = \o$, then the bad cycles are
    cycles with all $\d$-arcs in the same direction and 
    all $\o$-arcs in the same direction 
    (the direction of the $\d$-arcs may differ from that of the $\o$-arcs).

\smallskip

\noindent    
(3) If $\rr_1 = \o$ and $\rr_2 = \p$, then the bad cycles are
    (a) directed cycles of $\o$-arcs and 
    (b)  cycles with all $\p$-arcs in the forward direction, with
    every consecutive pair of $\o$-arcs in the reverse direction
    separated by a $\p$-arc (this case includes directed cycles of $\p$-arcs).

\smallskip

\noindent
(4) If $\rr_1 \in \{\f,\s\}$ and $\rr_2 \in \{\d, \o, \p\}$, then the bad cycles are
    (a) directed cycles of $\rr_1$-arcs and
    (b) cycles with at least one $\rr_2$-arc and all $\rr_2$-arcs in the same direction
    (and $\rr_1$-arcs directed arbitrarily).  
\end{restatable}
\setcounter{specialtheorem}{\thetheorem} %
\fi

\iflong
\badcycleslemma*

\begin{proof}
By Lemma~\ref{lem:cycles}, we only need to consider instances~$\III$ whose primal graph is a minimal bad cycle.
If~$G_\III$ is a loop with a label other than $\e$, then it is a bad cycle.
If~$G_\III$ contains two parallel arcs with different labels or two anti-parallel arcs with labels in $\{\d,\f,\m,\o,\p,\s\}$, then it is a bad cycle.

\begin{claim}\label{clm:oppposite-p-arcs}
If a cycle contains two $\p$-arcs in opposite directions, then it is satisfiable.
\end{claim}
\begin{proof}[Proof of claim]\renewcommand{\qedsymbol}{$\diamond$}
Let~$\III$ be such an instance and let~$\III'$ be the instance obtained by removing a forward $\p$-arc and a reverse $\p$-arc from~$G_{\III}$.
Note that~$G_{\III'}$ is disconnected, and let~$\III''$ and~$\III'''$ be the sub-instances of~$\III'$ that correspond to the two components of~$G_{\III'}$.
Without loss of generality, we may assume that the $\p$-arcs we removed are each from variables in~$V(\III'')$ to variables in~$V(\III''')$.
Since~$G_{\III''}$ and~$G_{\III'''}$ are paths, $\III''$ and~$\III'''$ are satisfiable by Lemma~\ref{lem:cycles}; let~$\varphi''$ and~$\varphi'''$, respectively, be satisfying assignments for these instances.
We now define an assignment~$\varphi$ for~$\III$ by setting $\varphi(v) = \varphi_{-(\varphi_{max}''+1)}''(v)$ if $v \in V(\III'')$ and $\varphi(v) = \varphi_{-(\varphi_{min}'''-1)}'''(v)$ if $v \in V(\III''')$.
Note that~$\varphi(v)^-$ and~$\varphi(v)^+$ are negative if $v \in V(\III'')$ and positive if $v \in V(\III''')$, so the two~$\p$ constraints we removed from~$\III$ are satisfied by~$\varphi$.
Furthermore, all the other constraints in~$\III$ appear in either~$\III''$ or~$\III'''$ and~$\varphi$ satisfies these by construction.
Therefore~$\varphi$ is a satisfying assignment for~$\III$.
\end{proof}

\begin{claim}\label{clm:d-arcs-head-to-head}
If a cycle contains two $\d$-arcs that meet head-to-head, then it is satisfiable.
\end{claim}
\begin{proof}[Proof of claim]\renewcommand{\qedsymbol}{$\diamond$}
Let~$\III$ be such an instance and let $x\d y$ and $z \d y$ be the two constraints corresponding to these arcs.
Let~$\III'$ be the instance obtained from~$\III$ by removing these two constraints, along with the variable~$y$.
The primal graph of~$\III'$ is a path, so~$\III$ has a satisfying assignment~$\varphi'$.
We set $\varphi(y)^-=\varphi'_{min}-1$, $\varphi(y)^+=\varphi'_{max}+1$ and $\varphi(v)=\varphi'(v)$ otherwise.
By construction, $\varphi$ satisfies the constraints $x\d y$ and $z \d y$, the constraint $\varphi(y)^- < \varphi(y)^+$ and all other constraints in~$\III$.
Therefore~$\varphi$ satisfies~$\III$.
\end{proof}

We now consider the four cases stated in the lemma.
\begin{enumerate}
\item If we have two $\p$-arcs in opposite directions, then the instance is satisfiable by Claim~\ref{clm:oppposite-p-arcs}.
If we have two $\d$-arcs meeting head-to-head, then the instance is satisfiable by Claim~\ref{clm:d-arcs-head-to-head}.
We claim that otherwise the instance~$\III$ is unsatisfiable.
Suppose, for contradiction, that the instance is satisfiable.
If there are no $\p$-arcs, then this would give a cycle of $\d$-arcs all in the same direction, which is clearly not satisfiable.
We may therefore assume there is at least one $\p$-arc.
Observe the each non-empty set of consecutive $\d$-arcs consists of a non-negative number of reverse $d$-arcs, followed by a non-negative number of forward $d$-arcs i.e. we have a sequence of constraints 
$$x_{k+1}\di x_{k+2}, \ldots, x_{k+i-1}\di x_{k+i}, x_{k+i} \d x_{k+i+1}, \ldots, x_{k+j-1} \d x_{k+j},$$ where $j\geq 1$.
This implies the point algebra constraints
$$x_{k+1}^-<\cdots<x_{k+i}^->\cdots > x_{k+j}^-$$ and 
$$x_{k+1}^+>\cdots>x_{k+i}^+<\cdots < x_{k+j}^+.$$
Since $x_{k+i}^-<x_{k+i}^+$, this means that $x_{k+i}^*=(x_{k+i}^-+x_{k+i}^+)/2$ is contained in every interval in $\{x_{k+1},\ldots,x_{k+j}\}$.
Now if $u \p v$, then for every $x \in [u^-,u^+]$ and every $y \in [v^-,v^+]$  we have $x<y$.
We therefore have a cyclic sequence of points in~$\mathbb{Q}$, where if $u \p v$ is a constraint in~$\III$, then the point corresponding to~$u$ is less than that corresponding to~$v$, and if $u \d v$ is a constraint in~$\III$, then the points corresponding to~$u$ and~$v$ are equal.
We therefore have a satisfied cycle of point algebra relations~$<$ and~$=$ with at least one~$<$ relation and all~$<$ relations in the same direction, a contradiction.
Therefore~$\III$ is indeed unsatisfiable.

\item
A constraint $u \d v$ implies that $u^- > v^-$ and $u^+ < v^+$.
A constraint $u \o v$ implies that $u^- < v^-$ and $u^+ < v^+$.
Thus a cycle of~$\d$ and~$\o$ constraints in which all of these constraints correspond to arcs in the same direction yields an unsatisfiable set of point algebra on $\{v^+ \; | \; v \in V(\III)\}$.
Similarly, a cycle of~$\d$ and~$\o$ constraints in which all of the~$\d$ constraints are in the forward direction, and all of the $\o$-arcs are in the reverse direction yields unsatisfiable set of point algebra on $\{v^- \; | \; v \in V(\III)\}$.
Therefore the described bad cycles really are indeed bad.

Suppose we have a cycle that is not of this form.
Then the set of implied constraints on pairs of variables in $\{v^- \; | \; v \in V(\III)\}$ has a satisfying assignment~$\psi^-$.
Similarly, the set of implied constraints on pairs of variables in $\{v^+ \; | \; v \in V(\III)\}$ has a satisfying assignment~$\psi^+$.
We may assume (by adding a fixed constant to the values) that~$\psi^-$ only takes negative values and~$\psi^+$ only takes positive values.
Now for $v \in V(\III)$, let $\varphi(v)^-=\psi^-(v)$ and $\varphi(v)^+=\psi^+(v)$.
Then~$\varphi$ satisfies all of the~$\d$ and $\o$ constraints in~$\III$, by construction.
Furthermore, for all $v \in V(\III)$, $\varphi(v)^-<0<\varphi(v)^+$, so~$\varphi$ is a satisfying assignment of~$\III$.

\item
Clearly, directed cycles of arcs of the same relation are bad.
If the cycle contains a forward $\p$-arc and a reverse $\p$-arc, then the instance is satisfiable by Claim~\ref{clm:oppposite-p-arcs}.
We may therefore restrict ourselves to looking at cycles in which all $\p$-arcs go in the forward direction and there is at least one $\o$-arc.
If the cycle contains only $\o$-arcs, then we are done by Case~\ref{case:bad-cycles-do}, so we may assume there is at least one $\p$-arc.
Now suppose that the cycle contains a forward $\p$-arc and two reverse $\o$-arcs that do not have a $\p$-arc between then, but may have a non-negative number of forward $\o$-arcs between them, say $x \oi y_1, y_1 \o y_2, \ldots, y_{k-1} \o y_k, y_k \oi z$, where $k \geq 1$.
Let~$\III'$ be the instance obtained from~$\III$ by deleting these constraints and the variables~$y_i$ and adding the constraint $x \pi z$.
Then~$\III'$ contains a forward and a reverse $\p$-arc, so it has a satisfying assignment~$\varphi'$ and note that $\varphi'(x)^-> \varphi'(z)^+$.
We let 
$$\varphi(y_i)^-=\frac{i \cdot \varphi'(z)^-+(k+1-i) \cdot \varphi'(z)^+}{k+1},$$ 
$$\varphi(y_i)^+=\frac{i \cdot \varphi'(x)^-+(k+1-i) \cdot \varphi'(x)^+}{k+1},$$ and let $\varphi(w)=\varphi'(w)$ otherwise.
Now $\varphi(y_i)^-<\varphi(y_{i+1})^-<\varphi(y_i)^+<\varphi(y_{i+1})^+$ for all $i \in [k-1]$.
Furthermore, $\varphi(y_1)^-<\varphi(x)^-<\varphi(y_1)^+<\varphi(x)^+$ and $\varphi(z)^-<\varphi(y_k)^-<\varphi(z)^+<\varphi(y_k)^+$.
Therefore~$\varphi$ satisfies $x \oi y_1, y_1 \o y_2, \ldots, y_{k-1} \o y_k, y_k \oi z$ and all the other constraints in~$\III$.
Therefore every bad cycle is of the form stated in the lemma.

Now suppose there is a cycle of the form stated in the lemma and assume, for contradiction, that it is satisfiable; note that the cycle must contain at least one $\p$-arc.
If there is a consecutive run of a non-negative number of forward $\o$-arcs, say $y_1 \o y_2, \ldots, y_{k-1}\o y_k$ for some $k \geq 1$, then $y_1^- \leq y_k^-$ and $y_1^+ \leq y_k^+$.
If there is a reverse $\o$-arc, say $x \oi y$, then $x^-<y^+$.
Therefore, if there is a consecutive run of $\o$-arcs, at most one of which is a reverse arc, from the variable~$x$ to the variable~$y$, then $x^-<y^+$.
If there is a consecutive run of forward $\p$-arcs from the variable~$x$ to the variable~$z$, say $x \p y_1, y_1 \p y_2, \ldots, y_{k-1} \p y_k, y_k \p z$ for some $k \geq 0$, then $x^+ < y^-$.
Therefore encoding~$\III$ in the point algebra yields an unsatisfiable instance, a contradiction.
This completes the proof for this case.

\item By Lemma~\ref{lem:reversal}, we may assume that $\rr_1=\s$ and we may replace $\d$ by $\di$ in the statement of the lemma.
We first show that the cycles described in the lemma are indeed bad.
Clearly, a directed cycle of $\s$-arcs is not satisfiable.
Now suppose that the $\rr_2$-arcs are all directed in the same direction.
We encode the instance~$\III$ as an instance~$\III'$ of the point algebra and we claim that~$\III'$ is not satisfiable.
Indeed, each $u\s v$ constraint in~$\III$ yields a $u^- = v^-$ constraint in~$\III'$.
If $u \di v$ or $u \o v$ is a constraint in~$\III$, then $u^-<v^-$ is a constraint in $\III'$.
If $u \p v$ is a constraint in~$\III$, then $u^+ < v^-$ is a constraint in~$\III'$; since $u^-<u^+$ is also a constraint in~$\III'$, this implies that $(u^-<v^-)$ must be satisfied by any satisfying assignment for~$\III'$.
We conclude that in each case, if the cycle is as described in the lemma, then the instance~$\III'$ is unsatisfiable.

Next, we show that if a cycle is not of the form written in the lemma, then it is satisfiable. If $\rr_2=\p$, this means that there must be two $\p$-arcs in opposite directions, so the instance is satisfiable by Claim~\ref{clm:oppposite-p-arcs}. 

It remains to consider $\rr_2 \in \{\o,\di\}$.
Let $\III''$ be the instance obtained from~$\III$ by replacing each $u \s v$ constraint by a $u \e v$ constraint.
Observe that a cycle containing an $\e$-edge is satisfiable if and only if the cycle obtained by contracting this edge away is satisfiable.
Now $\III''$ is satisfiable by Case~\ref{case:bad-cycles-do}, as its primal graph does not contain a cycle in which all $\rr_2$ relations are in the same direction.
Let $\varphi''$ be a satisfying assignment for~$\III''$.
Observe that we may assume all values taken by this assignment are pairs of integers (if not, then multiply all values by a common multiple of their divisors).
Let~$n$ be the number of variables in~$\III$.
Let $u_1{\sf t}_1u_2,u_2{\sf t}_2u_3,\ldots,u_{n-1}{\sf t}_{n-1}u_n,u_n{\sf t}_nu_1$ be the cycle that forms the instance~$\III$ and assume without loss of generality that ${\sf t}_n\in \{\rr_2,\rr_2\sf i\}$ (where $\rr_2\sf i=\oi$ or $\d$ if $\rr_2=\o$ or~$\di$, respectively).
We construct a satisfying assignment~$\varphi$ for $\III$ as follows.
For all $v \in V(\III)$, let $\varphi^-(v)=\varphi''^-(v)$.
We let $\varphi(u_1)^+=\varphi''(u_1)^+$.
Next, for each $i \in \{1,\ldots,n-1\}$ in turn, we let $\varphi(u_{i+1})^+=\varphi''(u_{i+1})^+$ if ${\sf t}_i\in \{\rr_2,\rr_2\sf i\}$, $\varphi(u_{i+1})^+=\varphi(u_i)^++\frac{1}{2n}$ if ${\sf t}_i=\s$ and $\varphi(u_{i+1})^+=\varphi(u_i)^+-\frac{1}{2n}$ if ${\sf t}_i=\si$.
This changes the value of the each $u^+$ variables by less than $\frac{1}{2}$, but keeps all $u^-$ variables unchanged.
Therefore all $u^-<u^+$ relations remain satisfied, all $\rr_2$ relations remain satisfied and the $u^-=v^-$ parts of the $u \s v$ relations also remain satisfied.
By construction, the $u^+<v^+$ parts of the $u \s v$ also become satisfied.
Therefore $\varphi$ is a satisfying assignment for~$\III$.

\end{enumerate}
\end{proof}
\fi

\ifshort
\begin{proof}[Proof Sketch]
By Lemma~\ref{lem:cycles}, we only need to consider instances~$\III$ whose primal graph is a minimal bad cycle. 
If a cycle contains two $\p$-arcs in opposite directions, we can delete these arcs to divide the instance into two disjoint sub-instances, which are satisfiable.
We then choose satisfying assignments so that all the intervals in one of the sub-instances are to left of those in the other.
The resulting assignment also satisfies the two constraints we deleted so
cycles with two $\p$-arcs in opposite directions are always satisfiable.

If a cycle contains two $\d$ arcs that meet head-to-head, we can delete these two arcs and their common vertex~$v$.
The remaining instance is satisfiable, so we choose a satisfying assignment and set the interval corresponding to~$v$ to be one that contains all of the other intervals.
The resulting assignment also satisfies the two constraints we deleted.
Therefore cycles with two $\d$-arcs that meet head-to-head are always satisfiable.

Next, observe that for every relation $\rr \in \{\s,\f,\d,\p,\o\}$, if there is a directed cycle of $\rr$-arcs, then the cycle is unsatisfiable.
We now sketch the proofs for cases (1) and (4); the remaining cases can be proved
in similar ways.

\smallskip

\noindent
(1) We only need to consider cycles where all $\p$ arcs lie in the forward direction, and every sequence of consecutive $\d$ arcs contains a non-negative number of reverse $\d$-arcs followed by a non-negative number of forwards $\d$-arcs.
This means that all of the intervals involving this run of $\d$-arcs share a common point.
The $\p$-arcs then imply that these common points must form a cycle, where each of these common points is bigger than the next.
This then shows that the cycle is unsatisfiable.

\smallskip

\noindent
(4) The $\s$ constraints give equality constraints on left end-points and the $\rr_2$ constraints lead to $<$ constraints on the left end-points so
the claimed unsatisfiable cycles are indeed unsatisfiable.
In the $\rr_2=\p$ case, since minimal unsatisfiable instances cannot have $\p$-arcs in opposite directions, this allows us to complete the proof.
In the $\rr_2 \in \{\o,\d\}$ cases, if we replace the $\s$-arcs with $\e$-edges, we get a satisfiable instance.
Taking a satisfying assignment for this modified instance, we can change it to satisfy the original instance by small perturbations to the values of the right end-points of the intervals.
\end{proof}
\fi

\section{FPT Algorithms} \label{sec:fpt-algs}

\iflong
We prove Lemmas~\ref{lem:mp-fpt} and \ref{lem:sf-pod-fpt} in this section.
This amounts to proving that \textsc{MinCSP}$(\Gamma)$ is in \FPT\ for
seven subsets $\Gamma \subseteq \A$ of interval relations:
these are
$\Gamma = \{\m, \p\}$ and
$\Gamma = \{\rr_1, \rr_2, \e\}$ for
$\rr_1 \in \{ \s, \f \}$ and $\rr_2 \in \{ \p, \d, \o \}$.
We begin in Section~\ref{ssec:mp-fpt} by presenting the algorithm for $\{\m,\p\}$,
which works by first reducing the problem to \mincsp{<,=} and 
then to \subdfas.
In Section~\ref{ssec:lsmfas-fpt} we treat the remaining six cases
by reducing all of them to a fairly natural
generalization of \textsc{Directed Feedback Arc Set} problem,
and showing that this problem is fixed-parameter tractable.
\fi

\ifshort
We prove Lemma~\ref{lem:mp-fpt} in this section.
The fpt algorithm for $\{\m,\p\}$ is simple and we omit the details:
it works by first reducing the problem to \mincsp{<,=} and 
then to \subdfas.
The remaining six cases
are handled by reducing them to a fairly natural
generalization of \textsc{Directed Feedback Arc Set} problem,
and showing that this problem is in \FPT.
\fi

\iflong

\subsection{Algorithm for \mincsp{\m, \p}}
\label{ssec:mp-fpt}

We show that \mincsp{\m, \p} is in FPT.
The algorithm works by reducing the problem to
\mincsp{<, =} and then to \textsc{Subset Directed Feedback Arc Set}
defined as follows.

\pbDefP{Subset Directed Feedback Arc Set (\subdfas)}
{A directed graph $G$, a subset of red arcs $R \subseteq A(G)$, 
and an integer $k$.}
{$k$.}
{Is there a subset $Z \subseteq A(G)$ of size at most $k$ 
such that $G - Z$ contains no cycles with at least one red arc?}

We refer to arcs in $A(G) \setminus R$ as  {\em black}.

\begin{theorem}[\cite{chitnis2015directed}] \label{thm:sdfas-exact}
  \subdfas is solvable in $O^*(2^{O(k^3)})$ time.
\end{theorem}

First, we show that \mincsp{<, =}, i.e. \textsc{MinCSP}
over an infinite ordered domain (say $\rationals$)
with binary relations `less-than' $<$ and `equals' $=$,
reduces to \subdfas.

\begin{lemma} \label{lem:less-eq=sdfas}
  There is an algorithm that takes an instance $(\III, k)$
  of ${\mincsp{<,=}}$ as input and produces in polynomial time 
  an instance $(G, R, k)$ of \subdfas
  such that $(\III, k)$ is a yes-instance if and only if
  $(G, R, k)$ is a yes-instance.
\end{lemma}
\begin{proof}
Let $(\III, k)$ be an instance of $\mincsp{<,=}$.
Construct a directed graph $G$ with the variables in $V(\III)$ as vertices.
If $\III$ contains a constraint $x = y$,
add black arcs $(x,y)$ and $(y,x)$ to $A(G)$.
If $\III$ contains a constraint $x < y$,
add a red arc $(x, y)$ to $A(G)$ and $R$.
This completes the construction. The colours of the arcs are
added only to make the proof easier to read.

For an intuition, one can think of the black arcs in $G$
as $\leq$ relations, and the red arcs as $<$ relations.
Then two black arcs $(x,y)$ and $(y,x)$
enforce equality $x = y$.
A directed cycle in $G$ consisting 
solely of black arcs (i.e. $\leq$ relations)
corresponds to a satisfiable subset of constraints --
indeed, we can assign the same value to all variables
and satisfy all these constrains.
On the other hand, if a directed cycle 
contains a red arc (i.e. $<$ relation),
then it corresponds to an inconsistent instance
since it implies $x < y$ and $x \geq y$ for a pair of variables.

To formally prove correctness, 
we first assume $X$ to be a set of constraints
such that $\III - X$ is consistent.
Create a set of arcs $Z$ as follows:
if $X$ contains the constraint $x = y$,
then add one black arc $(x, y)$ to $Z$;
if $X$ contains the constraint $x < y$,
then add the red arc $(x, y)$ to $Z$.
Clearly, $|X| = |Z|$.
We claim that $Z$ intersects every cycle that
contains an arc from $S \setminus Z$.
Assume to the contrary that $G - Z$
contains a cycle with a red arc,
i.e. a directed $(x,y)$-path $P$ and 
a red arc $(y,x) \in S \setminus Z$.
The constraints along $P$ imply $x \leq y$,
while the red arc $(y, x)$ corresponds to a constraint $x > y$,
contradicting that $\III - X$ is satisfiable.

For the opposite direction, suppose that
$Z'$ is a set of arcs such that 
$G - Z'$ does not contain any cycle with at least one red arc.
Construct a set of constraints $X'$ as follows:
if $Z'$ contains a black arc $(x, y)$,
then add $x = y$ to $X'$, and
if $Z'$ contains a red arc $(x, y)$,
then add $x < y$ to $X'$.
This construction implies that $|X'| \leq |Z'|$.
We show that $\III - X'$ is consistent by defining an explicit
assignment $\varphi : V(\III) \rightarrow \rationals$.
To this end, construct a graph $H$ from $G - Z'$ by identifying
every pair of vertices $u$ and $v$ that are 
strongly connected by black arcs.
We know that there is no cycle of black arcs in $H$.
Moreover, there is no cycle with a red arc because
it would correspond to a cycle with a red arc in $G - Z'$.
Hence, $H$ is acyclic.
Let $\varphi$ be a linear ordering of $H$, i.e.
an assignment that satisfies the following property:
if we have $\varphi(v) > \varphi(u)$, 
then there is no $(u,v)$-path in $H$.
We claim that $\varphi$ satisfies $\III - X'$.
If there is a constraint $x = y$ in $\III - X'$, then
$G - Z'$ contains two arcs $(x, y)$ and $(y, x)$,
making $x$ and $y$ strongly connected by black arcs.
Then $x$ and $y$ are identified in $H$ and $\varphi(x) = \varphi(y)$.
If $\III - X'$ contains a constraint $x < y$, then
$x$ and $y$ are not identified since there is no $(y, x)$-path in $G - Z'$.
Moreover, arc $(x, y)$ is present in $G - Z'$ and $H$,
and, since $\varphi$ is a linear ordering of $H$,
we have $\varphi(x) < \varphi(y)$, and the constraint $x < y$ is satisfied.
This concludes the proof.
\end{proof}

Now we are ready to prove that $\mincsp{\m, \p}$ is in \FPT.

\begin{proof}[Proof of Lemma~\ref{lem:mp-fpt}]
Lemma~\ref{lem:implementations} implies that we can replace
a $\p$-constraint with an implementation
using two $\m$ constraints.
This observation, together with Lemma~\ref{lem:impl-reduction}, 
implies that ${\mincsp{\m, \p}}$ is in \FPT if $\mincsp{\m}$ is in \FPT.
We show how to reduce the latter problem to ${\mincsp{<, =}}$.

Let $(\III, k)$ be an instance of $\mincsp{\m}$.
Construct an instance $(\III', k)$ of ${\mincsp{<, =}}$
by introducing two variables $x^-$ and $x^+$ in 
$V(\III')$ for every $x \in V(\III)$.
Add crisp constraints $x^- < x^+$ to $\III'$ for all $x \in V(\III)$, 
and soft constraints $x^+ = y^-$ to $\III'$ for every constraint
$x \{\m\} y$ in $\III$.
This completes the reduction and 
it can clearly be performed in polynomial time.

To prove correctness, we show that $\III$ and $\III'$ have the same cost.
To show that $\cost(\III) \leq \cost(\III')$, let $\varphi$
be an arbitrary assignment of non-empty 
intervals to the variables in $V(\III)$.
We construct an assignment $\varphi'$ to $\III'$ with the same cost.
In fact, one can view $\varphi$ as an assignment of rationals to the endpoints,
i.e. $\varphi' : V(\III') \to \rationals$ such that $\varphi(x) = [\varphi'(x^-), \varphi'(x^+)]$.
Clearly, we have $\varphi(v^-) < \varphi(v^+)$ for all $v \in V(\III)$
since the intervals are non-empty.
Furthermore, a constraint $x \{\m\} y$ holds if and only if 
$\varphi(x^+) = \varphi(y^-)$.
Hence, $\varphi$ and $\varphi'$ violate the same number of constraints in $\III$ and $\III'$, respectively.

To show that $\cost(\III') \leq \cost(\III)$, 
assume $\psi'$ is an assignment 
of rational values to $V(\III')$
that satisfies all crisp constraints.
Define an interval assignment $\psi$ as
$\psi(v) = [\psi(v^-), \psi(v^+)]$ for all $v \in V(\III)$.
The intervals are non-empty because $\psi(v^-) < \psi(v^+)$
is enforced by the crisp constraints.
Moreover, $\psi'$ satisfies $x^+ = y^-$ if and only if
$\psi$ satisfies $x \{\m\} y$.
Hence, $\psi$ and $\psi'$ violate the same number of constraints in $\III$ and $\III'$, respectively.
We conclude that $\cost(\III) = \cost(\III')$.
\end{proof}

\fi

\iflong
\subsection{Algorithm by Reduction to Bundled Cut}
\label{ssec:lsmfas-fpt}

In this section we prove fixed-parameter tractability of
$\mincsp{\rr_1, \rr_2, \e}$, where $\rr_1 \in \{\s, \f\}$ and
$\rr_2 \in \{\p, \o, \d\}$.
Lemma~\ref{lem:bad-cycles}.4 suggests that all six fragments 
allow uniform treatment.
Indeed, to check whether an instance of $\csp{\rr_1, \rr_2, \e}$
is consistent, one can identify 
all variables constrained to be equal.
This corresponds exactly to contracting all edges in the graph $G_\III$.
Then $\III$ becomes an instance of
$\csp{\rr_1, \rr_2}$, and the criterion of 
Lemma~\ref{lem:bad-cycles}.4 applies.
This observation allows us to formulate $\mincsp{\rr_1, \rr_2, \e}$
as a variant of feedback arc set on mixed graphs.

\begin{definition}
Consider a mixed graph $G$ with arcs of two types -- short and long --
and a walk $W$ in $G$ from $u$ to $v$ that may ignore direction of the arcs.
The walk $W$ in $G$ is \emph{undirected} if it only contains edges.
The walk $W$ is \emph{short} if it contains a short arc but no long arcs.
The walk $W$ is \emph{long} if it contains a long arc.
The walk $W$ is \emph{directed} if it is either short and all short arcs are directed from $u$ to $v$
or if it is long and all long arcs are directed from $u$ to $v$.
If $W$ is short or long, but not directed, it is \emph{mixed}.
\end{definition}

Note that short arcs on a directed long walk may be directed arbitrarily.

\pbDefP{Mixed Feedback Arc Set with Short and Long Arcs (\lsmfas)}
{A mixed graph $G$ with the arc set $A(G)$ partitioned into
\emph{short} $A_s$ and \emph{long} $A_\ell$, and an integer $k$.}
{$k$.}
{Is there a set $Z \subseteq E(G) \cup A(G)$ such that
$G - Z$ contains neither short directed cycles nor long directed cycles?}

The main result of this section is the following theorem.

\begin{theorem}\label{thm:lsmfas-algo}
  \lsmfas can be solved in time $O^*(2^{O(k^8 \log k)})$.
\end{theorem}

Lemma~\ref{lem:sf-pod-fpt} positing that
$\mincsp{\rr_1, \rr_2, \e}$ is in FPT whenever
$\rr_1 \in \{\s, \f\}$ and $\rr_2 \in \{\p, \o, \d\}$
is an immediate consequence of 
Lemma~\ref{lem:bad-cycles}.4 and Theorem~\ref{thm:lsmfas-algo}.

It is informative to understand the structure of mixed graphs
without bad cycles in the sense of \lsmfas.
One way of characterizing \lsmfas (with a pronounced connection to interval
constraints) is as a problem of placing 
vertices of the input graph $G$ on the number line
under certain distance constraints represented by the edges and arcs.
An edge $uv$ requires $u$ and $v$ to be placed at the same point.
An arc $(u, v)$ requires that $u$ is placed before $v$.
If the arc is long, then the distance from $u$ to $v$ should be big
(say, $2|V(G)|$ units of length), while if the arc is short,
the distance should be small (say, at most a unit of length).
We make the intuition precise in the following lemma.

\begin{lemma} \label{lem:lsmfas-struct}
  Let $G$ be a mixed graph with long and short arcs.
  Then $G$ contains no directed long cycles nor directed short cycles if and only if there exists a pair of mappings
  $\sigma_1,\sigma_2 : V(G) \to \mathbb{N}$ such that
  \begin{enumerate}
    \item for every $u,v \in V(G)$, $u$ and $v$ are connected by an undirected walk if and only if $(\sigma_1,\sigma_2)(u) = (\sigma_1,\sigma_2)(v)$;
    \item for every $u,v \in V(G)$, there exists a short $(u,v)$-walk in $G$ if and only if $\sigma_1(u) = \sigma_1(v)$;
    \item for every $u,v \in V(G)$, if there exists a directed short $(u,v)$-walk in $G$, then $\sigma_2(u) < \sigma_2(v)$;
    \item for every $u,v \in V(G)$, if there exists a directed long $(u,v)$-walk in $G$, then $\sigma_1(u) < \sigma_1(v)$.
  \end{enumerate}    
\end{lemma} 
\begin{proof}
{\bf Forward direction.} Consider a graph $G_1$ 
constructed from $G$ as follows: we replace every short arc with an undirect edge, and then contract
all undirect edges. While performing the contraction, we do not delete long arcs, even if they
become loops or multiple arcs. In other words, 
the vertex set of $G_1$ is the set of (weak) connected components of the subgraph of $G$
consisting of undirected edges and short arcs, and every long arc $(u,v)$ of $G$ yields
an arc in $G_1$ from the component containing $u$ to the component containing $v$.
Then, the assumption that $G$ has no directed long cycles implies that $G_1$ is acyclic: any cycle $C_1$ in $G_1$ can be lifted
to a directed long cycle $C$ in $G$ by padding with contracted undirected edges and short arcs.
Fix a topological ordering of $G_1$ and define $\sigma_1(v)$ for $v \in V(G)$ as the index of the image of $v$ in $G_1$ in the said topological
ordering. This gives the second and the fourth property.

Now, consider a graph $G_2$ defined as the graph $G$ with first all long arcs deleted and then all undirected edges contracted. (Similarly as for $G_1$, we do not delete short arcs while doing the
contractions, keeping all arising loops and multiple arcs.)
Note that the connected components of $G_2$ are in one-to-one correspondence with the vertices of $G_1$. 
Furthermore, similarly as in the case of $G_1$, the assumption that $G$ has no directed short cycles implies that $G_2$ is acyclic.
For every connected component $C$ of $G_2$, fix a topological ordering of $C$ and for $v \in C$ define $\sigma_2(v)$ as the index
of the image of $v$ in this topological ordering. This gives the first and the third properties.

\medskip

\noindent
{\bf Backward direction.}
Let $\sigma_1,\sigma_2$ be two mappings as in the statement of the lemma.
Assume first that $G$ contains a directed short cycle $C$. By the second property, $\sigma_1$ is constant on the vertices of $C$. 
By the third property, for every short arc $(u,v)$ on $C$ we have $\sigma_2(u) < \sigma_2(v)$ while
by the first property, for every undirected edge $uv$ on $C$ we have $\sigma_2(u) = \sigma_2(v)$. 
This is a contradiction with the fact that $C$ contains at least one short arc and no long arcs.

Assume now that $G$ contains a directed long cycle $C$. By the second property, for every undirected edge $uv$
or short arc $(u,v)$ on $C$, we have $\sigma_1(u) = \sigma_1(v)$. By the fourth property, for every long arc $(u,v)$ on $C$
we have $\sigma_1(u) < \sigma_1(v)$. This is a contradiction with the fact that $C$ contains at least one long arc.
\end{proof}

Now we introduce the technical machinery used in our algorithm for \lsmfas.
We start by iterative compression, a standard method in parameterized 
algorithms~(see~e.g.~Chapter~4~in~\cite{cygan2015parameterized}).
The idea is to solve the problem starting from an empty graph
and iteratively build up the graph by add edges and arcs, maintaining a solution of size $k$ in the process.
We clearly have a solution while the size of the graph is at most $k$.
Now consider an iteration where we have a graph $G$ together with a solution $Z$, and a new edge/arc $e$ is added.
Observe that $Z \cup \{e\}$ is a solution to $(V(G), E(G) \cup \{e\})$, thus we may assume access
to an approximate solution of size $k + 1$ at each iteration.

The problem resulting from iterative compression reduces to {\sc Bundled Cut} 
with pairwise linked deletable edges,
defined in~\cite{Kim:etal:FA3} and solved using the flow-augmentation technique of~\cite{Kim:etal:stoc2022}.
To describe the problem, let $G$ be a directed graph with two distinguished vertices $s,t \in V(G)$.
Let $\bundles$ be a family of pairwise disjoint subsets of $E(G)$, which we call \emph{bundles}.
The edges of $\bigcup \bundles$ are called \emph{soft} and the edges of $E(G) \setminus \bigcup \bundles$
are \emph{crisp}. A set $Z \subseteq \bigcup \bundles$ \emph{violates} a bundle $B \in \bundles$
if $Z \cap B \neq \emptyset$ and \emph{satisfies} $B$ otherwise. 

\pbDefP{Bundled Cut}
{A directed graph $G$, two distinguished vertices $s,t \in V(G)$, 
  a family $\bundles$ of pairwise disjoint subsets of $E(G)$, and an integer $k$.}
{$k$.}
{Is there an $st$-cut $Z \subseteq \bigcup \bundles$ that violates at most $k$ bundles?}

In general, \probbc is W[1]-hard even if all bundles are of size $2$. However,
there is a special case of \probbc that is tractable.
Let $(G,s,t,\bundles,k)$ be a \probbc instance. 
A soft arc $e$ is \emph{deletable} if there is no crisp copy of $e$ in $G$.
An instance $(G,s,t,\bundles,k)$ has \emph{pairwise linked deletable arcs}
if for every $B \in \bundles$ and every two deletable arcs $e_1,e_2 \in B$, there exists
in $G$ a path from an endpoint of one of the arcs $e_1,e_2$ to an endpoint of the second of those arcs
that does not use any arcs of $\bundles \setminus \{B\}$.

The assumption of pairwise linked deletable arcs makes \probbc tractable.
\begin{theorem}[Theorem~4.1 of~\cite{kim2023flowaugmentation}]\label{thm:probbc}
\probbc instances with pairwise linked deletable arcs can be solved
in time $O^*(2^{O(k^4 d^4 \log(kd))})$ where $d$ is the maximum number of deletable arcs in a single bundle.
\end{theorem}

Armed with Lemma~\ref{lem:lsmfas-struct} and Theorem~\ref{thm:probbc}, 
we are ready to prove the main result.

\begin{proof}[Proof of Theorem~\ref{thm:lsmfas-algo}]
Let $(G, A_s, A_\ell, k)$ be an instance of \lsmfas.
By iterative compression, we can assume that 
we have access to a set 
\iflong $X \subseteq E(G) \cup A(G)$ \fi 
\ifshort $Y \subseteq V(G)$ \fi
of 
size at most $k+1$ that intersects all bad cycles.
\iflong
Subdivide every edge and arc in $X$,
i.e. for every edge $uv \in X$ (arc $(u,v) \in X$, respectively),
replace it by edges $uy_{uv}$ and $y_{uv}v$
(arcs $(u, y_{uv})$ and $(y_{uv}, v)$, respectively),
where $y_{uv}$ is a fresh variable.
Let $Y = \{ y_{uv} : uv \in X \text{ or } (u, v) \in X \}$.
Note that $Y \subseteq V(G)$ intersects all bad cycles in $G$.
Moreover, the resulting graph has a set of $k$ edges
intersecting all bad cycles if and only if
the original graph has one.
\fi
We will refer to the vertices of $Y$ as \emph{terminals}.

Fix a hypothetical solution $Z \subseteq A(G) \cup E(G)$.
Guess which pairs of terminals are
connected by undirected paths in $G - Z$ and identify them.
Define \emph{ordering} $\ordering : Y \to \naturals \times \naturals$ 
that maps terminals to
\begin{equation*}
  \{
  (1, 1), \dots, (1, q_1), \cdots,
  (i, 1), \dots, (i, q_i), \cdots,
  (p, 1), \dots, (p, q_p)
  \}
\end{equation*}
such that the following hold.
For every pair of terminals $y,y' \in Y$,
let $\ordering(y) = (i, j)$ and
$\ordering(y') = (i', j')$.
\begin{itemize}
  \item If $y$ and $y'$ are connected by 
  a short path in $G - Z$, then $i = i'$.
  \item If $y$ reaches $y'$ by 
  a directed short path in $G - Z$, then $j < j'$.
  \item If $y$ reaches $y'$ by
  a directed long path in $G - Z$, then $i < i'$.
\end{itemize}
Note that $\ordering$ exists by Lemma~\ref{lem:lsmfas-struct}.
If an ordering satisfies the conditions above,
we say that it is \emph{compatible with} $G - Z$.
In the sequel we write $(i,j) < (i',j')$
to denote that $(i, j)$ 
lexicographically precedes $(i',j')$,
i.e. either $i = i'$ and $j < j'$ or $i < i'$.

\begin{figure}
  \begin{center}
  \begin{tikzpicture}[
roundnode/.style={circle, draw, minimum size=10mm},
]
  \def \dx {1.5}
  \def \dy {1.5}

  \node[roundnode] (r) at (2  * \dx,  2 * \dy) {\large $r$};
  \node[roundnode] (s) at (4  * \dx,  2 * \dy) {\large $s$};
  \node[roundnode] (t) at (8  * \dx,  2 * \dy) {\large $t$};

  \node[roundnode] (r11) at (2  * \dx,  1 * \dy) {$1,1$};
  \node[roundnode] (s12) at (4  * \dx,  1 * \dy) {$1,2$};
  \node[roundnode] (t21) at (8  * \dx,  1 * \dy) {$2,1$};  

  \draw[->] (r) -- node [left] {$\ordering$} (r11);
  \draw[->] (s) -- node [left] {$\ordering$} (s12);
  \draw[->] (t) -- node [left] {$\ordering$} (t21);

  \node[roundnode]                    (x11) at (0  * \dx,  0 * \dy) {$(1,1)$};
  \node[roundnode, fill=yellow!80]    (x21) at (1  * \dx,  0 * \dy) {$(2,1)$};
  \node[roundnode, fill=cyan!60]      (x22) at (2  * \dx,  0 * \dy) {$(2,2)$};
  \node[roundnode, fill=yellow!80]    (x23) at (3  * \dx,  0 * \dy) {$(2,3)$};
  \node[roundnode, fill=cyan!60]      (x24) at (4  * \dx,  0 * \dy) {$(2,4)$};
  \node[roundnode, fill=yellow!80]    (x25) at (5  * \dx,  0 * \dy) {$(2,5)$};
  \node[roundnode]                    (x31) at (6  * \dx,  0 * \dy) {$(3,1)$};
  \node[roundnode, fill=yellow!80]    (x41) at (7  * \dx,  0 * \dy) {$(4,1)$};
  \node[roundnode, fill=cyan!60]      (x42) at (8  * \dx,  0 * \dy) {$(4,2)$};
  \node[roundnode, fill=yellow!80]    (x43) at (9  * \dx,  0 * \dy) {$(4,3)$};
  \node[roundnode]                    (x51) at (10 * \dx,  0 * \dy) {$(5,1)$};
\end{tikzpicture}      
  \end{center}

  \vspace{1cm}
  
  \begin{center}
  \begin{tikzpicture}[
  roundnode/.style={circle, draw, minimum size=10mm}
]
  \def \dx {6}
  \def \dy {1.5}

  \node[roundnode]                    (x11) at (0 * \dx,  0 * \dy) {$x_{1,1}$};
  \node[roundnode, fill=yellow!80]    (x21) at (0 * \dx,  1 * \dy) {$x_{2,1}$};
  \node[roundnode, fill=cyan!60]      (x22) at (0 * \dx,  2 * \dy) {$x_{2,2}$};
  \node[roundnode, fill=yellow!80]    (x23) at (0 * \dx,  3 * \dy) {$x_{2,3}$};
  \node[roundnode, fill=cyan!60]      (x24) at (0 * \dx,  4 * \dy) {$x_{2,4}$};
  \node[roundnode, fill=yellow!80]    (x25) at (0 * \dx,  5 * \dy) {$x_{2,5}$};
  \node[roundnode]                    (x31) at (0 * \dx,  6 * \dy) {$x_{3,1}$};
  \node[roundnode, fill=yellow!80]    (x41) at (0 * \dx,  7 * \dy) {$x_{4,1}$};
  \node[roundnode, fill=cyan!60]      (x42) at (0 * \dx,  8 * \dy) {$x_{4,2}$};
  \node[roundnode, fill=yellow!80]    (x43) at (0 * \dx,  9 * \dy) {$x_{4,3}$};
  \node[roundnode]                    (x51) at (0 * \dx, 10 * \dy) {$x_{5,1}$};
  \node[roundnode, thick]             (x)   at (0 * \dx, 11 * \dy) {\large $x$};

  \node[roundnode]                    (y11) at (1 * \dx,  0 * \dy) {$y_{1,1}$};
  \node[roundnode, fill=yellow!80]    (y21) at (1 * \dx,  1 * \dy) {$y_{2,1}$};
  \node[roundnode, fill=cyan!60]      (y22) at (1 * \dx,  2 * \dy) {$y_{2,2}$};
  \node[roundnode, fill=yellow!80]    (y23) at (1 * \dx,  3 * \dy) {$y_{2,3}$};
  \node[roundnode, fill=cyan!60]      (y24) at (1 * \dx,  4 * \dy) {$y_{2,4}$};
  \node[roundnode, fill=yellow!80]    (y25) at (1 * \dx,  5 * \dy) {$y_{2,5}$};
  \node[roundnode]                    (y31) at (1 * \dx,  6 * \dy) {$y_{3,1}$};
  \node[roundnode, fill=yellow!80]    (y41) at (1 * \dx,  7 * \dy) {$y_{4,1}$};
  \node[roundnode, fill=cyan!60]      (y42) at (1 * \dx,  8 * \dy) {$y_{4,2}$};
  \node[roundnode, fill=yellow!80]    (y43) at (1 * \dx,  9 * \dy) {$y_{4,3}$};
  \node[roundnode]                    (y51) at (1 * \dx, 10 * \dy) {$y_{5,1}$};
  \node[roundnode, thick]             (y)   at (1 * \dx, 11 * \dy) {\large $y$};

  \draw[-{Latex[length=2mm]}] (x) -- node[above] {\large short} (y);

  \draw[-{Latex[length=2mm]}] (x51.south) -- (x43.north);
  \draw[-{Latex[length=2mm]}] (x43.south) -- (x42.north);
  \draw[-{Latex[length=2mm]}] (x42.south) -- (x41.north);
  \draw[-{Latex[length=2mm]}] (x41.south) -- (x31.north);
  \draw[-{Latex[length=2mm]}] (x31.south) -- (x25.north);
  \draw[-{Latex[length=2mm]}] (x25.south) -- (x24.north);
  \draw[-{Latex[length=2mm]}] (x24.south) -- (x23.north);
  \draw[-{Latex[length=2mm]}] (x23.south) -- (x22.north);
  \draw[-{Latex[length=2mm]}] (x22.south) -- (x21.north);
  \draw[-{Latex[length=2mm]}] (x21.south) -- (x11.north);

  \draw[-{Latex[length=2mm]}] (y51.south) -- (y43.north);
  \draw[-{Latex[length=2mm]}] (y43.south) -- (y42.north);
  \draw[-{Latex[length=2mm]}] (y42.south) -- (y41.north);
  \draw[-{Latex[length=2mm]}] (y41.south) -- (y31.north);
  \draw[-{Latex[length=2mm]}] (y31.south) -- (y25.north);
  \draw[-{Latex[length=2mm]}] (y25.south) -- (y24.north);
  \draw[-{Latex[length=2mm]}] (y24.south) -- (y23.north);
  \draw[-{Latex[length=2mm]}] (y23.south) -- (y22.north);
  \draw[-{Latex[length=2mm]}] (y22.south) -- (y21.north);
  \draw[-{Latex[length=2mm]}] (y21.south) -- (y11.north);

  \draw[-{Latex[length=2mm]}] (x11) -- (y11); 
  \draw[-{Latex[length=2mm]}] (x21) -- (y21); 
  \draw[-{Latex[length=2mm]}] (x22) -- (y23); 
  \draw[-{Latex[length=2mm]}] (x23) -- (y23); 
  \draw[-{Latex[length=2mm]}] (x24) -- (y25); 
  \draw[-{Latex[length=2mm]}] (x25) -- (y25); 
  \draw[-{Latex[length=2mm]}] (x31) -- (y31); 
  \draw[-{Latex[length=2mm]}] (x41) -- (y41); 
  \draw[-{Latex[length=2mm]}] (x42) -- (y43); 
  \draw[-{Latex[length=2mm]}] (x43) -- (y43); 
  \draw[-{Latex[length=2mm]}] (x51) -- (y51); 

  \draw[-{Latex[length=2mm]}, dashed] (y11) edge [out=190, in=-10] (x11); 
  \draw[-{Latex[length=2mm]}, dashed] (y21) edge [out=170, in=10] (x21); 
  \draw[-{Latex[length=2mm]}, dashed] (y22) edge [out=190, in=10] (x21); 
  \draw[-{Latex[length=2mm]}, dashed] (y23) edge [out=210, in=10] (x21); 
  \draw[-{Latex[length=2mm]}, dashed] (y24) edge [out=210, in=10] (x21); 
  \draw[-{Latex[length=2mm]}, dashed] (y25) edge [out=210, in=10] (x21); 
  \draw[-{Latex[length=2mm]}, dashed] (y31) edge [out=190, in=-10] (x31); 
  \draw[-{Latex[length=2mm]}, dashed] (y41) edge [out=190, in=-10] (x41); 
  \draw[-{Latex[length=2mm]}, dashed] (y42) edge [out=180, in=10]   (x41); 
  \draw[-{Latex[length=2mm]}, dashed] (y43) edge [out=210, in=10]   (x41); 
  \draw[-{Latex[length=2mm]}, dashed] (y51) edge [out=190, in=-10] (x51); 

\end{tikzpicture}
  \hspace{1cm}
  \begin{tikzpicture}[
  roundnode/.style={circle, draw, minimum size=10mm}
]
  \def \dx {6}
  \def \dy {1.5}

  \node[roundnode]                    (x11) at (0 * \dx,  0 * \dy) {$x_{1,1}$};
  \node[roundnode, fill=yellow!80]    (x21) at (0 * \dx,  1 * \dy) {$x_{2,1}$};
  \node[roundnode, fill=cyan!60]      (x22) at (0 * \dx,  2 * \dy) {$x_{2,2}$};
  \node[roundnode, fill=yellow!80]    (x23) at (0 * \dx,  3 * \dy) {$x_{2,3}$};
  \node[roundnode, fill=cyan!60]      (x24) at (0 * \dx,  4 * \dy) {$x_{2,4}$};
  \node[roundnode, fill=yellow!80]    (x25) at (0 * \dx,  5 * \dy) {$x_{2,5}$};
  \node[roundnode]                    (x31) at (0 * \dx,  6 * \dy) {$x_{3,1}$};
  \node[roundnode, fill=yellow!80]    (x41) at (0 * \dx,  7 * \dy) {$x_{4,1}$};
  \node[roundnode, fill=cyan!60]      (x42) at (0 * \dx,  8 * \dy) {$x_{4,2}$};
  \node[roundnode, fill=yellow!80]    (x43) at (0 * \dx,  9 * \dy) {$x_{4,3}$};
  \node[roundnode]                    (x51) at (0 * \dx, 10 * \dy) {$x_{5,1}$};
  \node[roundnode, thick]             (x)   at (0 * \dx, 11 * \dy) {\large $x$};

  \node[roundnode]                    (y11) at (1 * \dx,  0 * \dy) {$y_{1,1}$};
  \node[roundnode, fill=yellow!80]    (y21) at (1 * \dx,  1 * \dy) {$y_{2,1}$};
  \node[roundnode, fill=cyan!60]      (y22) at (1 * \dx,  2 * \dy) {$y_{2,2}$};
  \node[roundnode, fill=yellow!80]    (y23) at (1 * \dx,  3 * \dy) {$y_{2,3}$};
  \node[roundnode, fill=cyan!60]      (y24) at (1 * \dx,  4 * \dy) {$y_{2,4}$};
  \node[roundnode, fill=yellow!80]    (y25) at (1 * \dx,  5 * \dy) {$y_{2,5}$};
  \node[roundnode]                    (y31) at (1 * \dx,  6 * \dy) {$y_{3,1}$};
  \node[roundnode, fill=yellow!80]    (y41) at (1 * \dx,  7 * \dy) {$y_{4,1}$};
  \node[roundnode, fill=cyan!60]      (y42) at (1 * \dx,  8 * \dy) {$y_{4,2}$};
  \node[roundnode, fill=yellow!80]    (y43) at (1 * \dx,  9 * \dy) {$y_{4,3}$};
  \node[roundnode]                    (y51) at (1 * \dx, 10 * \dy) {$y_{5,1}$};
  \node[roundnode, thick]             (y)   at (1 * \dx, 11 * \dy) {\large $y$};

  \draw[-{Latex[length=2mm]}] (x) -- node[above] {\large long} (y);

  \draw[-{Latex[length=2mm]}] (x51.south) -- (x43.north);
  \draw[-{Latex[length=2mm]}] (x43.south) -- (x42.north);
  \draw[-{Latex[length=2mm]}] (x42.south) -- (x41.north);
  \draw[-{Latex[length=2mm]}] (x41.south) -- (x31.north);
  \draw[-{Latex[length=2mm]}] (x31.south) -- (x25.north);
  \draw[-{Latex[length=2mm]}] (x25.south) -- (x24.north);
  \draw[-{Latex[length=2mm]}] (x24.south) -- (x23.north);
  \draw[-{Latex[length=2mm]}] (x23.south) -- (x22.north);
  \draw[-{Latex[length=2mm]}] (x22.south) -- (x21.north);
  \draw[-{Latex[length=2mm]}] (x21.south) -- (x11.north);

  \draw[-{Latex[length=2mm]}] (y51.south) -- (y43.north);
  \draw[-{Latex[length=2mm]}] (y43.south) -- (y42.north);
  \draw[-{Latex[length=2mm]}] (y42.south) -- (y41.north);
  \draw[-{Latex[length=2mm]}] (y41.south) -- (y31.north);
  \draw[-{Latex[length=2mm]}] (y31.south) -- (y25.north);
  \draw[-{Latex[length=2mm]}] (y25.south) -- (y24.north);
  \draw[-{Latex[length=2mm]}] (y24.south) -- (y23.north);
  \draw[-{Latex[length=2mm]}] (y23.south) -- (y22.north);
  \draw[-{Latex[length=2mm]}] (y22.south) -- (y21.north);
  \draw[-{Latex[length=2mm]}] (y21.south) -- (y11.north);

  \draw[-{Latex[length=2mm]}] (x11) -- (y11); 
  \draw[-{Latex[length=2mm]}] (x21) -- (y31); 
  \draw[-{Latex[length=2mm]}] (x22) -- (y31); 
  \draw[-{Latex[length=2mm]}] (x23) -- (y31); 
  \draw[-{Latex[length=2mm]}] (x24) -- (y31); 
  \draw[-{Latex[length=2mm]}] (x25) -- (y31); 
  \draw[-{Latex[length=2mm]}] (x31) -- (y31); 
  \draw[-{Latex[length=2mm]}] (x41) -- (y51); 
  \draw[-{Latex[length=2mm]}] (x42) -- (y51); 
  \draw[-{Latex[length=2mm]}] (x43) -- (y51); 
  \draw[-{Latex[length=2mm]}] (x51) -- (y51); 

\end{tikzpicture}
  \caption{An example with three terminals $r,s,t$ where
  $\ordering(r) = (1,1)$,
  $\ordering(s) = (1,2)$,
  $\ordering(t) = (2,1)$.}
  \label{fig:drawing}
  \end{center}
\end{figure}

\begin{table}
  \begin{center}
    \begin{tabular}{| c || c | c | c |} 
      \hline
                                        & $i$ odd            & $i$ even, $j$ odd          & $i$ even, $j$ even  \\ 
      \hline
      \multirow{2}{*}{Edge $uv$}        & $(i,j) \to (i,j)$ &    $(i,j) \to (i,j)$ & $(i,j) \to (i,j)$                           \\
                                        & $(i,j) \leftarrow (i,j)$ &    $(i,j) \leftarrow (i,j)$ &       $(i,j) \leftarrow (i,j)$              \\
      \hline
      \multirow{2}{*}{Short $(u,v)$}    & $(i,j) \to (i,j)$ &  $(i,j) \to (i,j)$        & $(i,j) \to (i,j+1)$ \\
                                       & $(i,1) \leftarrow (i,j)$ &   $(i,1) \leftarrow (i,j)$      & $(i,1) \leftarrow (i,j)$ \\
      \hline
      Long $(u,v)$                      & $(i,1) \to (i,1)$  & $(i,j) \to (i+1,1)$ &    $(i,j) \to (i+1,1)$     \\            
      \hline
    \end{tabular}  
  \end{center}
  \caption{Correspondence between edges, short arcs and long arcs of the \lsmfas instance
  and the arcs introduced in the reduction to \textsc{Bundled Cut} in Theorem~\ref{thm:lsmfas-algo}. }
  \label{tb:reduction}
\end{table}

For the algorithm,
proceed by guessing an ordering $\ordering$, 
creating $2^{O(k \log k)}$ branches in total.
For each $\ordering$, create an instance 
$(H := H(G, \ordering), \bundles := \bundles(G, \ordering), k)$ of \probbc as follows.
Introduce two distinguished vertices $s$ and $t$ in $H$.
For every vertex $v \in V(G)$, create vertices $v^{i}_{1}$ in $H$  
for all odd $i \in [2p + 1]$ and vertices $v^{i}_{j}$ in $H$
for all even $i \in [2p + 1]$ and all $j \in [2q_i + 1]$.
Connect vertices created above by \emph{downward} arcs
$(v^{i}_{j}, v^{i'}_{j'})$ for all $(i, j) > (i', j')$.
For every terminal $y$, let $\sigma(y) = (i, j)$, and
add arcs $(s, y^{2i}_{2j})$ and $(y^{2i}_{2j+1}, t)$ in $H$.
Using the rules below, create a bundle $B_e$ in $\bundles$
for every $e \in E(G) \cup A(G)$,
add the newly created arcs to $H$.
\begin{itemize}
  \item For an edge $e = uv$, let $B_e$ consist of
  arcs $(u^{i}_{j}, v^{i}_{j})$ and $(v^{i}_{j}, u^{i}_{j})$ for all $(i, j)$.
  \item For short arcs $e = (u, v)$, let $B_e$ consist of
  arcs $(u^{i}_{j}, v^{i}_{j})$ for all $i, j$ such that $i$ or $j$ is odd,
  arcs $(u^{i}_{j}, v^{i}_{j+1})$ for all even $i, j$, and
  arcs $(v^{i}_{j}, u^{i}_{1})$ for all $i,j$.  
  \item For long arcs $e = (u, v)$, let $B_e$ consist of
  arcs $(u^{i}_{1}, v^{i}_{1})$ for all odd $i$, and arcs
  arcs $(u^{i}_{j}, v^{i+1}_{1})$ for all even $i$ and all $j$.
\end{itemize}
This completes the construction.
Bundle construction rules are summarized in Table~\ref{tb:reduction},
and Figure~\ref{fig:drawing} contains an illustrated example.

Observe that the downward arcs ensure 
that $(H, \bundles, k)$ has the pairwise-linked deletable arc property.
Moreover, the bundle size is $O(k)$, so we can
solve $(H, \bundles, k)$ in $O^*(2^{O(k^8 \log{k})})$ time.
It remains to show that the reduction is correct.

\smallskip

\noindent
{\bf Forward direction.}
Suppose $Z$ is a solution to $(G, A_s, A_\ell, k)$, and 
$\ordering$ is an ordering of terminals $Y$ compatible with $G - Z$. 
Consider the instance $(H := H(G, \ordering), \bundles := \bundles(G, \ordering), k)$. 
Let $Z' = \bigcup_{e \in Z} B_e$.
Clearly, $Z'$ intersects at most $|Z| \leq k$ bundles.
We claim that $Z'$ is an $(s, t)$-cut in $H$.

Suppose for the sake of contradiction that
there is an $(s, t)$-path $P$ in $H - Z'$.
By construction, the out-neighbour of $s$ on $P$ is $x^{2i}_{2j}$, and 
the in-neighbour of $t$ is $y^{2i'}_{2j'+1}$, where $x$ and $y$ are terminals
such that $\ordering(x) = (i,j)$ and $\ordering(y) = (i', j')$. 
Let $W_P$ be the $(x, y)$-walk in $G - Z$ formed by
the arcs $\{e : P \cap B_e \neq \emptyset\}$
traversed in the same order as $P$ does.
We consider two cases, and in each case deduce
that either $\ordering$ is incompatible with $G - Z$,
or that $G - Z$ contains a short or a directed long cycle.

Suppose $i = i'$.
If $W_P$ contains a long arc,
then, by construction of $H$,
all long arcs are in the same direction.
Since $G - Z$ has no directed long cycles,
$x$ reaches $y$ by a directed long path.
However, we have $i \nless i'$, 
which implies that $\ordering$ is incompatible with $G - Z$.
Now assume $W_P$ has no long arcs.
Then $W_P$ is a directed short walk:
indeed, if $W_P$ uses a short arc against its direction,
then $P$ contains a vertex with lower index $1$,
and the lower index of all subsequent vertices is also $1$
by construction of $H$;
however, $P$ reaches $y^{2i}_{2j'+1}$ and $2j' + 1 > 1$.
If $x = y$, then $W_P$ is a closed directed short walk,
and it contains a directed short cycle.
If $x \neq y$, then $x$ reaches $y$ in $G - Z$
by a directed short walk, hence $j < j'$
by the definition of $\ordering$.
However, indices of the vertices on the path $P$
cannot exceed $(2i, 2j+1)$ by construction of $H$,
hence $(2i', 2j'+1) \leq (2i, 2j+1)$ and $j' \leq j$.

Suppose $i \neq i'$.
By construction of $H$,
since $x^{2i}_{2j}$ appears on $P$,
the index of every other variable on $P$
is at most $(2i+1, 1)$.
Hence, $(2i + 1, 1) \geq (2i', 2j'+1)$ and $i \geq i'$.
Together with $i \neq i'$, this implies $i > i'$.
Clearly, $x \neq y$ in this case,
and the walk $W_P$ contains a long arc.
By construction of $H$, all long arcs on $W_P$
are in the same direction, i.e. from $x$ to $y$.
However, this implies that $\ordering$
is incompatible with $G - Z$ since
there is a directed long $(x,y)$-path in this graph,
$\ordering(x) = (i, j)$, $\ordering(y) = (i', j')$,
but $i > i'$.

\smallskip

\noindent
{\bf Backward direction.}
Now suppose $W$ is a solution to 
$(H := H(G, \ordering), \bundles := \bundles(G, \ordering), k)$
for some ordering $\ordering$.
Define $W' \subseteq E(G) \cup A(G)$
as the set of all edges and arcs $e$ such that 
$B_e \cap W \neq \emptyset$.
We claim that $W'$ is a solution to $(G, A_s, A_\ell, k)$.
Clearly, $|W'| \leq |W| \leq k$ since $W$ intersects
at most $k$ bundles in $\bundles$.
Now suppose $G - W'$ contains a bad cycle $C$.
Observe that $C$ contains at least one terminal.
We let $x$ and $y$ be two (not necessarily distinct)
terminals on $C$.
We have two cases to consider, depending on 
whether $C$ is long or short.
In each case, we deduce that $H - W$ contains an $(s,t)$-path.

In case $C$ is directed short,
assume by symmetry that $C$ consists of
a directed short $(x, y)$-path and 
a $(y,x)$-path that is either undirected or directed short.
Let $\ordering(x) = (i, j)$ and $\ordering(y) = (i', j')$.
Since $x$ reaches $y$ by a directed short path,
$H - W$ contains a path
$$s \to x_{2i,2j} \to \dots \to y_{2i,2j+1}.$$
If $y_{2i,2j+1}$ reaches $t$ in $H - W$, then we are done.
If not, then $x \neq y$.
Since $y$ reaches $x$ by either an undirected
or a directed short path in $G -W'$,
we have a path 
$$s \to \dots \to y_{2i,2j+1} \to \dots \to x_{2i,2j+1} \to t$$ 
in $H - W$.

In case $C$ is directed long,
assume by symmetry that $C$ consists of
a directed long $(x, y)$-path
and a $(y, x)$-path that is either short or directed long.
Let $\ordering(x) = (i, j)$ and $\ordering(y) = (i', j')$.
Since $G - W'$ contains a directed long $(x, y)$-path,
$H - W$ contains a path
$$s \to x_{2i,2j} \to \dots \to y_{2i+1,1}.$$
If $y_{2i+1, 1}$ reaches $t$, then we are done.
Otherwise, $x \neq y$.
Since $y$ reaches $x$ in $G - W'$ without using long arcs against their direction,
$H - W$ contains a path
$$s \to \dots \to y_{2i+1,1} \to \dots \to x_{a,b} \xrightarrow{\text{down}} x_{2i,2j+1} \to t$$
where $(a, b) \geq (2i+1, 1) > (2i, 2j+1)$.
\end{proof} 
\fi

\ifshort

Lemma~\ref{lem:bad-cycles}.4 suggests that all six remaining fragments 
allow uniform treatment.
Indeed, to check whether an instance of $\csp{\rr_1, \rr_2, \e}$
is consistent, one can identify 
all variables constrained to be equal.
This corresponds exactly to contracting all edges in the graph $G_\III$.
Then $\III$ becomes an instance of
$\csp{\rr_1, \rr_2}$, and the criterion of 
Lemma~\ref{lem:bad-cycles}.4 applies.
This observation allows us to formulate $\mincsp{\rr_1, \rr_2, \e}$
as a variant of feedback arc set on mixed graphs.

\begin{definition}
Consider a mixed graph $G$ with arcs of two types -- short and long --
and a walk $W$ in $G$ from $u$ to $v$ that may ignore direction of the arcs.
The walk $W$ is \emph{undirected} if it only contains edges, it is
\emph{short} if it contains a short arc but no long arcs, and it is
\emph{long} if it contains a long arc.
The walk $W$ is \emph{directed} if it is either short and all short arcs are directed from $u$ to $v$
or if it is long and all long arcs are directed from $u$ to $v$.
If $W$ is short or long, but not directed, it is \emph{mixed}.
\end{definition}

Note that short arcs on a directed long walk may be directed arbitrarily.

\pbDefP{Mixed Feedback Arc Set with Short and Long Arcs (\lsmfas)}
{A mixed graph $G$ with the arc set $A(G)$ partitioned into
\emph{short} $A_s$ and \emph{long} $A_\ell$, and an integer $k$.}
{$k$.}
{Is there a set $Z \subseteq E(G) \cup A(G)$ such that
$G - Z$ contains neither short directed cycles nor long directed cycles?}

The main result of this section is the following theorem.

\begin{theorem}\label{thm:lsmfas-algo}
  \lsmfas can be solved in time $O^*(2^{O(k^8 \log k)})$.
\end{theorem}

\iflong
Lemma~\ref{lem:sf-pod-fpt} positing that
$\mincsp{\rr_1, \rr_2, \e}$ is in FPT whenever
$\rr_1 \in \{\s, \f\}$ and $\rr_2 \in \{\p, \o, \d\}$
is an immediate consequence of 
Lemma~\ref{lem:bad-cycles}.4 and Theorem~\ref{thm:lsmfas-algo}.
\fi
\ifshort
We see that 
Lemma~\ref{lem:bad-cycles}.4 and Theorem~\ref{thm:lsmfas-algo} imply
$\mincsp{\rr_1, \rr_2, \e}$ being in FPT whenever
$\rr_1 \in \{\s, \f\}$ and $\rr_2 \in \{\p, \o, \d\}$.
\fi
It is informative to understand the structure of mixed graphs
without bad cycles in the sense of \lsmfas.
\iflong
One way of characterizing \lsmfas (with a pronounced connection to interval
constraints) is as a problem of placing 
vertices of the input graph $G$ on the number line
under certain distance constraints represented by the edges and arcs.
An edge $uv$ requires $u$ and $v$ to be placed at the same point.
An arc $(u, v)$ requires that $u$ is placed before $v$.
If the arc is long, then the distance from $u$ to $v$ should be big
(say, $2|V(G)|$ units of length), while if the arc is short,
the distance should be small (say, at most a unit of length).
\fi
The proof of the following lemma is fairly easy with 
the placing-vertices-on-the-number-line intuition from the introduction.

\begin{lemma} \label{lem:lsmfas-struct}
  Let $G$ be a mixed graph with long and short arcs.
  Then $G$ contains no directed long cycles nor directed short cycles if and only if there exists a pair of mappings
  $\sigma_1,\sigma_2 : V(G) \to \mathbb{N}$ such that
  \begin{enumerate}
    \item for every $u,v \in V(G)$, $u$ and $v$ are connected by an undirected walk if and only if $(\sigma_1,\sigma_2)(u) = (\sigma_1,\sigma_2)(v)$;
    \item for every $u,v \in V(G)$, there exists a short $(u,v)$-walk in $G$ if and only if $\sigma_1(u) = \sigma_1(v)$;
    \item for every $u,v \in V(G)$, if there exists a directed short $(u,v)$-walk in $G$, then $\sigma_2(u) < \sigma_2(v)$;
    \item for every $u,v \in V(G)$, if there exists a directed long $(u,v)$-walk in $G$, then $\sigma_1(u) < \sigma_1(v)$.
  \end{enumerate}    
\end{lemma}

Now we introduce the technical machinery used in our algorithm for \lsmfas.
We start by iterative compression, a standard method in parameterized 
algorithms~(see~e.g.~Chapter~4~in~\cite{cygan2015parameterized}).
\iflong
The idea is to solve the problem starting from an empty graph
and iteratively build up the graph by add edges and arcs, maintaining a solution of size $k$ in the process.
We clearly have a solution while the size of the graph is at most $k$.
Now consider an iteration where we have a graph $G$ together with a solution $Z$, and a new edge/arc $e$ is added.
Observe that $Z \cup \{e\}$ is a solution to $(V(G), E(G) \cup \{e\})$, thus we may assume access
to an approximate solution of size $k + 1$ at each iteration.

\fi
\ifshort
It allows us to assume access to a set of $k + 1$ edges and arcs
intersecting every bad cycle.
\fi
The problem resulting from iterative compression reduces to {\sc Bundled Cut} 
with pairwise linked deletable edges,
defined in~\cite{Kim:etal:FA3} and solved using the flow-augmentation technique of~\cite{Kim:etal:stoc2022}.
To describe \textsc{Bundled Cut}, let $G$ be a directed graph with two distinguished vertices $s,t \in V(G)$.
Let $\bundles$ be a family of pairwise disjoint subsets of $E(G)$, which we call \emph{bundles}.
The edges of $\bigcup \bundles$ are called \emph{soft} and the edges of $E(G) \setminus \bigcup \bundles$
are \emph{crisp}. A set $Z \subseteq \bigcup \bundles$ \emph{violates} a bundle $B \in \bundles$
if $Z \cap B \neq \emptyset$ and \emph{satisfies} $B$ otherwise. 

\pbDefP{Bundled Cut}
{A directed graph $G$, two distinguished vertices $s,t \in V(G)$, 
  a family $\bundles$ of pairwise disjoint subsets of $E(G)$, and an integer $k$.}
{$k$.}
{Is there an $st$-cut $Z \subseteq \bigcup \bundles$ that violates at most $k$ bundles?}

In general, \probbc is W[1]-hard even if all bundles are of size $2$. However,
there is a special case of \probbc that is tractable.
Let $(G,s,t,\bundles,k)$ be a \probbc instance. 
A soft arc $e$ is \emph{deletable} if there is no crisp copy of $e$ in $G$.
An instance $(G,s,t,\bundles,k)$ has \emph{pairwise linked deletable arcs}
if for every $B \in \bundles$ and every two deletable arcs $e_1,e_2 \in B$, there exists
in $G$ a path from an endpoint of one of the arcs $e_1,e_2$ to an endpoint of the second of those arcs
that does not use any arcs of $\bundles \setminus \{B\}$.
The assumption of pairwise linked deletable arcs makes \probbc tractable.
\begin{theorem}[Theorem~4.1 of~\cite{kim2023flowaugmentation}]\label{thm:probbc}
\probbc instances with pairwise linked deletable arcs can be solved
in time $O^*(2^{O(k^4 d^4 \log(kd))})$ where $d$ is the maximum number of deletable arcs in a single bundle.
\end{theorem}

Armed with Lemma~\ref{lem:lsmfas-struct} and Theorem~\ref{thm:probbc}, 
we are ready to prove the main result.

\begin{proof}[Proof of Theorem~\ref{thm:lsmfas-algo}]
Let $(G, A_s, A_\ell, k)$ be an instance of \lsmfas.
By iterative compression, we can assume that 
we have access to a set 
\iflong $X \subseteq E(G) \cup A(G)$ \fi 
\ifshort $Y \subseteq V(G)$ \fi
of size at most $k+1$ that intersects all bad cycles.
\iflong
Subdivide every edge and arc in $X$,
i.e. for every edge $uv \in X$ (arc $(u,v) \in X$, respectively),
replace it by edges $uy_{uv}$ and $y_{uv}v$
(arcs $(u, y_{uv})$ and $(y_{uv}, v)$, respectively),
where $y_{uv}$ is a fresh variable.
Let $Y = \{ y_{uv} : uv \in X \text{ or } (u, v) \in X \}$.
Note that $Y \subseteq V(G)$ intersects all bad cycles in $G$.
Moreover, the resulting graph has a set of $k$ edges
intersecting all bad cycles if and only if
the original graph has one.
\fi
We refer to the vertices of $Y$ as \emph{terminals}.

Fix a hypothetical solution $Z \subseteq A(G) \cup E(G)$.
Guess which pairs of terminals are
connected by undirected paths in $G - Z$ and identify them.
Define \emph{ordering} $\ordering : Y \to \naturals \times \naturals$ 
that maps terminals to
\begin{equation*}
  \{
  (1, 1), \dots, (1, q_1), \cdots,
  (i, 1), \dots, (i, q_i), \cdots,
  (p, 1), \dots, (p, q_p)
  \}
\end{equation*}
\iflong
such that the following hold.
For every pair of terminals $y,y' \in Y$,
let $\ordering(y) = (i, j)$ and
$\ordering(y') = (i', j')$.
\begin{itemize}
  \item If $y$ and $y'$ are connected by 
  a short path in $G - Z$, then $i = i'$.
  \item If $y$ reaches $y'$ by 
  a directed short path in $G - Z$, then $j < j'$.
  \item If $y$ reaches $y'$ by
  a directed long path in $G - Z$, then $i < i'$.
\end{itemize}
\fi
\ifshort
such that the following hold.
For every pair of terminals $y,y' \in Y$,
let $\ordering(y) = (i, j)$ and
$\ordering(y') = (i', j')$ where (1) $i=i'$ if
$y$ and $y'$ are connected by 
  a short path in $G - Z$, (2) $j < j'$ if
  $y$ reaches $y'$ by 
  a directed short path in $G - Z$, and (3) $i < i'$
  if $y$ reaches $y'$ by
  a directed long path in $G - Z$.
\fi
Note that $\ordering$ exists by Lemma~\ref{lem:lsmfas-struct}.
If an ordering satisfies the conditions above,
we say that it is \emph{compatible with} $G - Z$.
In the sequel we write $(i,j) < (i',j')$
to denote that $(i, j)$ 
lexicographically precedes $(i',j')$,
i.e. either $i = i'$ and $j < j'$ or $i < i'$.
\iflong
\begin{figure}
  \begin{center}
  \includegraphics[scale=0.6, trim = {3cm, 2cm, 2cm, 3cm}]{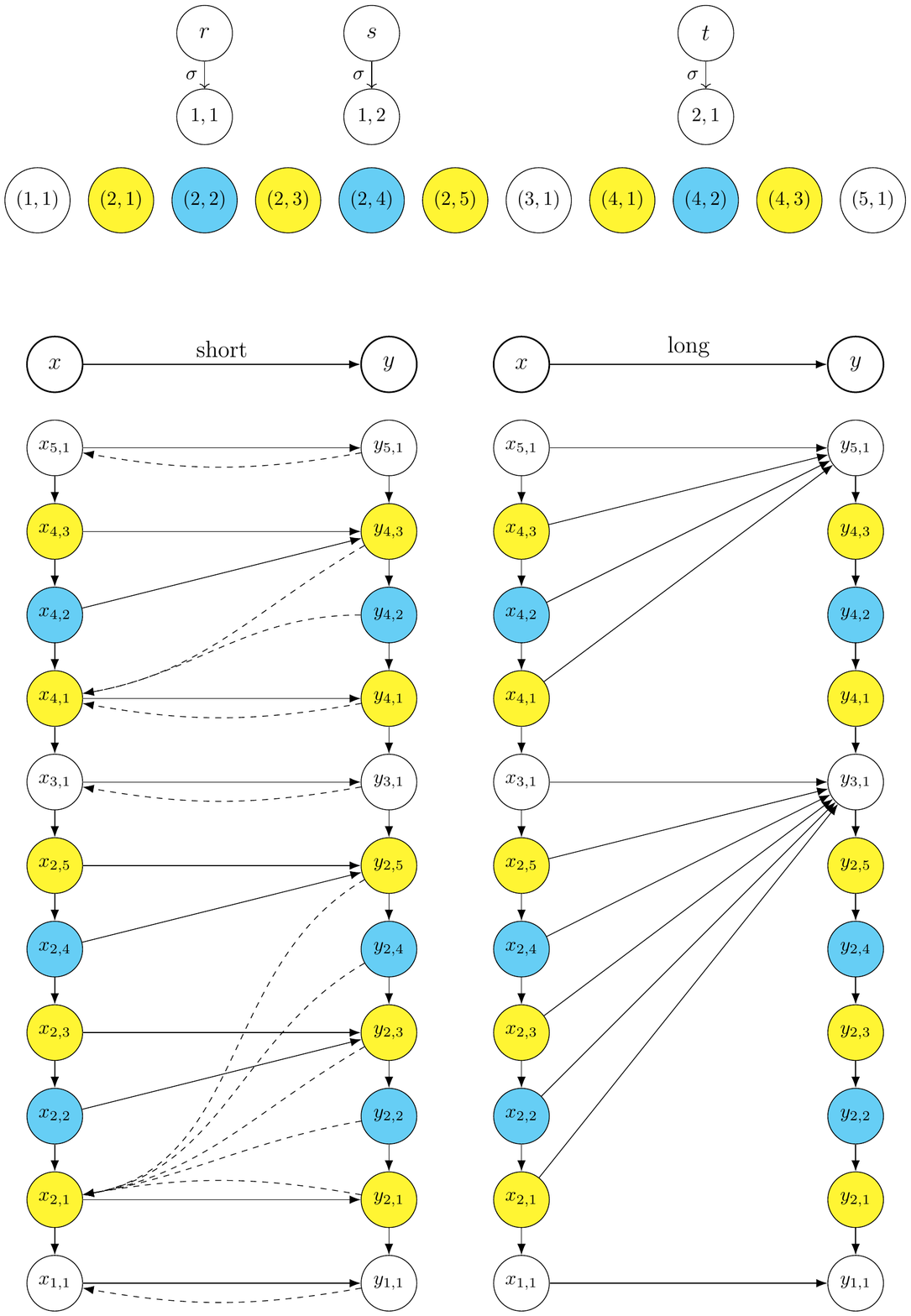}
  \caption{An example with three terminals $r,s,t$ where
  $\ordering(r) = (1,1)$,
  $\ordering(s) = (1,2)$,
  $\ordering(t) = (2,1)$.}
  \label{fig:drawing}
  \end{center}
\end{figure}
\fi

\begin{table}
  \begin{center}
    \begin{tabular}{| c || c | c | c |} 
      \hline
                                        & $i$ odd            & $i$ even, $j$ odd          & $i$ even, $j$ even  \\ 
      \hline
      \multirow{2}{*}{Edge $uv$}        & $(i,j) \to (i,j)$ &    $(i,j) \to (i,j)$ & $(i,j) \to (i,j)$                           \\
                                        & $(i,j) \leftarrow (i,j)$ &    $(i,j) \leftarrow (i,j)$ &       $(i,j) \leftarrow (i,j)$              \\
      \hline
      \multirow{2}{*}{Short $(u,v)$}    & $(i,j) \to (i,j)$ &  $(i,j) \to (i,j)$        & $(i,j) \to (i,j+1)$ \\
                                       & $(i,1) \leftarrow (i,j)$ &   $(i,1) \leftarrow (i,j)$      & $(i,1) \leftarrow (i,j)$ \\
      \hline
      Long $(u,v)$                      & $(i,1) \to (i,1)$  & $(i,j) \to (i+1,1)$ &    $(i,j) \to (i+1,1)$     \\            
      \hline
    \end{tabular}  
  \end{center}
  \caption{Correspondence between edges, short arcs and long arcs of the \lsmfas instance
  and the arcs introduced in the reduction to \textsc{Bundled Cut} in Theorem~\ref{thm:lsmfas-algo}. }
  \label{tb:reduction}
\end{table}

For the algorithm,
proceed by guessing an ordering $\ordering$, 
creating $2^{O(k \log k)}$ branches in total.
For each $\ordering$, create an instance 
$(H := H(G, \ordering), \bundles := \bundles(G, \ordering), k)$ of \probbc as follows.
Introduce two distinguished vertices $s$ and $t$ in $H$.
For every vertex $v \in V(G)$, create vertices $v^{i}_{1}$ in $H$  
for all odd $i \in [2p + 1]$ and vertices $v^{i}_{j}$ in $H$
for all even $i \in [2p + 1]$ and all $j \in [2q_i + 1]$.
Connect vertices created above by \emph{downward} arcs
$(v^{i}_{j}, v^{i'}_{j'})$ for all $(i, j) > (i', j')$.
For every terminal $y$, let $\sigma(y) = (i, j)$, and
add arcs $(s, y^{2i}_{2j})$ and $(y^{2i}_{2j+1}, t)$ in $H$.
Using the rules below, create a bundle $B_e$ in $\bundles$
for every $e \in E(G) \cup A(G)$,
add the newly created arcs to $H$.
\begin{itemize}
  \item For an edge $e = uv$, let $B_e$ consist of
  arcs $(u^{i}_{j}, v^{i}_{j})$ and $(v^{i}_{j}, u^{i}_{j})$ for all $(i, j)$.
  \item For short arcs $e = (u, v)$, let $B_e$ consist of
  arcs $(u^{i}_{j}, v^{i}_{j})$ for all $i, j$ such that $i$ or $j$ is odd,
  arcs $(u^{i}_{j}, v^{i}_{j+1})$ for all even $i, j$, and
  arcs $(v^{i}_{j}, u^{i}_{1})$ for all $i,j$.  
  \item For long arcs $e = (u, v)$, let $B_e$ consist of
  arcs $(u^{i}_{1}, v^{i}_{1})$ for all odd $i$, and arcs
  arcs $(u^{i}_{j}, v^{i+1}_{1})$ for all even $i$ and all $j$.
\end{itemize}
This completes the construction.
\iflong
Bundle construction rules are summarized in Table~\ref{tb:reduction},
and Figure~\ref{fig:drawing} contains an illustrated example.
\fi
\ifshort
Bundle construction rules are summarized in Table~\ref{tb:reduction}.
The full version of the paper features a full-page illustrated example.
\fi
Observe that the downward arcs ensure 
that $(H, \bundles, k)$ has the pairwise-linked deletable arc property.
Moreover, the bundle size is $O(k)$, so we can
solve $(H, \bundles, k)$ in $O^*(2^{O(k^8 \log{k})})$ time.

A correctness proof can be found in the full version of the paper.
Here, we only give some intuition. 
Fix a guessed ordering $\sigma$. 
For any candidate solution $W$ in $(H := H(G,\sigma), \bundles = \bundles(G, \sigma), k)$,
the existence of downward arcs imply that for every $v \in V(G)$ there is a threshold
$(i_v, j_v)$ such that $v_i^j$ is reachable from $s$ in $H-W$ if and only if
$(i,j) \leq (i_v,j_v)$. This threshold is meant to indicate that $v$ should be placed
on the line somewhere around the terminal $x$ for which $\sigma(x) = (\lfloor i_v/2 \rfloor, \lfloor j_v/2 \rfloor)$. 
A short walk from $u$ to $v$ in $G$ projects, for every even $i$ and even $j$, to a walk
from $u_i^j$ to $v_i^{j+1}$. 
A long walk from $u$ to $v$ in $G$ projects, for every odd $i$, to a walk from $u_i^1$ to $v_i^1$.
Together with the fact that terminals intersect all forbidden cycles in $G$, this
gives a correspondence between forbidden cycles in $G$ and $st$-paths in $H$.
\end{proof} 
\fi

\section{W[1]-hard Problems} \label{sec:w-hard}

\ifshort
Here, we show Lemmas~\ref{lem:me-ms-mf-hard} and \ref{lem:do-po-dp-hard}. As the first and most challenging step, we show \W{1}-hardness for variants of paired and simultaneous graph cut problems from which we then reduce to the hard variants of \textsc{MinCSP}$(\Gamma)$.
\fi
\iflong
Here, we show Lemmas~\ref{lem:me-ms-mf-hard} and \ref{lem:do-po-dp-hard}. Our \W{1}-hardness results for various \textsc{MinCSP}$(\Gamma)$ problems are
based on reductions from graph problems.
These graph problems fall into two categories: {\em paired} and
{\em simultaneous} cut problems.
In paired cut problems, the input consists of two graphs
together with a pairing of their edges, and
the goal is to compute cuts in both graphs 
using at most $k$ pairs.
While the problems on individual graphs are solvable in polynomial time,
the pairing requirement leads to \W{1}-hardness.
In simultaneous cut problems, the two input graphs
share some arcs/edges that 
can be deleted simultaneously at unit cost, and
the goal is to compute a set of $k$ arcs/edges
that is a cut in both graphs.
As for the paired problems, computing cuts for individual graphs
is feasible in polynomial time.
However, the possibility of choosing
common arcs/edges while deleting them at a unit cost
correlates the choices in one graph and the other,
and leads to \W{1}-hardness.
The hardness results for paired and simultaneous cut problems
can be found in Sections~\ref{ssec:pair-graph} and \ref{ssec:sim-graph}, respectively.
We use paired cut problems for proving hardness
of $\mincsp{\m, \e}$ and $\mincsp{\m,\s}$
in Section~\ref{ssec:me-ms-hard}, and
we use simultaneous cut problems for proving
hardness of
$\mincsp{\d, \p}$, $\mincsp{\d, \o}$, and $\mincsp{\p, \o}$,
in Section~\ref{ssec:do-dp-op-hard}.
\fi
Our reductions will make use of the following well-known problem, whose \W{1}-hardness follows by a simple reduction from \textsc{Multicolored Clique}~(see~e.g.~Exercise~13.3~in~\cite{cygan2015parameterized}).

\pbDefP{Multicolored Biclique (\mbicl)}
{An undirected graph $G$ with a partition $V(G) =
  A_1 \uplus \ldots \uplus A_k \uplus B_1 \uplus \ldots \uplus B_k$,
  where $|A_i|=|B_i|=n$ for each $i \in [k]$ and both $\uplus_{i \in
    [k]}A_i$ and $\uplus_{i \in [k]}B_i$ form independent sets in $G$.}
{$k$.}
{Does $G$ contain $K_{k,k}$ as a subgraph, a.k.a. a \emph{multicolored biclique}?}

\iflong
\begin{proposition} \label{prop:biclique-hard}
  \mbicl\ is \textnormal{\W{1}}-hard.
\end{proposition}
\fi

\subsection{Paired Problems} \label{ssec:pair-graph}

\newcommand{\pc}{\textsc{Paired Cut}}
\newcommand{\pcfas}{\textsc{PCFAS}}

We consider the problems \pc\ and \textsc{Paired Cut Feedback Arc Set} (\pcfas) in the sequel.

\pbDefP{Paired Cut}
{Undirected graphs $G_1$ and $G_2$, vertices $s_i, t_i \in V(G_i)$,
a set of disjoint edge pairs 
$\bundles \subseteq E(G_1) \times E(G_2)$,
and an integer $k$.}
{$k$.}
{Is there a subset $X \subseteq \bundles$ such that
$|X| \leq k$ and $X_i=\{ e_i \; | \; \{e_1,e_2\} \in X \}$ 
is an $\{s_i,t_i\}$-cut in $G_i$ for both $i \in \{1,2\}$?} 

\ifshort
The \pcfas\ problem is similar, but $G_2$ is directed and $X_2$
is required to be such that $G_2-X_2$ is acyclic (instead of being an $\{s_2,t_2\}$-cut).
\fi
\iflong
\pbDefP{Paired Cut FAS (\pcfas)}
{An undirected graph $G_1$ and a digraph $G_2$, vertices $s, t \in V(G_1)$,
a set of disjoint edge pairs 
$\bundles \subseteq E(G_1) \times E(G_2)$,
and an integer $k$.}
{$k$.}
{Is there a subset $X \subseteq \bundles$ such that
$|X| \leq k$, $X_1$
is an $\{s,t\}$-cut in $G_1$ and $G_2-X_2$ is acyclic,
where $X_i=\{ e_i \; | \; \{e_1,e_2\} \in X \}$ for both $i \in \{1,2\}$?}
\fi
We show \W{1}-hardness of both problems.
Since both reductions are from \mbicl{} and quite similar, we start to
decribe the common part of both reductions.
Let $I=(G,A_1,\dotsc,A_k,B_1,\dotsc,B_k,k)$ be an instance of
\mbicl{}. We define two directed graphs $G_A$ and $G_B$ as follows.
$G_A$ contains the vertices $s_A$ and $t_A$. Moreover, for every $i \in [k]$, $G_A$
contains the vertices in $P_i^A=\{v_{i,1},\dotsc,v_{i,n-1}\}$. For
convinence, we use that $v_{i,0}=s_A$ and $v_{i,n}=t_A$ for every
$i \in [k]$. Moreover, for every vertex $a_{i,j}$ and
every $i' \in [k]$, $G_A$ contains the directed path $P^A_{i,j,i'}$ from
$v_{i,j-1}$ to $v_{i,j}$ that has one edge (using fresh auxiliary
vertices) for every edge between $a_{i,j}$ and a vertex in
$B_{i'}$. Therefore, we can assume in the following that there is a
bijection between the edges of $P_{i,j,i'}^A$ and the edges between
$a_{i,j}$ and a vertex in $B_{i'}$. This concludes the description of
$G_A$.
\iflong \fi
$G_B$ is defined very similarily to $G_A$ with the roles of the
sets $A_1,\dotsc,A_k$ and $B_1,\dotsc,B_k$ being reversed. \iflong In
particular, $G_B$ contains the vertices $s_B$ and $t_B$ and for every $i \in [k]$, $G_B$
contains the vertices in $P_i^B=\{u_{i,1},\dotsc,u_{i,n-1}\}$. For
convinence, we will assume that $u_{i,0}=s_B$ and $u_{i,n}=t_B$ for every
$i \in [k]$. Moreover, for every vertex $b_{i,j}$ and
every $i' \in [k]$, $G_B$ contains the directed path $P_{i,j,i'}^B$ from
$u_{i,j-1}$ to $u_{i,j}$ that has one edge (using novel auxiliary
vertices) for every edge between $b_{i,j}$ and a vertex in
$A_{i'}$. Therefore, we can assume in the following that there is a
bijection between the edges of $P_{i,j,i'}^B$ and the edges between
$b_{i,j}$ and a vertex in $A_{i'}$.\fi

Finally, define a set $\bundles\subseteq E(G_A)\times E(G_B)$ of
bundles as follows. For
every edge $e=\{a_{i,j},b_{i',j'}\} \in E(G)$, $\bundles$ contains the
pair $(e^A,e^B)$, where $e^A$ is the edge corresponding to $e$ on the
path $P_{i,j,i'}^A$ and $e^B$ is the edge corresponding to $e$ on the
path $P_{i',j',i}^B$. This concludes the construction and the following
lemma shows its main property.
\begin{lemma}\label{lem:pc-hard-aux}
  $I=(G,A_1,\dotsc,A_k,B_1,\dotsc,B_k,k)$ is a yes-instance of \mbicl{} if and
  only if there is a set $X\subseteq \bundles$ with $|X|=k^2$ and
  $X_A=\{ e \mid (e,e') \in X\}$ ($X_B=\{ e' \mid (e,e') \in X\}$) is
  an $(s_A,t_A)$-cut ($(s_B,t_B)$-cut) in $G_A$ ($G_B$).
\end{lemma}
\iflong
\begin{proof}
  {\bf Forward direction.} Let $K$ be a biclique in $G$
  with vertices
  $$a_{1,j_1},\dotsc,a_{k,j_k},b_{1,\ell_1},\dotsc,b_{k,\ell_k}.$$
  Let $X \subseteq \bundles$ be the set of bundles
  containing the pair $(e^A,e^B)$ for every edge $e$ of $K$. Then,
  $|X|=k^2$ and we claim that $X_A$ ($X_B$) is an $(s_A,t_A)$-cut ($(s_B,t_B)$-cut)
  in $G_A$ ($G_B$). This holds because:
  \begin{itemize}
  \item For every $a_{i,j_i}$ and $i' \in [k]$, $X_A$ contains an edge
    from the path $P_{i,j_i,i'}^A$ and therefore there is no path from
    $a_{i,j_i-1}$ to $a_{i,j_i}$ in $G_A-X_A$.
  \item For every $b_{i,\ell_i}$ and $i' \in [k]$, $X_B$ contains an edge
    from the path $P_{i,\ell_i,i'}^B$ and therefore there is no path from
    $b_{i,\ell_i-1}$ to $b_{i,\ell_i}$ in $G_B-X_B$.    
  \end{itemize}

\medskip

\noindent
  {\bf Backward direction.} Let $X \subseteq \bundles$
  with $|X|=k^2$ such that
  $X_A=\{ e \mid (e,e') \in X\}$ ($X_B=\{ e' \mid (e,e') \in X\}$) is
  an $(s,t)$-cut in $G_A$ ($G_B$). Then, for every $i \in [k]$
  there must exist a $j \in[k]$ such that for every $i' \in [k]$, $X_A$ contains an edge
  from the path $P_{i,j,i'}^A$. Similarily, for every $i \in [k]$
  there must exist a $j \in[k]$ such that for every $i' \in [k]$, $X_B$ contains an edge
  from the path $P_{i,j,i'}^B$. Because $|X|=k^2$, $X_A$ and $X_B$ cannot
  contain any other edges. Therefore, for every $i \in [k]$ there is a
  unique vertex $a_{i,j_i}$ in $A_i$ such that all paths between
  $a_{i,j_i-1}$ and $a_{i,j_i}$ are cut in $G_A-X_A$. Similarily,
  for every $i \in [k]$ there is a
  unique vertex $b_{i,\ell_{i}}$ in $B_{i'}$ such that all paths between
  $b_{i,\ell_{i}-1}$ and $b_{i,\ell_{i}}$ are cut in $G_B-X_B$. But
  due to the definition of $\bundles$, this is only possible if there is an edge in $G$ between every
  vertex $a_{i,j_i}$ and every vertex $b_{i,\ell_i}$ so
  $\{a_{1,j_1},\dotsc,a_{k,j_k},b_{1,j_1},\dotsc,b_{k,j_k}\}$ is indeed a
  biclique in $G$.
\end{proof}
\fi

The lemma above makes it relatively straightforward to show \W{1}-hardness for \pc{} and \pcfas{}.

\ifshort
\begin{lemma}\label{lem:pcpcfas-hard} 
  \pc{} and \pcfas{} are \textnormal{\W{1}}-hard.
\end{lemma}
\fi
\iflong
\begin{lemma}\label{lem:pc-hard}%
  \pc{} is \textnormal{\W{1}}-hard.
\end{lemma}
\begin{proof}
  Let $I=(G,A_1,\dotsc,A_k,B_1,\dotsc,B_k,k)$ be an instance of
  \mbicl\ and let $G_A$, $G_B$, and $\bundles$ be the corresponding
  gadgets defined in Lemma~\ref{lem:pc-hard-aux}. First note that Lemma~\ref{lem:pc-hard-aux}
  still holds if we consider the underlying undirected graphs $G_A'$
  and $G_B'$ of $G_A$ and $G_B$, respectively. Therefore, using the
  same lemma, we obtain that $I$ is equivalent to the instance
  $(G_A',G_B',s_A,t_A,s_B,t_B,\bundles,k)$ of \pc{}, as required.
\end{proof}

\begin{lemma}\label{lem:pcfas-hard}
    \pcfas{} is \textnormal{\W{1}}-hard.
\end{lemma}
\begin{proof}
  Let $I=(G,A_1,\dotsc,A_k,B_1,\dotsc,B_k,k)$ be an instance of
  \mbicl\ and let $G_A$, $G_B$, and $\bundles$ be the corresponding
  gadgets defined in Lemma~\ref{lem:pc-hard-aux}. First note that Lemma~\ref{lem:pc-hard-aux}
  still holds if we consider the underlying undirected graphs $G_A'$
  of $G_A$. Let $G_B'$ be the directed graph obtained from $G_B$ after
  adding the edge from $t_A$ to $s_A$. Because $G_B$ is acyclic, it
  holds that a set $X_B$ of edges of $G_B$ is an $(s_B,t_B)$-cut if
  and only if $X_B$ is a FAS for $G_B'$. Therefore, using
  Lemma~\ref{lem:pc-hard-aux}, we obtain that $I$ is equivalent to the
  instance $(G_A,G_B',s_A,t_A,\bundles,k)$ of \pcfas{}, as required.
\end{proof}
\fi

\subsection{Simultaneous Problems} \label{ssec:sim-graph}

In this section we prove \W{1}-hardness of several simultaneous cut problems.
Our basis is the following problem.

\pbDefP{Simultaneous Separator (\simsep{})}
{Two directed graphs $D_1$ and $D_2$ with $V=V(D_1) = V(D_2)$, vertices $s,t \in V$, and an integer $k$.}
{$k$.}
{Is there a subset $X \subseteq V\setminus \{s,t\}$ of size at most $k$ such that
 neither $D_1-X$ nor $D_2-X$ contain a path from $s$ to $t$?}

We begin by proving that this problem is W[1]-hard in Theorem~\ref{thm:sim-sep-hard}.
We will then prove that simultaneous variants of {\sc Directed st-Cut}
and {\sc Directed Feedback Arc Set} are W[1]-hard via reductions from \simsep{}; these results can be found in
Theorems~\ref{thm:sim-cut-hard} and \ref{thm:sim-dfas-hard}, respectively.
It will be convenient to use the term $st$-{\em separator} when working with directed graphs:
given a directed graph $G=(V,E)$ and two vertices $s,t \in V$,
we say that $X \subseteq V \setminus \{s,t\}$ is an $st$-{\em separator} if 
the graph $G - X$ contains no directed path from $s$ to $t$ and
no directed path from $t$ to $s$.

\begin{theorem} \label{thm:sim-sep-hard}
  \simsep{} is \textnormal{\W{1}}-hard even if both input digraphs are acyclic.
\end{theorem}
\begin{proof}

Given an instance $(G,k)$ of \mbicl{} we construct an instance $(D_1,D_2,\sss,\ttt,k')$ of  \simsep{} in polynomial time such that $k' =O(k^2)$ and $(G,k)$ is a \yes-instance of \mbicl{} if and only if $(D_1,D_2,\sss,\ttt,k')$ is a \yes-instance of \simsep{}. Moreover, the digraphs $D_1$ and $D_2$ are both acyclic.

\iflong
\paragraph*{Construction.} 
\fi
Recall that $V(G) = A_1 \uplus \cdots \uplus A_k \uplus B_1 \uplus \cdots \uplus B_k$.
Let $A_i= \{a^i_1, \ldots, a^i_n\}$ and $B_i=\{b^i_1, \ldots, b^i_n\}$. 
The digraphs $D_1,D_2$ (as per the problem definition) should be
defined on the same vertex set. For convenience we define the vertex
set of $D_1$ and $D_2$ separately, state which of the vertices are
common to both, and for the others we provide a one-to-one-map between those in
$D_1$ and those in $D_2$. The vertex set of the digraphs $D_1, D_2$
are described in the following five steps.

\medskip

\iflong
\noindent
{\bf (V-1)} Corresponding to each $A_i$, $i \in [k]$, there
are two vertex sets in $D_1$ and $D_2$: 
$\Ac^i =\{\alpha^i_1, \ldots, \alpha^i_n\}$ and a set of dummy vertices 
$$\D^{i,A} =\{d^{i,A}_{1,\inn}, d^{i,A}_{1,\out}, \ldots, d^{i,A}_{n,\inn}, d^{i,A}_{n,\out}\}.$$ (The vertex sets $\Ac^i$ and $\D^{i,A}$ are common in $D_1$ and $D_2$.)
\fi

\ifshort
\noindent
{\bf (V-1)} Corresponding to each $A_i$, $i \in [k]$, there
are two vertex sets in $D_1$ and $D_2$: 
$\Ac^i =\{\alpha^i_1, \ldots, \alpha^i_n\}$ and a set of dummy vertices 
$\D^{i,A} =\{d^{i,A}_{1,\inn}, d^{i,A}_{1,\out}, \ldots, d^{i,A}_{n,\inn}, d^{i,A}_{n,\out}\}.$ The vertex sets $\Ac^i$ and $\D^{i,A}$ are common in $D_1$ and $D_2$.
\fi

\smallskip

\noindent
{\bf (V-2)} \iflong Corresponding to each $B_i$, $i \in [k]$,  there
are two vertex sets in $D_1$ and $D_2$: 
$\B^i =\{\beta^i_1, \ldots, \beta^i_n\}$ and a set of dummy vertices $$\D^{i,B} =\{d^{i,B}_{1,\inn}, d^{i,B}_{1,\out}, \ldots, d^{i,B}_{n,\inn}, d^{i,B}_{n,\out}\}.$$ (The vertex sets $\B^i$ and $\D^{i,B}$ are common in $D_1$ and $D_2$.)
\fi
\ifshort
Similarily to \textbf{V-2}, for each $B_i$ the two vertex sets $\B^i =\{\beta^i_1, \ldots, \beta^i_n\}$ and $\D^{i,B} =\{d^{i,B}_{1,\inn}, d^{i,B}_{1,\out}, \ldots, d^{i,B}_{n,\inn}, d^{i,B}_{n,\out}\}.$ are common in $D_1$ and $D_2$.\fi

\smallskip

\noindent
{\bf (V-3)} For each ordered pair $i,j \in [k] \times [k]$, there is a set of dummy vertices in $D_1$ and $D_2$: $\D^{i,j}= \{d^{i,j}_{1,\inn}, d^{i,j}_{1,\out}, \ldots, \allowbreak d^{i,j}_{n,\inn}, \allowbreak d^{i,j}_{n,\out} \}$. (The sets $\D^{i,j}$ are common in $D_1$ and $D_2$.)

\smallskip

\noindent
{\bf (V-4)} For each ordered pair $i,j \in [k] \times [k]$, there is a set of vertices $\PP^{i,j}_A=\{p^{i,j}_{1,1}, \ldots, p^{i,j}_{1,n}, \allowbreak p^{i,j}_{2,1}, \ldots p^{i,j}_{2,n}, \ldots, \allowbreak  p^{i,j}_{n,1}, \ldots, p^{i,j}_{n,n}\}$ in $D_1$ and 
a set of vertices $\PP^{i,j}_B=\{\hat{p}^{i,j}_{1,1}, \ldots, \hat{p}^{i,j}_{1,n}, \allowbreak \hat{p}^{i,j}_{2,1}, \ldots \hat{p}^{i,j}_{2,n}, \ldots, \allowbreak  \hat{p}^{i,j}_{n,1}, \ldots, \hat{p}^{i,j}_{n,n}\}$ in $D_2$. Since in the construction we desire the vertex set of $D_1$ to be the same as the vertex set of $D_2$, we map 
the vertices of $\PP^{i,j}_A$ with $\PP^{i,j}_B$ such that the mapped vertices corresponds to the same vertex (which is just called by different names in the two digraphs). The desired map, maps the vertex $p^{i,j}_{q,r}$ to the vertex $\hat{p}^{i,j}_{r,q}$.

\smallskip

\noindent
{\bf (V-5)} In addition to the above vertices, there is a source vertex $\sss$ and a sink vertex $\ttt$ in both digraphs $D_1$ and $D_2$.

\medskip

\noindent
The arc set of $D_1$ is constructed as follows. See Figure~\ref{fig:full-construction} for an illustration.

\medskip

\noindent
{\bf (A$_1$-1)} For each $i \in [k]$, there is a directed path on the vertex set $\Ac^i \cup \D^{i,A}$ defined as $\pth^A_i = (d^{i,A}_{1,\inn}, \alpha^i_1,d^{i,A}_{1,\out}, \allowbreak  d^{i,A}_{2,\inn},  \alpha^i_2,d^{i,A}_{2,\out}, \allowbreak \ldots, \allowbreak d^{i,A}_{n,\inn}, \allowbreak \alpha^i_n, \allowbreak d^{i,A}_{n,\out})$ of length $3n$ (here length is the number of vertices on the path).

\smallskip

\iflong
\noindent
{\bf (A$_1$-2)} For each ordered pair $i,j \in [k] \times [k]$, there is a directed path $\pth_{i,j}$ (of length $n^2 + 2n$) on the vertex set $\PP^{i,j}_A \cup \D^{i,j}$ defined as $$(d^{i,j}_{1,\inn}, p^{i,j}_{1,1}, \allowbreak \ldots, \allowbreak p^{i,j}_{1,n}, d^{i,j}_{1,\out}\allowbreak d^{i,j}_{2,\inn}, p^{i,j}_{2,1}, \ldots p^{i,j}_{2,n}, d^{i,j}_{2, \out}, \ldots, \allowbreak  d^{i,j}_{n,\inn}, p^{i,j}_{n,1}, \ldots, p^{i,j}_{n,n}, d^{i,j}_{n,\out}).$$ 
\fi

\ifshort
\noindent
{\bf (A$_1$-2)} For each ordered pair $i,j \in [k] \times [k]$, there is a directed path $\pth_{i,j}$ (of length $n^2 + 2n$) on the vertex set $\PP^{i,j}_A \cup \D^{i,j}$ defined as $(d^{i,j}_{1,\inn}, p^{i,j}_{1,1}, \allowbreak \ldots, \allowbreak p^{i,j}_{1,n}, d^{i,j}_{1,\out}\allowbreak d^{i,j}_{2,\inn}, p^{i,j}_{2,1}, \ldots p^{i,j}_{2,n}, d^{i,j}_{2, \out}, \ldots, \allowbreak  d^{i,j}_{n,\inn}, p^{i,j}_{n,1}, \ldots, p^{i,j}_{n,n}, d^{i,j}_{n,\out}).$
\fi

\smallskip

\noindent
{\bf (A$_1$-3)} The vertex $\sss$ is a source and has as out-neighbours the first vertices of each of the paths $\pth^A_i$ and $\pth_{i,j}$. The vertex $\ttt$ is a sink and has as in-neighbours the last vertices of each of the paths $\pth^A_i$ and $\pth_{i,j}$.

\smallskip

\noindent
{\bf (A$_1$-4)} For each $i,j \in [k]\times [k]$ and $q \in [n]$, $(d^{i,j}_{q, \inn}, d^{i,A}_{q,\inn})$ and $(d^{i,j}_{q,\out}, d^{i,A}_{q, \out})$ are arcs in $D_1$. These are the {\em block} arcs.

\smallskip

\noindent
{\bf (A$_1$-5)} For each $i,j \in [k]\times [k]$ and $q,r \in [n]$, $(p^{i,j}_{q,r}, \alpha^i_q) \in A(D_1)$ if and only if $(a^i_q,b^j_r) \in E(G)$. These edges are called {\em adjacency-encoders} in $D_1$.

\medskip

\noindent
Observe that for each $i \in [k]$, the vertices of $\B_i \cup \D^{i,B}$ are isolated vertices in $D_1$.
\ifshort The arc set of $D_2$ is defined analogously with occurences of $\alpha^i_j$, $d^{i,A}_{j,\in}$, $d^{i,A}_{j,\out}$, and $p^{i,j}_{x,y}$
replaced by $\beta^i_j$, $d^{i,B}_{j,\in}$, $d^{i,B}_{j,\out}$, and $\hat{p}^{i,j}_{x,y}$, respectively, resulting in the paths
$\pth^B_i$ and $\hat{\pth}_{i,j}$.\fi
\iflong We continue by constructing the arc set of $D_2$.

\medskip

\noindent
{\bf (A$_2$-1)} For each $i \in [k]$, there is a directed path on the vertex set $\B^i \cup \D^{i,B}$ defined as $\pth^B_i = (d^{i,B}_{1,\inn}, \beta^i_1, d^{i,B}_{1,\out},  \allowbreak d^{i,B}_{2,\inn}, \beta^i_2,d^{i,B}_{2,\out}, \allowbreak \ldots, \allowbreak d^{i,B}_{n,\inn}, \allowbreak \beta^i_n, \allowbreak d^{i,B}_{n,\out})$ of length $3n$.

\medskip

\noindent
{\bf (A$_2$-2)} For each ordered pair $i,j \in [k] \times [k]$, there is a directed path $\hat{\pth}_{i,j}$ (of length $n^2 + 2n$)
on the vertex set $\PP^{i,j}_B \cup \D^{i,B}$ defined as $$(d^{i,j}_{1,\inn}, \hat{p}^{i,j}_{1,1}, \ldots, \hat{p}^{i,j}_{1,n}, d^{i,j}_{1,\out}\allowbreak d^{i,j}_{2,\inn}, \hat{p}^{i,j}_{2,1}, \ldots \hat{p}^{i,j}_{2,n}, d^{i,j}_{2, \out}, \ldots, \allowbreak  d^{i,j}_{n,\inn}, \hat{p}^{i,j}_{n,1}, \ldots, \hat{p}^{i,j}_{n,n}, d^{i,j}_{n,\out}).$$ 

\medskip

\noindent
{\bf (A$_2$-3)} The vertex $\sss$ is a source and has as out-neighbours the first vertices of each of the paths $\pth^B_i$ and $\hat{\pth}_{i,j}$. The vertex $\ttt$ is a sink and has as in-neighbours the last vertices of each of the paths $\pth^B_i$ and $\hat{\pth}_{i,j}$.

\medskip

\noindent
{\bf (A$_2$-4)} For each $i,j \in [k]\times [k]$ and $q \in [n]$, $(d^{i,j}_{q, \inn}, d^{i,B}_{q,\inn})$ and $(d^{i,j}_{q,\out}, d^{i,B}_{q, \out})$ are arcs in $D_2$. 

\medskip

\noindent
{\bf (A$_2$-5)}  For each $i,j \in [k]\times [k]$ and $q,r \in [n]$, $(\hat{p}^{i,j}_{q,r}, \beta^j_q) \in A(D_1)$ if and only if $(b^j_q,a^i_r) \in E(G)$.
\fi
Also observe that, for each $i \in [k]$, the vertices of $\Ac_i \cup \D^{i,A}$ are isolated in $D_2$.
\iflong We refer to the arcs introduced in {\bf (A$_2$-4)} and {\bf (A$_2$-5)}
as block arcs and adjacency-encoders, respectively, just as we did in $D_1$.
\fi
All the vertices that are not incident to the adjacency-encoder edges should be undeletable. This can be obtained by adding
$k'+1$ copies of these vertices which have the same neighbourhood as the original vertex (that is they form a set of pairwise false twins). It is straightforward  to verify (with the aid of Figure~\ref{fig:full-construction}) that the digraphs $D_1$ and $D_2$ are indeed
acyclic.
 Finally, we set $k'= 2k+k^2$ and note the that resulting instance of \simsep{} can be computed in polynomial time.
 
We continue by proving the correctness of the reduction.
Observe from the adjacency of the vertices $\sss$ and $\ttt$ that any $\sss\ttt$-separator in $D_1$ picks at least one vertex from each path $\pth^A_i$, $i \in [k]$, and from each path $\pth_{i,j}$, $i,j \in [k]$. 
Similarly, any $\sss\ttt$-separator in $D_2$ picks at least one vertex from each path $\pth^B_i$ and from each path $\hat{\pth}_{i,j}$. Further, since the vertices of $\pth^A_i$ and $\pth^B_i$ are distinct, and $k '=2k+k^2$, the solution for \simsep{} picks exactly one vertex from each $\pth^A_i$, one vertex from each $\pth^B_i$, one vertex from $\pth_{i,j}$ and one vertex from $\hat{\pth}_{i,j}$. Further, for each $i,j \in [k] \times [k]$, if the solution picks $p^{i,j}_{q,r}$ from $\pth_{i,j}$ then it is picking $\hat{p}^{i,j}_{r,q}$ from $\hat{\pth}_{i,j}$.
\ifshort
It is now relatively straightforward to verify \fi\iflong We show \fi that $G$ has a multicolored biclique
if and only if there is a set $X \subseteq V(D_1)$ ($= V(D_2)$) such that $|X| \leq k'$ and
$D_1 - X, D_2 - X$ contain no st-paths.
\iflong

\smallskip

\noindent
{\bf Forward direction.} If $G$ has a multicolored biclique consisting of vertices $\{a^i_{f(i)} : i \in [k], f(i) \in [n]\}$ and $\{b^i_{g(i)} : i \in [k], g(i) \in [n]\}$, then construct $X \subseteq V(D_1) (=V(D_2))$ as follows: 

\begin{enumerate}
\item
$X$ contains the vertex set $\{\alpha^i_{f(i)} : i \in [k]\} \cup \{\beta^j_{g(j)} : j \in [k]\}$, and

\item
for each $i , j \in [k] \times [k]$, $X$ contains $p^{i,j}_{f(i),q(j)}$. 
\end{enumerate}

Recall that $p^{i,j}_{f(i),g(j)} = \hat{p}^{i,j}_{g(j),f(i)}$, 
none of the vertices of $X$ are undeletable, and $|X| = 2k + k^2 =k'$.
We will now show that $X$ is an $\sss\ttt$-separator in $D_1$ and in $D_2$.

For the sake of contradiction, we assume that $D_1-X$ contains an $\sss\ttt$-path.
We know that $X$ picks a vertex from each of the paths $\pth^A_i, \pth^B_i, \pth_{i,j}$ and $\hat{\pth}_{i,j}$. We also know that there are no direct connections (edges) between the paths $\pth^A_i$ and $\pth^A_j$, and between the paths $\pth_{i,j}$ and $\pth_{k,\ell}$. Hence, if there is an $st$-path in $D_1-X$, then such a path intersects $\pth^A_i$ and $\pth_{i,j}$ for some $i,j \in [k]$. In fact, such a path intersects 
only $\pth^A_i$ and $\pth_{i,j}$ for some $i,j \in [k]$. 
Since there is no arc from a vertex before $\alpha^i_{f(i)}$ in $\pth^A_i$ to a vertex after $p^{i,j}_{f(i),g(j)}$, and $\alpha^i_{f(i)}, p^{i,j}_{f(i),g(i)} \in X$, we conclude that there is no $\sss\ttt$-path in $D_1-X$.

Using symmetric arguments, it is immediate that there is no $\sss\ttt$-path in $D_2-X$ either.

\smallskip

\noindent
{\bf Backward direction.} Let $X \subseteq V(D_1) (=V(D_2))$ such that $|X| \leq k'$ and, $D_1 -X, D_2-X$ have no $\sss\ttt$-paths. Recall that since $|X| \leq k'$, $X$ intersects each of the paths $\pth^A_i,\pth^B_i, \pth_{i,j}, \hat{\pth}_{i,j}$ in exactly one vertex. In fact, the bijection between the vertices of the paths $\pth_{i,j}$ and $\hat{\pth_{i,j}}$ implies that if $p^{i,j}_{q,r} \in X$, then $\hat{p}^{i,j}_{r,q} \in X$, too.

We first claim that if $p^{i,j}_{q,r} \in X$, then $\alpha^i_q \in X$. Assume to the contrary that this is not true. Then consider the following $\sss\ttt$-path in $D_1-X$. The path starts from $\sss$, traverses the path $\pth_{i,j}$ until $d^{i,j}_{q,\inn}$, uses the block edge $(d^{i,j}_{q,\inn},d^{i,A}_{q,\inn})$, traverses the $(d^{i,A}_{q,\inn}, d^{i,A}_{q,\out})$-subpath of $\pth^{A}_i$ which exists because $d^{i,A}_{q,\inn},d^{i,A}_{q,\out}$ are undeletable and $\alpha^i_q \not \in X$ by assumption, traverses the block edge $(d^{i,A}_{q,\out},d^{i,j}_{q,\out})$, and then traverses the $(d^{i,j}_{q,\out},d^{i,j}_{n,\out})$-subpath of $\pth_{i,j}$, and finally reaches $\ttt$ from $d^{i,j}_{n,\out}$.

Using symmetric arguments, one can show that if $\hat{p}^{i,j}_{r,q} \in X$, then $\beta^j_r \in X$. 
Using these claims, say $p^{i,j}_{q,r} \in X$ then $\alpha^i_q \in X$. 
Since $p^{i,j}_{q,r} = \hat{p}^{i,j}_{r,q}$, then $\beta^j_r \in X$. Also $(a^i_q,b^j_r) \in E(G)$ since $p^{i,j}_{q,r}, \alpha^i_q, \beta^j_r$ are not undeletable and hence they are incident to adjacency-encoder arcs. Thus, the set $\{a^i_q : \alpha^i_q \in X \} \cup \{b^j_r : \beta^j_r \in X\}$ induces a multicolored biclique $K_{k,k}$ in $G$. 

\medskip

\noindent
Since \mbicl{} is W[1]-hard parameterized by the solution
size (Proposition~\ref{prop:biclique-hard}), the above reduction shows that the
\simsep{} is \Whard\ parameterized by the solution size even when the
input digraphs are acyclic.
\fi
\end{proof}

\newcommand{\xmin}{1}
\newcommand{\xmax}{15}
\newcommand{\ymin}{1}
\newcommand{\ymax}{\ymin +15}
\begin{figure}
\centering
\begin{minipage}[c]{0.2\textwidth}
\begin{tikzpicture}[%
scale=.8,
transform shape,
myvertex/.style={circ,scale=2}
]

\node[myvertex,label=left:{$a^i_1$}] (69) at (6,6) {};
\node[myvertex,label=left:{$a^i_2$}] (610) at (6,5) {};
\node[myvertex,label=left:{$a^i_3$}] (611) at (6,4) {};

\node[myvertex,label=right:{$b^j_1$}] (89) at (8,6) {};
\node[myvertex,label=right:{$b^j_2$}] (810) at (8,5) {};
\node[myvertex,label=right:{$b^j_3$}] (811) at (8,4) {};

\draw[thick] (69) -- (810);
\draw[thick] (610) -- (810);
\draw[thick] (610) -- (811);
\draw[thick] (611) -- (89);
\end{tikzpicture}
\end{minipage}%
\begin{minipage}{0.8\textwidth}
\begin{tikzpicture}[%
scale=.8,
transform shape,
myvertex/.style={circ,scale=2}
]

\node[myvertex,label=above:{$\alpha^i_{1}$}] (32) at (3,2) {};
\node[myvertex,label=above:{$\alpha^i_{2}$}] (82) at (8,2) {};
\node[myvertex,label=above:{$\alpha^i_{3}$}] (132) at (13,2) {};

\node[myvertex,fill=brown,label=above:{$d^{i,A}_{1,\inn}$}] (12) at (1,2) {};
\node[myvertex,fill=brown,label=above:{$d^{i,A}_{1,\out}$}] (52) at (5,2) {};
\node[myvertex,fill=brown,label=above:{$d^{i,A}_{2,\inn}$}] (62) at (6,2) {};
\node[myvertex,fill=brown,label=above:{$d^{i,A}_{2,\out}$}] (102) at (10,2) {};
\node[myvertex,fill=brown,label=above:{$d^{i,A}_{3,\inn}$}] (112) at (11,2) {};
\node[myvertex,fill=brown,label=above:{$d^{i,A}_{3,\out}$}] (152) at (15,2) {};

\draw[->,>=stealth,thick] (12) -- (32);
\draw[->,>=stealth,thick] (32) -- (52);
\draw[->,>=stealth,thick] (52) -- (62);
\draw[->,>=stealth,thick] (62) -- (82);
\draw[->,>=stealth,thick] (82) -- (102);
\draw[->,>=stealth,thick] (102) -- (112);
\draw[->,>=stealth,thick] (112) -- (132);
\draw[->,>=stealth,thick] (132) -- (152);

\node[myvertex,fill=brown,label=below:{$d^{i,j}_{1,\inn}$}] (11) at (1,1) {}; 
\node[myvertex,fill=brown,label=below:{$d^{i,j}_{1,\out}$}] (51) at (5,1) {}; 
\node[myvertex,fill=brown,label=below:{$d^{i,j}_{2,\inn}$}] (61) at (6,1) {};
\node[myvertex,fill=brown,label=below:{$d^{i,j}_{2,\out}$}] (101) at (10,1) {};  
\node[myvertex,fill=brown,label=below:{$d^{i,j}_{3,\inn}$}] (111) at (11,1) {}; 
\node[myvertex,fill=brown,label=below:{$d^{i,j}_{3,\out}$}] (151) at (15,1) {};

\node[myvertex,label=below:{$p^{i,j}_{1,1}$}] (21) at (2,1) {};
\node[myvertex,label=below:{$p^{i,j}_{1,2}$}] (31) at (3,1) {};
\node[myvertex,label=below:{$p^{i,j}_{1,3}$}] (41) at (4,1) {};
\node[myvertex,label=below:{$p^{i,j}_{2,1}$}] (71) at (7,1) {};
\node[myvertex,label=below:{$p^{i,j}_{2,2}$}] (81) at (8,1) {};
\node[myvertex,label=below:{$p^{i,j}_{2,3}$}] (91) at (9,1) {};
\node[myvertex,label=below:{$p^{i,j}_{3,1}$}] (121) at (12,1) {};
\node[myvertex,label=below:{$p^{i,j}_{3,2}$}] (131) at (13,1) {};
\node[myvertex,label=below:{$p^{i,j}_{3,3}$}] (141) at (14,1) {};

\foreach \i in {1,...,14}
{\pgfmathtruncatemacro{\j}{1}
  \pgfmathtruncatemacro{\z}{\i+1}
    \draw[->,>=stealth,thick] (\i\j) -- (\z\j);}

\draw[->,>=stealth,thick] (31) -- (32);
\draw[->,>=stealth,thick] (81) -- (82);
\draw[->,>=stealth,thick] (91) -- (82);
\draw[->,>=stealth,thick] (121) -- (132);

\draw[->,>=stealth,thick,color=brown] (11) -- (12);
\draw[->,>=stealth,thick,color=brown] (52) -- (51);
\draw[->,>=stealth,thick,color=brown] (61) -- (62);
\draw[->,>=stealth,thick,color=brown] (102) -- (101);
\draw[->,>=stealth,thick,color=brown] (111) -- (112);
\draw[->,>=stealth,thick,color=brown] (152) -- (151);

\node[myvertex,label=below:{$\sss$}] (55) at (0,1.5) {};
\node[myvertex,label=below:{$\ttt$}] (66) at (16,1.5) {};

\draw[->,>=stealth,thick] (55) -- (11);
\draw[->,>=stealth,thick] (55) -- (12);
\draw[->,>=stealth,thick] (151) -- (66);
\draw[->,>=stealth,thick] (152) -- (66);


\node[myvertex,label=above:{$\beta^j_{1}$}] (35) at (3,-2) {};
\node[myvertex,label=above:{$\beta^j_{2}$}] (85) at (8,-2) {};
\node[myvertex,label=above:{$\beta^j_{3}$}] (135) at (13,-2) {};

\node[myvertex,fill=brown,label=above:{$d^{j,B}_{1,\inn}$}] (15) at (1,-2) {};
\node[myvertex,fill=brown,label=above:{$d^{j,B}_{1,\out}$}] (55) at (5,-2) {};
\node[myvertex,fill=brown,label=above:{$d^{j,B}_{2,\inn}$}] (65) at (6,-2) {};
\node[myvertex,fill=brown,label=above:{$d^{j,B}_{2,\out}$}] (105) at (10,-2) {};
\node[myvertex,fill=brown,label=above:{$d^{j,B}_{3,\inn}$}] (115) at (11,-2) {};
\node[myvertex,fill=brown,label=above:{$d^{j,B}_{3,\out}$}] (155) at (15,-2) {};

\draw[->,>=stealth,thick] (15) -- (35);
\draw[->,>=stealth,thick] (35) -- (55);
\draw[->,>=stealth,thick] (55) -- (65);
\draw[->,>=stealth,thick] (65) -- (85);
\draw[->,>=stealth,thick] (85) -- (105);
\draw[->,>=stealth,thick] (105) -- (115);
\draw[->,>=stealth,thick] (115) -- (135);
\draw[->,>=stealth,thick] (135) -- (155);

\node[myvertex,fill=brown,label=below:{$d^{i,j}_{1,\inn}$}] (14) at (1,-3) {}; 
\node[myvertex,fill=brown,label=below:{$d^{i,j}_{1,\out}$}] (54) at (5,-3) {}; 
\node[myvertex,fill=brown,label=below:{$d^{i,j}_{2,\inn}$}] (64) at (6,-3) {};
\node[myvertex,fill=brown,label=below:{$d^{i,j}_{2,\out}$}] (104) at (10,-3) {};  
\node[myvertex,fill=brown,label=below:{$d^{i,j}_{3,\inn}$}] (114) at (11,-3) {}; 
\node[myvertex,fill=brown,label=below:{$d^{i,j}_{3,\out}$}] (154) at (15,-3) {};

\node[myvertex,label=below:{$\hat{p}^{i,j}_{1,1}$}] (24) at (2,-3) {};
\node[myvertex,label=below:{$\hat{p}^{i,j}_{1,2}$}] (34) at (3,-3) {};
\node[myvertex,label=below:{$\hat{p}^{i,j}_{1,3}$}] (44) at (4,-3) {};
\node[myvertex,label=below:{$\hat{p}^{i,j}_{2,1}$}] (74) at (7,-3) {};
\node[myvertex,label=below:{$\hat{p}^{i,j}_{2,2}$}] (84) at (8,-3) {};
\node[myvertex,label=below:{$\hat{p}^{i,j}_{2,3}$}] (94) at (9,-3) {};
\node[myvertex,label=below:{$\hat{p}^{i,j}_{3,1}$}] (124) at (12,-3) {};
\node[myvertex,label=below:{$\hat{p}^{i,j}_{3,2}$}] (134) at (13,-3) {};
\node[myvertex,label=below:{$\hat{p}^{i,j}_{3,3}$}] (144) at (14,-3) {};

\foreach \i in {1,...,14}
{\pgfmathtruncatemacro{\j}{4}
  \pgfmathtruncatemacro{\z}{\i+1}
    \draw[->,>=stealth,thick] (\i\j) -- (\z\j);}


\draw[->,>=stealth,thick,color=brown] (14) -- (15);
\draw[->,>=stealth,thick,color=brown] (55) -- (54);
\draw[->,>=stealth,thick,color=brown] (64) -- (65);
\draw[->,>=stealth,thick,color=brown] (105) -- (104);
\draw[->,>=stealth,thick,color=brown] (114) -- (115);
\draw[->,>=stealth,thick,color=brown] (155) -- (154);

\draw[->,>=stealth,thick] (44) -- (35);
\draw[->,>=stealth,thick] (74) -- (85);
\draw[->,>=stealth,thick] (84) -- (85);
\draw[->,>=stealth,thick] (134) -- (135);

\node[myvertex,label=below:{$\sss$}] (88) at (0,-2.5) {};
\node[myvertex,label=below:{$\ttt$}] (99) at (16,-2.5) {};

\draw[->,>=stealth,thick] (88) -- (14);
\draw[->,>=stealth,thick] (88) -- (15);
\draw[->,>=stealth,thick] (154) -- (99);
\draw[->,>=stealth,thick] (155) -- (99);


\draw[draw=red,line width=0.3mm,fill=none] (8,2) circle[radius=3mm];

\draw[draw=red,line width=0.3mm,fill=none] (9,1) circle[radius=3mm];

\draw[draw=red,line width=0.3mm,fill=none] (13,-3) circle[radius=3mm];

\draw[draw=red,line width=0.3mm,fill=none] (13,-2) circle[radius=3mm];
\end{tikzpicture}
\end{minipage}
\caption{An illustration of the construction for $n=3$. 
The left figure is a snippet of the adjacency between the set $A_i$ and the set $B_j$ in $G$ corresponding to which the next sets of paths have been constructed in $D_1$ and $D_2$ respectively.
The first set of two paths are $\pth^A_i$ and $\pth_{i,j}$ respectively in $D_1$. The next set of two paths are $\pth^B_j$ and $\hat{\pth}_{i,j}$ in $D_2$. The brown vertices are the dummy vertices. The vertices marked with red circles are in the solution.}
\label{fig:full-construction}
\end{figure}
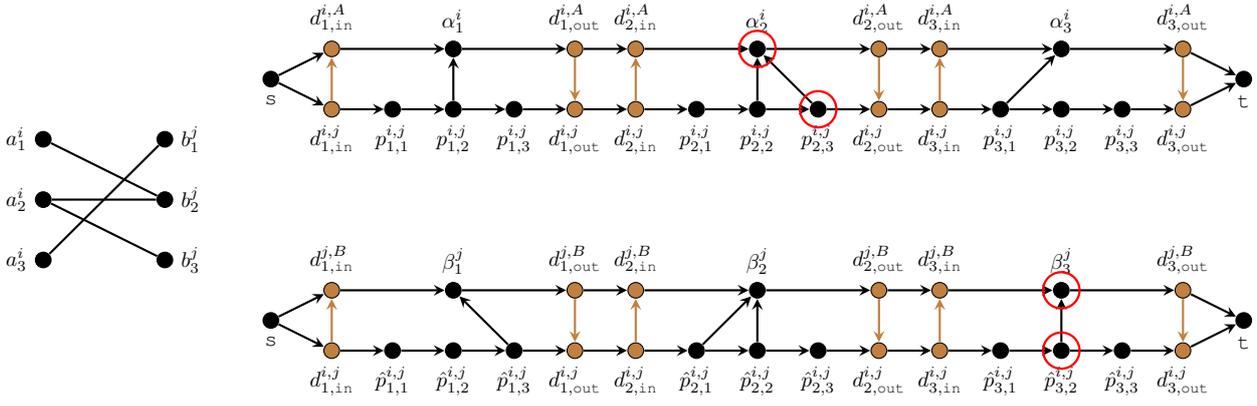

We continue by using the \W{1}-hardness of  \simsep{} for proving \W{1}-hardness of the following
two problems.

\pbDefP{Simultaneous Directed st-Cut (\simcut{})}
{Two directed graphs $D_1$ and $D_2$ with $V=V(D_1) = V(D_2)$, vertices $s,t \in V$, and an integer $k$.}
{$k$.}
{Is there a subset $X \subseteq E(D_1) \cup E(D_2)$ of size at most $k$ such that
 neither $D_1-X$ nor $D_2-X$ contain a path from $s$ to $t$?}

\pbDefP{Simultaneous Directed Feedback Arc Set (\simdfas)}
{Directed graphs $D_1$, $D_2$ with $V=V(D_1) = V(D_2)$, and an integer $k$.}
{$k$.}
{Is there a subset $X \subseteq E(D_1) \cup E(D_2)$ of size at most $k$ such that
both $D_1 - X$ and $D_2 - X$ are acyclic?}

\begin{theorem} \label{thm:sim-cut-hard}
  \simcut{} is \textnormal{\W{1}}-hard even if both input digraphs are acyclic.
\end{theorem}
\iflong
\begin{proof}
  We use the \W{1}-hardness of \simsep{} (Theorem~\ref{thm:sim-sep-hard})
  and the construction in the proof
as our basis.
  Hence, assume that we have started with an
  instance $(G,k)$ of \mbicl{} and used the construction
for computing two acyclic directed graphs $D_1$ and $D_2$.
We perform the following operations
on $D_1$ and $D_2$.

  \begin{enumerate}
  \item
  Split
  each vertex $v$ in the graph into two vertices $v^-$ and $v^+$. 
  
  \item
  Add an arc
  $(v^-,v^+)$, which we refer to as the {\em splitting arc}. 

  \item
  Let all
  in-neighbours of $v$ be in-neighbours of $v^-$ and all 
  out-neighbours of $v$ be out-neighbours of $v^+$. 

  \item
  Make all arcs in the graph, except the splitting arcs,
  undeletable. This can be achieved by making $k'+1$ copies of each of
  these arcs and subdividing them once.
\end{enumerate}

  Let $D'_1,D'_2$ denote the resulting graphs
  and note that they are acyclic.
  It is now easy to see that $(D'_1,D'_2,k)$ is a yes-instance of \simcut{}
  if and only if $(G,k)$ is a yes-instance
  of \mbicl{}.
  We conclude that \simcut{} is \W{1}-hard.
\end{proof}
\fi

\begin{theorem}\label{thm:sim-dfas-hard}
  \simdfas{} is \textnormal{\W{1}}-hard.
\end{theorem}
\iflong
\begin{proof}
Recall the proof of Theorem~\ref{thm:sim-cut-hard}.
The graphs $D_1,D_2$ that are constructed in the proof are both
acyclic.
  To each of these graphs,
  add the arc $(\ttt, \sss)$ and make it undeletable.
  Then all directed cycles in $D_1$ and $D_2$ must use this undeletable
  arc $(\ttt,\sss)$. Thus, a set of arcs hits all cycles in $D_1$ and $D_2$ if and
  only if it hits all $\sss\ttt$-paths in $D_1$ and $D_2$.
\end{proof}
\fi

\iflong
\subsection{Intractable Fragments Containing $\m$} \label{ssec:me-ms-hard}
\fi

\ifshort
\subsection{Intractable Fragments} \label{ssec:me-ms-hard}
\fi

\iflong
We will now show that the problems $\mincsp{\m, \e}$ and $\mincsp{\m,
  \s}$ are \W{1}-hard. 
\fi
\ifshort
We begin by proving that $\mincsp{\m, \e}$ and $\mincsp{\m,
  \s}$ are \W{1}-hard. 
\fi
We introduce two
binary relations:
let $\e^-$ denote the \emph{left-equals} relation and $\e^+$ denote the
\emph{right-equals} relation, which hold for any pair of intervals
with matching left endpoints and right endpoints, respectively.
Both relations can be implemented using only $\m$ as follows:
$\{ z \m x, z \m y \}$ implements $x \e^- y$
where $z$ is a fresh variable; similarly, 
 $\{ x \m z, y \m z \}$ implements $x \e^+ y$.
Thus, we may assume that relations $\e^-$ and $\e^+$ are available whenever we have access to the $\m$ relation.
We are now ready to present the reduction for $\mincsp{\m,
  \e}$, which will be from the \pc{} problem that was shown to be \W{1}-hard in Lemma\iflong~\ref{lem:pc-hard}\fi\ifshort~\ref{lem:pcpcfas-hard}\fi.
\iflong The reduction uses a gadget that is also used in the proof of Theorem~6.2~in~\cite{Dabrowski:etal:soda2023} 
(see also Figure\iflong~\ref{fig:gadget}\fi\ifshort~\ref{fig:gadget-short}\fi).\fi

\iflong
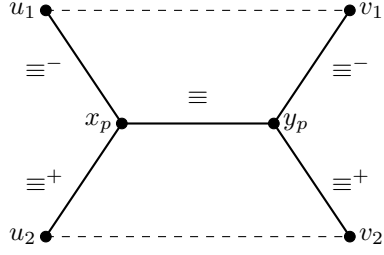
\begin{figure}
  \centering
  \begin{tikzpicture}
    \coordinate (u1) at (0,3);
    \coordinate (v1) at (4,3);
    \coordinate (u2) at (0,0);
    \coordinate (v2) at (4,0);
    \coordinate (xp) at (1,3/2);
    \coordinate (yp) at (3,3/2);
    
    \filldraw[black] (u1) circle (2pt) node[anchor=east]{$u_1$};
    \filldraw[black] (v1) circle (2pt) node[anchor=west]{$v_1$};
    \filldraw[black] (u2) circle (2pt) node[anchor=east]{$u_2$};
    \filldraw[black] (v2) circle (2pt) node[anchor=west]{$v_2$};
    \filldraw[black] (xp) circle (2pt) node[anchor=east]{$x_p$};
    \filldraw[black] (yp) circle (2pt) node[anchor=west]{$y_p$};
    
    \draw[dashed] (u1) -- (v1);
    \draw[dashed] (u2) -- (v2);
    \draw[thick] (u1) -- node[midway, label=left:$\e^{-}$] () {} (xp);
    \draw[thick] (u2) -- node[midway, label=left:$\e^+$] () {} (xp);
    \draw[thick] (xp) -- node[midway, label=above:$\e$] () {} (yp);
    \draw[thick] (yp) -- node[midway, label=right:$\e^-$] () {} (v1);
    \draw[thick] (yp) -- node[midway, label=right:$\e^+$] () {} (v2);
\end{tikzpicture}
  \caption{Gadget from Theorem~\ref{thm:me-hard}.}
  \label{fig:gadget}
\end{figure} 
\fi

\ifshort
\begin{figure}
  \centering
  \begin{minipage}{0.5\textwidth}
  \begin{tikzpicture}
    \coordinate (u1) at (0,3);
    \coordinate (v1) at (4,3);
    \coordinate (u2) at (0,0);
    \coordinate (v2) at (4,0);
    \coordinate (xp) at (1,3/2);
    \coordinate (yp) at (3,3/2);
    
    \filldraw[black] (u1) circle (2pt) node[anchor=east]{$u_1$};
    \filldraw[black] (v1) circle (2pt) node[anchor=west]{$v_1$};
    \filldraw[black] (u2) circle (2pt) node[anchor=east]{$u_2$};
    \filldraw[black] (v2) circle (2pt) node[anchor=west]{$v_2$};
    \filldraw[black] (xp) circle (2pt) node[anchor=east]{$x_p$};
    \filldraw[black] (yp) circle (2pt) node[anchor=west]{$y_p$};
    
    \draw[dashed] (u1) -- (v1);
    \draw[dashed] (u2) -- (v2);
    \draw[thick] (u1) -- node[midway, label=left:$\e^{-}$] () {} (xp);
    \draw[thick] (u2) -- node[midway, label=left:$\e^+$] () {} (xp);
    \draw[thick] (xp) -- node[midway, label=above:$\e$] () {} (yp);
    \draw[thick] (yp) -- node[midway, label=right:$\e^-$] () {} (v1);
    \draw[thick] (yp) -- node[midway, label=right:$\e^+$] () {} (v2);
\end{tikzpicture}
  \end{minipage}%
  \begin{minipage}{0.5\textwidth}
\begin{tikzpicture}[%
      scale=.8,
      transform shape,
      myvertex/.style={circ,scale=2}
      ]

      \draw
      node[myvertex,label=below:{$u$}] (u) {}
      node[myvertex, right of=u, label=below:{$h_{u,v}$}] (h1) {}
      node[myvertex, right of=h1, label=below:{$v$}] (v) {}
      ;
      
      \draw[thick,-latex] (u) -- node[midway, label=above:$\p$] () {} (h1);
      \draw[thick,-latex] (h1) -- node[midway, label=above:$\d$] () {}
      (v);
      \draw
      (u.west) +(-1cm, 0cm) node[anchor=east] (c1) {$(u,v) \in A_1\setminus A_2$}
      ;
      
      \draw
      (u) +(0cm,-1.5cm)  node[myvertex,label=below:{$u$}] (u1) {}
      node[myvertex, right of=u1, label=below:{$h_{u,v}$}] (h11) {}
      node[myvertex, right of=h11, label=below:{$v$}] (v1) {}
      ;
      
      \draw[thick,-latex] (h11) -- node[midway, label=above:$\p$] () {} (u1);
      \draw[thick,-latex] (h11) -- node[midway, label=above:$\d$] () {}
      (v1);
      \draw
      (u1.west) +(-1cm, 0cm) node[anchor=east] (c2) {$(u,v) \in A_2\setminus A_1$}
      ;
      
      \draw
      (u1) +(0cm,-1.5cm)  node[myvertex,label=below:{$u$}] (u2) {}
      node[myvertex, right of=u2, label=below:{$v$}] (v2) {}
      ;
      
      \draw[thick,-latex] (u2) -- node[midway, label=above:$\d$] () {} (v2);
      \draw
      (u2.west) +(-1cm, 0cm) node[anchor=east] (c2) {$(u,v) \in A_2\cap A_1$}
      ;

    \end{tikzpicture}
    \end{minipage}
    \caption{\textbf{Left:} Gadget from Theorem~\ref{thm:me-hard}. \textbf{Right:} The variables and constraints introduced in the
      construction of the proof of Theorem~\ref{thm:pdWh} for every
      arc $(u,v)$ of the \simdfas{} instance.}
      \label{fig:gadget-short}
\end{figure}
\fi

\begin{theorem} \label{thm:me-hard}
  \mincsp{\m, \e} is \textnormal{\W{1}}-hard.
\end{theorem}
\begin{proof}
  Let $(G_1, G_2, s_1, t_1, s_2, t_2, \bundles, k)$
  be an arbitrary instance of \textsc{Paired Cut}.
  We construct an instance $(\III, k)$ of $\mincsp{\m, \e}$ with 
  the same parameter.
  Start by introducing a variable for every vertex of $V(G_1) \cup V(G_2)$ and adding a crisp constraint \ifshort $u \e v$ \fi
  \iflong \begin{equation} \label{eq:crisp-edges}
    u \e v
  \end{equation}\fi
  for every edge $uv \in E(G_1)\cup E(G_2)$ with $uv \notin \bigcup \bundles$, i.e.
  edges that are not included in any bundle pair and are thus undeletable.
  We also add the following crisp constraints: \ifshort $s_1 \m t_1$, $t_1 \m s_2$, $s_2 \m t_2$.\fi
  \iflong \begin{equation} \label{eq:sts}
    s_1 \m t_1, t_1 \m s_2, s_2 \m t_2.
  \end{equation}\fi
  For every pair $p \in \bundles$ with edges
  $(u_i, v_i)$ for $i \in \{1,2\}$, 
  introduce two new variables $x_p$ and $y_p$, 
  and add the following constraints as illustrated in Figure\iflong~\ref{fig:gadget}\fi\ifshort~\ref{fig:gadget-short}\fi:
  \ifshort
  crisp constraints $u_1 \e^- x_p$ and $u_2 \e^+ x_p$, crisp constraints $v_1 \e^- y_p$ and $v_2 \e^+ y_p$, soft constraint $x_p \equiv y_p$.
  \fi
  \iflong
  \begin{itemize}
  \item crisp constraints $u_1 \e^- x_p$ and $u_2 \e^+ x_p$,
  \item crisp constraints $v_1 \e^- y_p$ and $v_2 \e^+ y_p$,
  \item soft constraint $x_p \equiv y_p$.
  \end{itemize}
  \fi
  This completes the reduction.
  Clearly, it can be implemented in polynomial time.
  The equivalence of the instances $(G_1, G_2, s_1, t_1, s_2, t_2,
  \bundles, k)$ and $(\III, k)$ now follows relatively
  straightforwardly after the observation that removing a soft
  constraint $x_p\e y_p$ corresponds to removing any (or both) edge(s) in bundle $p$.
\iflong
  We claim that $(G_1, G_2, s_1, t_1, s_2, t_2, \bundles, k)$ is a yes-instance
  of \textsc{Paired Cut} if and only if $(\III, k)$ is a yes-instance
  of \mincsp{\m, \e}. 

\smallskip

\noindent
  {\bf Forward direction.} Suppose
  $X \subseteq \bundles$ is a subset of $k$ pairs
  such that $X_i = \{ e_i : (e_1, e_2) \in X \}$
  is an $\{s_i,t_i\}$-cut in $G_i$ 
  for $i \in \{1,2\}$.
  Let $X'$ be the set of constraints that contains
  $x_p \e y_p$ for all $p \in X$.
  We claim that $I - X'$ is satisfied by 
  assignment $\alpha$ constructed as follows.
  First, set 
  \begin{gather*}
    \alpha(s_1) = [0,1], \\
    \alpha(t_1) = [1,2], \\
    \alpha(s_2) = [2,3], \\
    \alpha(t_2) = [3,4].
  \end{gather*}
  For every vertex $u_i \in V(G_i)$, set 
  \begin{equation*}
    \alpha(u_i) = \begin{cases}
      \alpha(s_i) & \text{there is an } s_iu_i\text{-path in } G_i - X, \\
      \alpha(t_i) & \text{otherwise}.
    \end{cases}
  \end{equation*}
  For every pair $p \in \bundles$ consisting of two edges
  $u_i v_i$ for $i \in \{1,2\}$, assign 
  \begin{gather*}
    \alpha(x_p) = [\alpha(u_1)^-, \alpha(u_2)^+], \\
    \alpha(y_p) = [\alpha(v_1)^-, \alpha(v_2)^+].    
  \end{gather*}
  Observe that $\alpha(u_1)^-,\alpha(v_1)^- \leq 1$
  and $\alpha(u_2)^+,\alpha(v_2)^+ \geq 2$,
  so $\alpha$ assigns non-empty intervals to all variables.
  By definition, $\alpha$ satisfies all crisp constraints in $I$.
  All soft constraints are of the form $x_p \e y_p$ for some pair $p \in \bundles$.
  If such a constraint is not in $X'$, then the edges 
  $u_1 v_1$ and $u_2 v_2$ are present in $G - X$.
  Then $\alpha(u_i) = \alpha(v_i)$ for $i \in \{1,2\}$ so
  $\alpha(x_p) = \alpha(y_p)$.
  This implies that $\alpha$ can leave at most $|X| \leq k$ 
  constraints unsatisfied.

\smallskip

\noindent
{\bf Backward direction.}
  Suppose $Y$ is a set of 
  at most $k$ soft constraints such that $I - Y$ is consistent.
  Let $Y' \subseteq \bundles$ contain all pairs of edges $p$
  such that the constraint $x_p \e y_p$ is in $Y$.
  We claim that $Y'$ is an $st$-cut in $G_1$ and $G_2$.
  For the sake of contradiction, assume there is a path
  in $G_1 - Y'$ from $s_1$ to $t_1$.
  This path corresponds to a chain of constraints in $I$
  linking $s_1$ and $t_1$ by $\e$-constraints and $\e^{-}$-constraints.
  This implies that the left endpoints of $s_1$ and $t_1$ match and this
  contradicts the crisp constraint $s_1 \m t_1$. 
  The case with $G_2$, $s_2$, $t_2$ is similar.
  \fi
\end{proof}

We continue by showing that $\mincsp{\m, \s}$ is \W{1}-hard.
First note even though we no longer have access to $\e$, 
we can add constraints
$x \e^- y$ and $x \e^+ y$ which imply $x \e y$.
As previously, relations $\e^-$ and $\e^+$ can be implemented using only $\m$.
We remark that $\{x \e^- y, x \e^+ y\}$ is \emph{not} an implementation of $\e$,
so we can only use $\e$ in crisp constraints.
Our reduction is based on the
\pcfas{} problem, which was shown to be \W{1}-hard in Lemma\iflong~\ref{lem:pcfas-hard}\fi\ifshort~\ref{lem:pcpcfas-hard}\fi.
While the reduction is quite similar to the reduction for $\mincsp{\m, \e}$, it is non-trivial to replace the
role of $\e$ with $\s$. \iflong We therefore provide the reduction here in its entirety.\fi

\begin{theorem}
  \mincsp{\m, \s} is \textnormal{\W{1}}-hard.
\end{theorem}
\iflong
\begin{proof}
  Let $(G_1, G_2, s, t, \bundles, k)$ be an arbitrary instance of
  \pcfas{} and let $\bundles_i=\{ e_i \mid (e_1,e_2) \in \bundles \}$
  for every $i \in \{1,2\}$. We construct an instance $(\III, k)$ of
  $\mincsp{\m, \s}$ with the same parameter.
  
  We start by introducing a variable for every vertex of $V(G_1) \cup V(G_2)$
  and add the following crisp constraints:
  \begin{itemize}
  \item $u \e v$ for every $uv \in E(G_1)\setminus \bundles_1$,
  \item $u \s v$ for every $(u,v) \in E(G_2)\setminus \bundles_2$,
  \item $s \m t$,
  \item $t \m v$ for every $v \in V(G_2)$.
  \end{itemize}
  Moreover, for every bundle $b=(xy, (u,v)) \in \bundles$, we introduce
  two fresh variables $z_b$ and $z_b'$ and the following constraints:
  \begin{itemize}
  \item crisp constraints $x \s z_b$ and $y \s z_b'$,
  \item crisp constraints $z_b \e^+ u$ and $z_b' \e^+ v$ and
  \item the soft constraint $ z_b \s z_b'$.
  \end{itemize}
  This completes the reduction.
  Clearly, it can be implemented in polynomial time.
  We claim that $(G_1, G_2, s, t, \bundles, k)$ is a yes-instance
  of \pcfas{} if and only if $(\III, k)$ is a yes-instance
  of \mincsp{\m, \s}. 

\smallskip

\noindent
{\bf Forward direction.} Suppose
  $X \subseteq \bundles$ is a subset of $k$ pairs
  such that $X_1$
  is an $st$-cut in $G_1$ and $X_2$ is a FAS for $G_2$, where
  $X_i=\{ e_i \mid (e_1,e_2) \in X\}$. 
  Let $X'$ be the set of constraints that contains
  $z_b \s z_b'$ for all $b \in X$.
  We claim that $I - X'$ is satisfied by 
  assignment $\alpha$ constructed as follows.
  First, set 
  \begin{gather*}
    \alpha(s) = [0,1], \\
    \alpha(t) = [1,2], 
  \end{gather*}
  For every vertex $u \in V(G_1)$, set 
  \begin{equation*}
    \alpha(u) = \begin{cases}
      \alpha(s) & \text{there is an } su\text{-path in } G_1 - X_1, \\
      \alpha(t) & \text{otherwise}.
    \end{cases}
  \end{equation*}
  We know that $G_2-X_2$ is acyclic. Hence, there is an ordering $u_1,\dotsc,u_\ell$ of the
  vertices in $G_2$ such that there is no arc in $G_2-X_2$ from $u_i$ to $u_j$ for any $j<i$. Then, we set $\alpha(u_i)=[2,2+i]$ for every $i$ with $1 \leq i \leq \ell$.
  For every pair $b=(xy,uv) \in \bundles$, assign 
  \begin{gather*}
    \alpha(z_b) = [\alpha(x)^-, \alpha(u)^+], \\
    \alpha(z_b') = [\alpha(y)^-, \alpha(v)^+].    
  \end{gather*}

  By definition, $\alpha$ satisfies all crisp constraints in $I$.
  All soft constraints are of the form $z_b \s z_b'$ for some pair $b=(xy,(u,v) \in \bundles$.
  If such a constraint is not in $X'$, then the edge 
  $xy$ is present in $G_1 - X_1$.
  Therefore, $\alpha(x) = \alpha(y)$ and $\alpha(z_b)^-=\alpha(z_b')^-$. Similarily,
  the arc $(u,v)$ is present in $G_2-X_2$, which implies that
  $\alpha(u)^+< \alpha(v)^+$ and therefore $\alpha(z_b)^+ <\alpha(z_b')^+$. Therefore,
  the assignment $\alpha$ satisfies the soft constraint $z_b \s z_b'$.

  \medskip

\noindent
  {\bf Backward direction.}
  Suppose $Y$ is a set of 
  at most $k$ soft constraints such that $\III - Y$ is consistent.
  Let $Y' \subseteq \bundles$ contain all pairs of edges $b=(xy,(u,v)\})$
  such that constraint $z_b \s z_b'$ is in $Y$. Let $Y_i'=\{ e_i \mid (e_1,e_2)\}$.
  We start by showing that $Y_1'$ is an $st$-cut in $G_1$.
  For the sake of contradiction, assume there is a path
  in $G_1 - Y_1'$ from $s$ to $t$.
  This path corresponds to a chain of constraints in $I$
  linking $s$ and $t$ by $\e$-constraints and $\s$-constraints.
  This implies that the left endpoints of $s$ and $t$ match and this
  contradicts the crisp constraint $s \m t$.

  We finish by showing that $Y_2'$ is a FAS for $G_2$.
  Suppose to the contrary that this is not the case and let $C$ be a
  directed cycle in $G_2-Y_2'$.
  The cycle $C$ corresponds to a cycle of constraints in $I$
  linked by $\s$-constraints and $\e^+$-constraints.
  The constraints implied on the right endpoints of the intervals on $C$
  are not satisfiable since there is at least one $\s$-constraint on the cycle,
  hence $C$ is inconsistent.
\end{proof}
\fi

\iflong
\subsection{Intractable Fragments Containing $\d$, $\o$ or $\p$} 
\fi
\label{ssec:do-dp-op-hard}

\iflong
The following three theorems, which show \W{1}-hardness of
$\mincsp{\d, \p}$, $\mincsp{\d, \o}$, and $\mincsp{\p, \o}$,
are all based on fairly similar parameterized reductions from \simdfas{} (which is a
\W{1}-hard problem by Theorem~\ref{thm:sim-dfas-hard}). 
\fi
\ifshort
We finally show the \W{1}-hardness of
$\mincsp{\d, \p}$, $\mincsp{\d, \o}$, and $\mincsp{\p, \o}$
via parameterized reductions from \simdfas{} (which is a
\W{1}-hard problem by Theorem~\ref{thm:sim-dfas-hard}). 
\fi
\iflong
In all three
cases, we reduce an instance $\III=(D_1,D_2,k)$
with $V=V(D_1)=V(D_2)$ of \simdfas{} to an instance of the respective
fragment of \textsc{MinCSP}.
Our reductions preserve the parameter $k$
exactly and they introduce one
variable $v$ for every vertex in $V$ together with additional variables and
constraints for every arc $(u,v) \in A_1\cup A_2$.
\fi

\begin{theorem}\label{thm:pdWh}
\iflong
  $\mincsp{\d, \p}$ is \textnormal{\W{1}}-hard.
\fi
\ifshort
$\mincsp{\d, \p}$,  $\mincsp{\d, \o}$, and $\mincsp{\p, \o}$ are \textnormal{\W{1}}-hard.
\fi
\end{theorem}
\begin{proof}
\iflong
    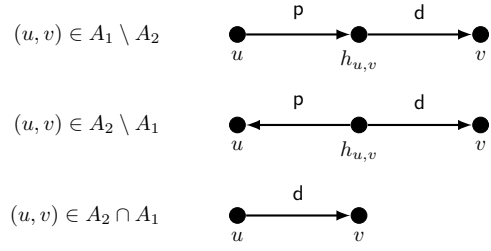
\begin{figure}
    \centering
      \begin{tikzpicture}[%
      scale=.8,
      transform shape,
      myvertex/.style={circ,scale=2}
      ]

      \draw
      node[myvertex,label=below:{$u$}] (u) {}
      node[myvertex, right of=u, label=below:{$h_{u,v}$}] (h1) {}
      node[myvertex, right of=h1, label=below:{$v$}] (v) {}
      ;
      
      \draw[thick,-latex] (u) -- node[midway, label=above:$\p$] () {} (h1);
      \draw[thick,-latex] (h1) -- node[midway, label=above:$\d$] () {}
      (v);
      \draw
      (u.west) +(-1cm, 0cm) node[anchor=east] (c1) {$(u,v) \in A_1\setminus A_2$}
      ;
      
      \draw
      (u) +(0cm,-1.5cm)  node[myvertex,label=below:{$u$}] (u1) {}
      node[myvertex, right of=u1, label=below:{$h_{u,v}$}] (h11) {}
      node[myvertex, right of=h11, label=below:{$v$}] (v1) {}
      ;
      
      \draw[thick,-latex] (h11) -- node[midway, label=above:$\p$] () {} (u1);
      \draw[thick,-latex] (h11) -- node[midway, label=above:$\d$] () {}
      (v1);
      \draw
      (u1.west) +(-1cm, 0cm) node[anchor=east] (c2) {$(u,v) \in A_2\setminus A_1$}
      ;
      
      \draw
      (u1) +(0cm,-1.5cm)  node[myvertex,label=below:{$u$}] (u2) {}
      node[myvertex, right of=u2, label=below:{$v$}] (v2) {}
      ;
      
      \draw[thick,-latex] (u2) -- node[midway, label=above:$\d$] () {} (v2);
      \draw
      (u2.west) +(-1cm, 0cm) node[anchor=east] (c2) {$(u,v) \in A_2\cap A_1$}
      ;

    \end{tikzpicture}
    \caption{The variables and constraints introduced in the
      construction of the proof of Theorem~\ref{thm:pdWh} for every
      arc $(u,v)$ of the \simdfas{} instance.}
    \label{fig:w1-pd}
  \end{figure}
\fi
  \ifshort
 We present a proof outline for  $\mincsp{\d, \p}$; the other two cases are similar.
  \fi
  Let $\III=(D_1,D_2,k)$
  with $V=V(D_1)=V(D_2)$ be an instance of \simdfas{}.
  We construct an instance $\III'$ of $\mincsp{\d,
    \p}$ with the same parameter $k$ and where we have one variable $v$ for
  every vertex in $V$. Additionally, we introduce the following variables and
  constraints.
   For every arc $(u,v) \in (A_1\cup
  A_2)\setminus (A_1\cap A_2)$, we add a fresh variable $h_{u,v}$.
  Moreover, for every arc $(u,v) \in
  A_1\cup A_2$, we distinguish the following cases, which are illustrated in Figure\iflong~\ref{fig:w1-pd}\fi\ifshort~\ref{fig:gadget-short}\fi:
  \begin{itemize}
  \item If $(u,v) \in A_1\setminus A_2$, then we add the constraints
    $u \p h_{u,v}$ and $h_{u,v} \d  v$.
  \item If $(u,v) \in A_2\setminus A_1$, then we add the constraints:
    $h_{u,v} \p u$ and $h_{u,v} \d  v$.
  \item If $(u,v) \in A_1\cap A_2$, then we add the constraint
    $u \d v$.

  \end{itemize}
  This completes the construction of $\III'$, which is clearly a
  polynomial-time reduction. \iflong It remains to show that $\III$ and $\III'$
  are equivalent. \fi Bad cycles in this case
  are cycles with $\p$-arcs in the same direction and
  no $\d$-arcs meeting head-to-head by  
  Lemma~\ref{lem:bad-cycles}.\ifshort The proof of the equivalence of
  $\III$ and $\III'$ is now relatively straightforward after observing
  that every directed cycle in $D_i$ gives rise to a bad
  cycle in $\III'$ but this is not the case for cycles that use at
  least one arc from $A_1\setminus A_2$ and at least one arc from
  $A_2\setminus A_1$.
  \fi
\iflong  
\smallskip

\noindent
  {\bf Forward direction.}  Let $X \subseteq A_1\cup A_2$
  be a solution for $\III$. We claim that any set $X'$ that for every arc in $X'$ contains an
  arbitrary constraint introduced for the arc is a solution
  for $\III'$. Suppose for a contradiction that this is not the case and
  let $C'$ be a bad cycle of $\III'-X'$ witnessing this.
  Let $C$ be the sequence of arcs corresponding to
  $C'$ in $(V,A_1\cup A_2)$. If $C$ is not a directed cycle, then it
  contains 2 consecutive arcs that are head-to-head and therefore
  also $C'$ contains 2 consequitive $d$-arcs that are head-to
  head, contradicting our assumption that $C'$ is a bad cycle.
  Therefore, $C$ is directed and contains an arc in
  $A_1\setminus A_2$ and an arc in $A_2\setminus A_1$. However, then
  $C'$ is not a bad cycle because it contains two $p$-arcs in opposite directions.
  Therefore, $C$ is a directed cycle in $D_1$ or in $D_2$ and this
  contradicts our assumption that $X$ is a solution for $\III$. 

  \smallskip

 \noindent
  {\bf Backward direction.}  Let $X'$ be a solution for
  $\III'$ and let $X\subseteq A_1\cup A_2$ be the set obtained from $X'$
  by taking every arc such that $X'$ contains a corresponding
  constraint. We claim that $X$ is a solution for $\III$. Suppose for a
  contradiction that this is not the case and let $C$ be a cycle in
  $D_i-X$ witnessing this. Then, the corresponding cycle $C'$ in $\III'$ must be a bad
  cycle and this contradicts our assumption that $X'$ is a solution for $\III$.
  \fi
\end{proof}
\ifshort
\fi
\iflong
\begin{theorem}\label{thm:doWh}
  $\mincsp{\d, \o}$ is \textnormal{\W{1}}-hard.
\end{theorem}
\begin{proof}
    \begin{figure}
    \centering
    \begin{tikzpicture}[%
      scale=.8,
      transform shape,
      myvertex/.style={circ,scale=2}
      ]

      \draw
      node[myvertex,label=below:{$u$}] (u) {}
      node[myvertex, right of=u, label=below:{$h_{u,v}$}] (h1) {}
      node[myvertex, right of=h1, label=below:{$v$}] (v) {}
      ;
      
      \draw[thick,-latex] (u) -- node[midway, label=above:$\d$] () {} (h1);
      \draw[thick,-latex] (v) -- node[midway, label=above:$\o$] () {} (h1);
      \draw
      (u.west) +(-1cm, 0cm) node[anchor=east] (c1) {$(u,v) \in A_1\setminus A_2$}
      ;
      
      \draw
      (u) +(0cm,-1.5cm)  node[myvertex,label=below:{$u$}] (u1) {}
      node[myvertex, right of=u1, label=below:{$h_{u,v}$}] (h11) {}
      node[myvertex, right of=h11, label=below:{$v$}] (v1) {}
      ;
      
      \draw[thick,-latex] (u1) -- node[midway, label=above:$\d$] () {} (h11);
      \draw[thick,-latex] (h11) -- node[midway, label=above:$\o$] () {} (v1);
      \draw
      (u1.west) +(-1cm, 0cm) node[anchor=east] (c2) {$(u,v) \in A_2\setminus A_1$}
      ;
      
      \draw
      (u1) +(0cm,-1.5cm)  node[myvertex,label=below:{$u$}] (u2) {}
      node[myvertex, right of=u2, label=below:{$v$}] (v2) {}
      ;
      
      \draw[thick,-latex] (u2) -- node[midway, label=above:$\d$] () {} (v2);
      \draw
      (u2.west) +(-1cm, 0cm) node[anchor=east] (c2) {$(u,v) \in A_2\cap A_1$}
      ;

    \end{tikzpicture}
    \caption{The variables and constraints introduced in the
      construction of the proof of Theorem~\ref{thm:doWh} for every
      arc $(u,v)$ of the \simdfas{} instance.}
    \label{fig:w1-do}
  \end{figure}
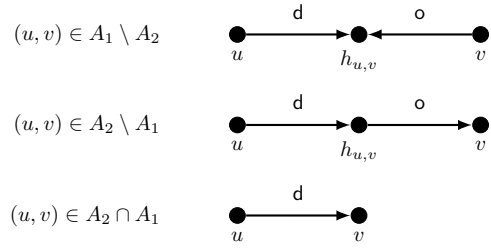

Let $\III=(D_1,D_2,k)$
with $V=V(D_1)=V(D_2)$ be an instance of \simdfas{}.
  We construct an instance $\III'$ of $\mincsp{\d,\o}$ with the same parameter $k$ and where we have one variable $v$ for
  every vertex in $V$. Additionally, we introduce the following variables and
  constraints.
 For every arc $(u,v) \in (A_1\cup
  A_2)\setminus (A_1\cap A_2)$, we add the variable $h_{u,v}$.
  Moreover, for
  every arc $(u,v) \in A_1\cup A_2$, we distinguish the following
  cases, which are illustrated in Figure~\ref{fig:w1-do}:
  \begin{itemize}
  \item If $(u,v) \in A_1\setminus A_2$, we add the constraints
    $u \d h_{u,v}$ and $v \o h_{u,v}$.
  \item If $(u,v) \in A_2\setminus A_1$, we add the constraints
    $u \d h_{u,v}$ and $h_{u,v} \o v$.
  \item If $(u,v) \in A_1\cap A_2$, we add the constraint
    $u \d v$.
  \end{itemize}
  This completes the construction of $\III'$, which is clearly a
  polynomial-time reduction. It remains to show that $\III$ and $\III'$
  are equivalent. Bad cycles here are
  cycles with all $\d$-arcs in the same direction and 
  all $\o$-arcs in the same direction 
  (the directions of a $\d$-arc and an $\o$-arc may differ) by Lemma~\ref{lem:bad-cycles}.

\smallskip

\noindent
  {\bf Forward direction.}  Let $X \subseteq A_1\cup A_2$
  be a solution for $\III$. We claim that any set $X'$ that for every arc in $X$ contains
  an
  arbitrary constraint introduced for the arc is a solution
  for $\III'$. Suppose for a contradiction that this is not the case and
  let $C'$ be a bad cycle of $\III'-X'$ witnessing this.
  Let $C$ be the sequence of arcs corresponding to
  $C'$ in $(V,A_1\cup A_2)$. If $C$ is not a directed cycle, then it
  contains 2 arcs in opposite directions, which implies that $C'$ contains 2 $\d$-arcs
  in opposite directions and therefore contradicts our assumption that $C'$ is
  a bad cycle.
  Therefore, $C$ is directed. Suppose now that $C$ contains an arc in
  $A_1\setminus A_2$ and an arc in $A_2\setminus A_1$. However, then
  $C'$ is not a bad cycle because it contains two $\o$-arcs in opposite directions.
  Therefore, $C$ is a directed cycle in $D_1$ or in $D_2$,
  contradicting our assumption that $X$ is a solution for $\III$. 

  \smallskip

 \noindent
  {\bf Backward direction.} Let $X'$ be a solution for
  $\III'$ and let $X\subseteq A_1\cup A_2$ be the set obtained from $X'$
  by taking every arc such that $X'$ contains a corresponding
  constraint. We claim that $X$ is a solution for $\III$. Suppose for a
  contradiction that this is not the case and let $C$ be a cycle in
  $D_i-X$ witnessing this. Then, the corresponding cycle $C'$ in $\III'$ must be a bad
  cycle, a contradiction to our assumption that $X'$ is a solution for
  $\III$.
\end{proof}

\begin{theorem}\label{thm:poWh}
  $\mincsp{\p, \o}$ is 
  \textnormal{\W{1}}-hard.
\end{theorem}
\begin{proof}

  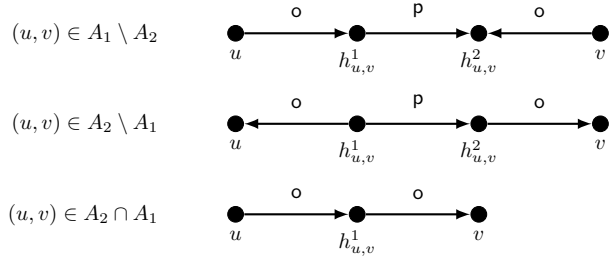
\begin{figure}
    \centering
    \begin{tikzpicture}[%
      scale=.8,
      transform shape,
      myvertex/.style={circ,scale=2}
      ]

      \draw
      node[myvertex,label=below:{$u$}] (u) {}
      node[myvertex, right of=u, label=below:{$h_{u,v}^1$}] (h1) {}
      node[myvertex, right of=h1, label=below:{$h_{u,v}^2$}] (h2) {}
      node[myvertex, right of=h2, label=below:{$v$}] (v) {}
      ;
      
      \draw[thick,-latex] (u) -- node[midway, label=above:$\o$] () {} (h1);
      \draw[thick,-latex] (h1) -- node[midway, label=above:$\p$] () {} (h2);
      \draw[thick,-latex] (v) -- node[midway, label=above:$\o$] () {} (h2);
      \draw
      (u.west) +(-1cm, 0cm) node[anchor=east] (c1) {$(u,v) \in A_1\setminus A_2$}
      ;
      
      \draw
      (u) +(0cm,-1.5cm)  node[myvertex,label=below:{$u$}] (u1) {}
      node[myvertex, right of=u1, label=below:{$h_{u,v}^1$}] (h11) {}
      node[myvertex, right of=h11, label=below:{$h_{u,v}^2$}] (h12) {}
      node[myvertex, right of=h12, label=below:{$v$}] (v1) {}
      ;
      
      \draw[thick,-latex] (h11) -- node[midway, label=above:$\o$] () {} (u1);
      \draw[thick,-latex] (h11) -- node[midway, label=above:$\p$] () {} (h12);
      \draw[thick,-latex] (h12) -- node[midway, label=above:$\o$] () {} (v1);
      \draw
      (u1.west) +(-1cm, 0cm) node[anchor=east] (c2) {$(u,v) \in A_2\setminus A_1$}
      ;
      
      \draw
      (u1) +(0cm,-1.5cm)  node[myvertex,label=below:{$u$}] (u2) {}
      node[myvertex, right of=u2, label=below:{$h_{u,v}^1$}] (h21) {}
      node[myvertex, right of=h21, label=below:{$v$}] (v2) {}
      ;
      
      \draw[thick,-latex] (u2) -- node[midway, label=above:$\o$] () {} (h21);
      \draw[thick,-latex] (h21) -- node[midway, label=above:$\o$] () {} (v2);
      \draw
      (u2.west) +(-1cm, 0cm) node[anchor=east] (c2) {$(u,v) \in A_2\cap A_1$}
      ;

    \end{tikzpicture}
    \caption{The variables and constraints introduced in the
      construction of the proof of Theorem~\ref{thm:poWh} for every
      arc $(u,v)$ of the \simdfas{} instance.}
    \label{fig:w1-po}
  \end{figure}

Let $\III=(D_1,D_2,k)$ 
with $V=V(D_1)=V(D_2)$ be an instance of \simdfas{}.
  We construct an instance $\III'$ of $\mincsp{\p,
    \o}$ with the same parameter $k$ and where we have one variable $v$ for
  every vertex in $V$. Additionally, we introduce the following variables and
  constraints. For every arc $(u,v) \in (A_1\cup A_2)\setminus (A_1\cap
  A_2)$, we add two variables $h_{u,v}^1$ and
  $h_{u,v}^2$. Finally, for every arc $(u,v) \in A_1\cup A_2$, we
  distinguish the following cases, which are illustrated in Figure~\ref{fig:w1-po}:
  \begin{itemize}
  \item If $(u,v) \in A_1\setminus A_2$, we add the constraints
    $u \o h_{u,v}^1$, $h_{u,v}^1 \p h_{u,v}^2$, and $v \o
    h_{u,v}^2$.
  \item If $(u,v) \in A_2\setminus A_1$, we add the constraints
    $h_{u,v}^1 \o u$, $h_{u,v}^1 \p h_{u,v}^2$, and $h_{u,v}^2 \o
    v$.
  \item If $(u,v) \in A_1\cap A_2$, we add one fresh
    variable $h_{u,v}$ and the constraints
    $u \o h_{u,v}$ and $h_{u,v} \o v$.    
  \end{itemize}
  This completes the construction of $\III'$ and this is clearly a
  polynomial-time reduction. It remains to show that $\III$ and $\III'$
  are equivalent. Before the proof, we recall the definition of bad cycles for the case of
  $\mincsp{\p, \o}$ from Lemma~\ref{lem:bad-cycles}: bad cycles are (1)
  directed cycles of $\o$-arcs and (2) 
  cycles with all $\p$-arcs in the same (forward) direction with
  every consecutive pair of $\o$-arcs in the reverse direction
  separated by a $\p$-arc. Note that (2) implies
  that directed cycles of $\p$-arcs are bad.
    
\smallskip

   \noindent
  {\bf Forward direction.} 
  Let $X \subseteq A_1\cup A_2$
  be a solution for $\III$. We claim that every set $X'$ that for every arc $a$ in $X$ contains any
  of the constraints introduced for $a$ is a solution
  for $\III'$. Suppose to the contrary that this is not the case and
  let $C'$ be a bad cycle in $\III'-X'$ witnessing this.
  Let $C$ be the sequence of arcs corresponding to
  $C'$ in $(V,A_1\cup A_2)$. We first show that $C$ is a directed cycle.
  Otherwise, $C$
  contains 2 arcs $a$ and $a'$ that are pointing in opposite directions. If
  $a\in A_1\setminus A_2$ and $a' \in A_2\setminus A_1$ (or vice
  versa), then $C'$ contains two $p$-arcs in opposite directions,
  contradicting our assumption that $C'$ is a bad cycle. Otherwise, $a
  \in A_1\cap A_2$ and $a' \in A_1\cup A_2$, but then $C'$ contains
  two consecutive $o$-arcs in reverse direction;
  this contradicts once again
  our assumption that $C'$ is bad.
  This shows that $C$ is a directed cycle.
  Therefore, if $C$ is also in $D_1$ or in $D_2$, we
  obtain a contradiction to our assumption that $X$ is a solution for
  $\III$.
  Otherwise, $C$ uses at least one arc from
  $A_1\setminus A_2$ and at least one arc from $A_2\setminus A_1$.
  Hence, $C$ contains an arc $(u_1,v_1)\in A_1\setminus A_2$ and an
  arc $(u_2,v_2)\in A_2\setminus A_1$ such that $C$ contains only arcs from
  $A_1\cap A_2$ in the path from $a_1$ to $a_2$ in $C$. But then,
  $C'$ contains the constraints $v_1 \o h_{u_1,v_1}^2$ and
  $h_{u_2,v_2}^1 \o u_2$. Moreover, between these two constraints,
  $C'$ contains only forward $o$ constraints, which contradicts our
  assumption that $C'$ is a bad cycle in $\III'$.

\smallskip

\noindent
  {\bf Backward direction.} Let $X'$ be a solution for
  $\III'$ and let $X\subseteq A_1\cup A_2$ be the set obtained from $X'$
  by taking every arc such that $X'$ contains a corresponding
  constraint. We claim that $X$ is a solution for $\III$. Suppose for a
  contradiction that this is not the case and let $C$ be a cycle in
  $D_i-X$ witnessing this. Then, the corresponding cycle $C'$ in $\III'$ is a bad
  cycle contradicting our assumption that $X'$ is a solution for $\III'$.
\end{proof}
\fi

\iflong
\section{Approximation} \label{sec:approx}

Our previous results imply that \mincsp{\A} is W[1]-hard.
We will now take a brief look on the approximability of this problem, and show that
\mincsp{\rr} is not approximable in polynomial time within any constant when $\rr \in \A \setminus \{\equiv\}$,
but approximable in fpt time within a factor of~$2$.
We refer the reader to the textbook by Ausiello et. al.~\cite{ausiello2012complexity} 
for more information about polynomial-time approximability, and 
to the surveys by Marx~\cite{Marx:tcj2008} and Feldmann et al.~\cite{Feldmann:etal:algorithms2020}
for an introduction to parameterized approximability.
We will exclusively consider minimisation problems in the sequel.
Formally speaking, a {\em minimisation problem} is a 3-tuple $(X, sol, cost)$ where

\begin{enumerate}
\item
$X$ is the set of instances. 

\item
For an instance $x \in X$, $sol(x)$ is the set of feasible solutions
for $x$, the length of each $y \in sol(x)$ is polynomially bounded in $\norm{x}$, and it is decidable in
polynomial time (in $\norm{x}$) whether $y \in sol(x)$ holds for given $x$ and $y$. 

\item
For an instance
$x \in X$ with feasible solution $y$, $cost(x, y)$ is a polynomial-time computable positive integer. 
\end{enumerate}

The objective of the minimisation problem $(X, sol, cost)$ is to find an solution 
$z$ of minimum cost for a given instance $x \in X$ 
i.e. a solution $z$ with $cost(x, z) = opt(x)$ where $opt(x) = \min \{cost(x, y) \mid y \in sol(x)\}$.
If $y$ is a solution for the instance $x$, then the performance ratio of $y$ is defined as 
$cost(x, y)/opt(x)$. For $c \geq 1$, we say that an algorithm is a factor-$c$ 
{\em approximation} algorithm if it always computes a solution with performance ratio at most $c$
i.e. the cost of solutions returned by the algorithm is at most $c$ times greater than the minimum cost.
One often replaces the constant $c$ with a function of instance $x$
but this is not needed in this paper.

The starting point of the polynomial-time hardness-of-approximation 
results is the following consequence
of the Unique Games Conjecture (UGC) of Khot~\cite{khot2002power}.

\begin{theorem}[Corollary~1.2~in~\cite{Guruswami:etal:sicomp2011}]
  Assuming UGC, it is \NP-hard to $c$-approximate
  \textsc{Directed Feedback Arc Set} for any $c > 1$.
\end{theorem}

By Lemma~\ref{lem:bad-cycles},
an instance $I$ of $\csp{\rr}$ for $\rr \in \A \setminus \{\e, \m\}$ is consistent if and only if the arc-labelled graph $G_I$ has no directed cycles.
Hence, $\mincsp{\rr}$ is equivalent to DFAS in these cases.
By Lemma~\ref{lem:impl-reduction}, the relation $\m$ implements $\p$ so there is a polynomial-time reduction from $\mincsp{\p}$
to $\mincsp{\m}$ that preserves approximation factors.
Thus, we conclude the following.

\begin{observation} \label{thm:poly-apx-hard}
  Let $\Gamma \subseteq \A$ be interval constraint language.
  If $\Gamma \neq \{\e\}$, then,
  assuming UGC, it is \NP-hard to $c$-approximate
  \mincsp{\Gamma} for any $c > 1$.
\end{observation}

We continue with
fpt-approximability of minimization problems. 
For a minimization problem $(X, sol, cost)$ parameterized by the solution
size $k$, a factor-$c$ {\em fpt-approximation algorithm} 
is an algorithm that 
\begin{enumerate}
\item
takes a tuple $(x,k)$ as input where $x \in X$ is an instance and $k \in \naturals$,

\item
either returns that there is no solution of size at most $k$ or returns a solution of size at
most $c \cdot k$, and

\item
runs in time $f(k) \cdot \norm{x}^{O(1)}$ where $f: \naturals \rightarrow \naturals$ 
is some computable function.
\end{enumerate}

\begin{theorem}[\cite{lokshtanov2021fpt}] \label{thm:sdfas-approx}
  \subdfas is $2$-approximable in $O^*(2^{O(k)})$ time. 
\end{theorem}

\begin{theorem}
  \mincsp{\A} is $2$-approximable in $O^*(2^{O(k^3)})$ time
  and $4$-approximable in $O^*(2^{O(k)})$ time.
\end{theorem}
\begin{proof}
Let $(\III, k)$ be an instance of \mincsp{\A}.
Replace every constraint $x\{\o\}y$ by
its implementation in $\{\s,\f\}$
according to Lemma~\ref{lem:implementations}.
By Lemma~\ref{lem:impl-reduction}, this does not change
the cost of the instance.
Using Table~\ref{tb:allen}, we can rewrite all constraints of $\III'$
as conjunctions of two atomic constraints of the form $x < y$ and $x = y$.
Let $S$ be the set of all atomic constraints.
Note that we can view $S$ as an instance of $\mincsp{<,=}$.
The algorithm is then to check whether $(S, 2k)$ is 
a yes-instance of $\mincsp{<,=}$ by casting it as an instance
of \subdfas according to Lemma~\ref{lem:less-eq=sdfas}.
For correctness, observe that deleting $k$ constraints in $\III'$
corresponds to deleting at most $2k$ constraints in $S$,
hence if $(\III',k)$ is a yes-instance, then so is $(S, 2k)$.
On the other hand, if deleting any subset of $2k$ constraints 
from $S$ does not make it consistent, then deleting any subset
of $k$ constraints from $\III'$ cannot make it consistent either.

The parameter $k$ is unchanged so 
the time complexity follows from 
Theorems~\ref{thm:sdfas-exact}~and~\ref{thm:sdfas-approx}, respectively.
\end{proof}

\textsc{Weighted MinCSP} is a generalization of \textsc{MinCSP}
where every constraint in the instance $(\III, k)$ is assigned 
a positive integer weight, and we also have a weight budget $W$.
The goal is now to find a set $X$ of at most $k$ constraints
with total weight at most $W$ such that $I - X$ is consistent.
The parameter is still $k$ -- the size of deletion set.
This strictly generalizes \textsc{MinCSP} by
assigning unit weights to all constraints and setting the
weight budget $W = k$.
\textsc{Weighted Subset Directed Feedback Arc Set} has been
shown to be in \FPT~\cite{kim2022weighted} using 
recently introduced directed flow augmentation technique.
While the running time is not explicitly stated in the corresponding theorem of~\cite{kim2022weighted},
it is analysed in Section~5 of~\cite{kim2022weighted}.
\begin{theorem}[\cite{kim2022weighted}]
\textsc{Weighted Subset Directed Feedback Arc Set} can be solved in time
$O^*(2^{O(k^8 \log k)})$.
\end{theorem}

The reduction in Lemma~\ref{lem:less-eq=sdfas} readily
implies that \textsc{Weighted \mincsp{<,=}} is in \FPT and we obtain the following.

\begin{observation}
  \textsc{Weighted \mincsp{\A}} is $2$-approximable in time $O^*(2^{O(k^8 \log k)})$.
\end{observation}

\fi

\section{Discussion}
\label{sec:discussion}

We have initiated a study of the parameterized complexity of \textsc{MinCSP} for 
Allen's interval algebra.
We prove that \textsc{MinCSP} restricted to the relations in $\A$ exhibits a dichotomy:
\textsc{MinCSP}$(\Gamma)$ is either fixed-parameter tractable or \W{1}-hard when
$\Gamma \subseteq \A$.
Even though the restriction to the relations in $\A$ may seem severe, one should keep in mind
that a CSP instance over $\A$ is sufficient for representing {\em definite} information
about the relative positions of intervals. In other words, such an instance can be viewed as a data set of interval
information and the \textsc{MinCSP} problem can be viewed as a way of filtering out
erroneous information (that may be the result of
contradictory sources of information, noise in the
measurements, human mistakes etc.) Various ways of ``repairing'' inconsistent data sets of qualitative information have been thoroughly
discussed by many authors;
see, for instance, \cite{bertossi2004query,chomicki2005minimal,condotta2016sat}
and the
references therein.

Proving a full parameterized complexity classification for Allen's interval algebra
is hindered by a barrier: such a classification would settle the parameterized
complexity of \textsc{Directed Symmetric Multicut}, and this problem
is considered to be one of
the main open problems in the area of directed graph separation
problems~\cite{EibenRW22,kim2022weighted}.
This barrier comes into play even in very restricted cases: as an example, it is not difficult
to see that \textsc{MinCSP} for the two Allen relations
$(\f \: \cup \: \fii)$
and $(\f \: \cup \: \e)$ is equivalent to the \textsc{MinCSP} problem for the
two PA relations $\neq$ and $\leq$ and thus equivalent to 
\textsc{Directed Symmetric Multicut}.

One way of continuing this work without necessarily settling the parameterized
complexity of \textsc{Directed Symmetric Multicut} is to consider fpt approximability: it is known that
\textsc{Directed Symmetric Multicut} is $2$-approximable in fpt time~\cite{EibenRW22}.
Thus, a possible research direction is to analyse the fpt approximability for {\sc MinCSP}$(\Gamma)$
when $\Gamma$ is a subset of $2^{\A}$ or, more ambitiously, when $\Gamma$ is first-order definable in $\A$.
A classification
that separates the cases that are constant-factor fpt approximable from those that are not
may very well be easier to obtain than mapping the \FPT/\W{1} borderline.
There is at least one technical reason for optimism here, and we introduce some definitions to outline
this idea.
\iflong
An $n$-ary relation $R$
is said
to have a {\em primitive positive definition} (pp-definition) in
a structure $\Gamma$ if 

 \[R(x_1, \ldots, x_{n}) \equiv \exists x_{n+1}, \ldots,
x_{n+n'} \colon R_1(\mathbf{x}_1) \wedge \ldots \wedge
R_m({\mathbf{x}_m}),\]
where each $R_i \in \Gamma \cup \{=\}$ and
each $\mathbf{x}_i$ is a tuple of variables over $x_1,\ldots, x_{n}$, $x_{n+1}, \ldots, x_{n+ n'}$ with the same length as the arity of $R_i$.
If, in addition, each $R_i
\in \Gamma$ then we say that $R$ has an {\em equality-free primitive
  positive definition} (efpp-definition) in $\Gamma$.
\fi
\ifshort
An $n$-ary relation $R$
is said
to have a {\em primitive positive definition} (pp-definition) in
a structure $\Gamma$ if it can be first-order defined by only using the relations
in $\Gamma$ together with the equality relation and the operators
existential quantification and conjunction.
If the equality relation is not needed then we say that $R$ has an {\em equality-free primitive
  positive definition} (efpp-definition) in $\Gamma$.
  \fi
Bonnet et al.~\cite[Lemma 10]{Bonnet:etal:esa2016} have shown that constant-factor fpt approximability is preserved by efpp-definitions~\cite{Bonnet:etal:esa2016}, i.e.
if $R$ is efpp-definable in $\Gamma$ and \textsc{MinCSP}$(\Gamma)$ is constant-factor fpt approximable,
then \textsc{MinCSP}$(\Gamma \cup \{R\})$ is also constant-factor fpt approximable.
Bonnet et al. focus on Boolean domains, but it is clear
that their Lemma~10 works for problems with arbitrarily large domains.
Lagerkvist~\cite[Lemma~5]{Lagerkvist:ismvl2020} has shown that in most cases one can use pp-definitions instead of efpp-definitions. This
implies that the standard algebraic approach via polymorphisms (that for instance underlies the full complexity classification of finite-domain CSPs~\cite{Bulatov:focs2017,Zhuk:jacm2020}) often becomes applicable when analysing constant-factor
fpt approximability. One should note that, on the other, the exact complexity of {\sc MinCSP} is only preserved
by much more limited constructions such as {\em proportional implementations}~(see Section~5.2.~in~\cite{kim2021solving}).
We know from the literature that this may be an important difference: it took several years after Bonnet et al.'s
classification of approximability before
the full classification of exact parameterized complexity
was obtained using a much more complex framework~\cite{Kim:etal:FA3}.
It is also worth noting that parameterized approximation results for \textsc{MinCSP}
may have very interesting consequences, e.g.
\cite{Bonnet:etal:esa2016} resolved the parameterized complexity of 
\textsc{Even Set}, which was a long-standing open problem.

\bibliographystyle{plain}
\bibliography{references}

\end{document}